%% file: main.tex
\def\llncs{0}											
\def\LNCSpreview{0}								
\def\pagelimit{}									 
\def\anonymous{0}								 
\def\acknowledgments{0}						 
\def\overflow{0}									 
\def\showlabels{0}								   
\def\authnotes{0}									
\def\dieordollar{0}									
\def\notxfont{0}
\def\stuffedtitlepage{0}						 
\def\abbrevref{0}									
\def\dropargs{0}									
\def\rmvtheoremspace{0}						 
\def\allowbreaks{0}								  
\def\choosebibstyle{1}							
\def\changefont{0}								  
\def\submission{0}
\def\cameraready{0}
\def\noaux{1}

\ifnum\submission=1
\def\llncs{1}
\def\anonymous{1}
\def\authnotes{1}
\def\choosebibstyle{0}
\def\noaux{0}
\else
\fi

\ifnum\cameraready=1
\def\submission{1}
\def\llncs{1}
\def\authnotes{0}
\def\choosebibstyle{0}
\def\acknowledgments{1}	
\else
\fi

\ifnum\anonymous=1
\def\authnotes{0}
\else
\fi

\ifnum\authnotes=0  
\newcommand{\fuyuki}[1]{}
\newcommand{\ryo}[1]{}
\newcommand{\ECrev}[1]{}
\else
\newcommand{\fuyuki}[1]{$\ll$\textsf{\color{blue} Fuyuki: { #1}}$\gg$}
\newcommand{\ryo}[1]{$\ll$\textsf{\color{red} Ryo: { #1}}$\gg$}
\newcommand{\ECrev}[1]{$\ll${\color{darkgreen} EC'22: { #1}}$\gg$}
\fi

\def\confvers{0}										   

\ifnum\llncs=1
\def\titletext{											
Watermarking PRFs \\ against Quantum Adversaries
}
\def\runningtitle{									
}
\else
\def\titletext{											
Watermarking PRFs against Quantum Adversaries
}
\def\runningtitle{									
}
\fi

\date{}										  

\def\choosepubinfo{5}						   

\def\pubinfoYEAR{2016}								
\def\pubinfoSUBMISSIONDATE{}		  
\def\pubinfoDOI{ }								 
\def\pubinfoBIBDATA{ }						 
\def\pubinfoCONFERENCE{CRYPTO}				

\ifnum\submission=1
\ifnum\noaux=1
\def\pubinfoindividual{}
\else
\def\pubinfoindividual{{\color{red}{\textbf{We attached the full version of this paper as a supplementary material}}}.}
\fi
\else
\def\pubinfoindividual{}
\fi

\makeatletter
\expandafter\newcommand{\createauthor}[5]{%
	\@namedef{#1name}{#2}%
	\@namedef{#1running}{#3}%
	\@namedef{#1institute}{#4}%
	\@namedef{#1thanks}{#5}%
}

\expandafter\newcommand{\createinstitute}[4]{%
	\@namedef{#1instname}{#2}%
	\@namedef{#1mail}{#3}%
	\@namedef{#1number}{#4}%
}
\makeatother

\newcounter{authorcount}
\newcommand{\newauthor}[4]{
	\stepcounter{authorcount}
	\createauthor{\theauthorcount}{#1}{#2}{#3}{#4}
}

\newcounter{institutecount}
\newcommand{\newinstitute}[3]{
	\stepcounter{institutecount}
	\createinstitute{\theinstitutecount}{#1}{#2}{#3}
}

\newauthor{Fuyuki~Kitagawa}{F.~Kitagawa}{1}{}
\newauthor{Ryo~Nishimaki}{R.~Nishimaki}{1}{}

\newinstitute{NTT Corporation, Tokyo, Japan}{\{fuyuki.kitagawa.yh,ryo.nishimaki.zk\}@hco.ntt.co.jp}{1}

\def\contactmail{}

\def\keywords{
watermarking, pseudorandom function, post-quantum cryptography
}

\def\acknowledgmenttext{
}

\def\authorsofpdf{
\ifcsname 5name\endcsname
	\csname 1name\endcsname~et~al.
\else
	\ifcsname 1name\endcsname
		\csname 1name\endcsname
		\ifcsname 3name\endcsname
			,\ \csname 2name\endcsname
			\ifcsname 4name\endcsname
				,\ \csname 3name\endcsname\ and \csname 4name\endcsname
			\else%
				\ and\ \csname 3name\endcsname
			\fi
		\else
			\ifcsname 2name\endcsname
				\csname 2name\endcsname
			\fi
		\fi
	\fi
\fi
}

\let\accentvec\vec

\documentclass[a4paper,11pt]{article}
		\usepackage{fullpage}

\input{additional}
\input{stydef}

\usepackage{lmodern}
\usepackage[T1]{fontenc}
\usepackage[utf8]{inputenc}
\usepackage{etex}
\usepackage{graphicx,pstricks}
\usepackage{tikz}
\usetikzlibrary{matrix,arrows}
\usetikzlibrary{positioning}
\usetikzlibrary{decorations,decorations.text}
\usetikzlibrary{decorations.pathmorphing}
\usetikzlibrary{shapes}
\usepackage{centernot}
\usepackage{xspace}
\usepackage{pifont}
\usepackage{paralist}
\usepackage{booktabs}
\usepackage{dashbox}
\usepackage{multirow}
\usepackage{lscape}
\usepackage{amsmath}
\usepackage{fancybox}
\usepackage{braket}
\usepackage{physics}
\usepackage{autonum}
\usepackage{dsfont}
\DeclareMathAlphabet{\mathpzc}{OT1}{pzc}{m}{it}

\ifnum\submission=1
\renewcommand*{\backref}[1]{}
\def\notxfont{1}
\else
\fi

\ifnum\llncs=1
\renewcommand{\subparagraph}{\paragraph}
\else
\ifnum\notxfont=1
\else
\usepackage{mathpazo}
\usepackage{newtxtext}
\usepackage{helvet}
\fi
\fi
\ifnum\showlabels=1
\usepackage{showkeys}
\else
\fi

\usepackage{enumitem}
\ifnum\submission=1
\setlist[itemize]{
  topsep=0.2\baselineskip,
  itemsep=0\baselineskip,
}
\setlist[description]{
  topsep=0.2\baselineskip,
  itemsep=0\baselineskip,
}
\setlist[enumerate]{
  topsep=0.2\baselineskip,
  itemsep=0\baselineskip,
}

\else
\setlist[itemize]{
  topsep=0.4\baselineskip,
  itemsep=0.1\baselineskip,
}
\setlist[description]{
  topsep=0.4\baselineskip,
  itemsep=0.1\baselineskip,
}
\setlist[enumerate]{
  topsep=0.4\baselineskip,
  itemsep=0.1\baselineskip,
}
\fi

\usepackage{fontawesome}

\input{macros-others}

\input{macros-Ryo}

\input{macros-Fuyuki}

\begin{document}
	\input{titlepage}
		\ifnum\confvers=1
\else
\ifnum\submission=1
\else
\newpage
  \setcounter{tocdepth}{2}      
  \setcounter{secnumdepth}{2}   
  \setcounter{page}{0}          
  \tableofcontents
  \thispagestyle{empty}
\clearpage
\fi
\fi
\input{section/intro}

\input{section/prelim}

\input{section/def_q_watermarking}

\input{section/def_extless_watermarking}
\input{section/wmprf-from-elwmprf}

\input{section/elwmprf_from_LWE_simple_no_mk}

\input{section/elwmprf_from_IO_no_mk}

\ifnum\submission=0
\input{section/altogether}

\else\fi

	\ifnum\anonymous=0
	\ifnum\acknowledgments=1
	\paragraph{\textbf{Acknowledgments.}}	
	\acknowledgmenttext
	\fi
	\fi
	\ifnum\llncs=1
	\bibliographystyle{extreme_alpha}
	\bibliography{abbrev3,crypto,siamcomp_jacm,other-bib}
	\else
	\ifnum\choosebibstyle=0
	\else
	\ifnum\choosebibstyle=1
	\bibliographystyle{alpha}
	\else
	\bibliographystyle{abbrv}
	\fi
	\bibliography{abbrev3,crypto,siamcomp_jacm,other-bib}
	\fi
	\fi

	
\ifnum\cameraready=0
	\ifnum\llncs=0
	\appendix
	
	\input{section/qelwmprf-from-elwmprf}
\input{section/pe_construction}
\else
	\newpage
	 	\appendix
	 	\setcounter{page}{1}
 	{
	\noindent
 	\begin{center}
	{\Large SUPPLEMENTAL MATERIALS}
	\end{center}
 	}
	\setcounter{tocdepth}{2}
	 	\ifnum\noaux=1
 	\else
{\color{red}{We attached the full version of this paper as a separated file (auxiliary supplemental material) for readability. It is available from the program committee members.}}
\fi

\input{section/related_work}\input{section/notation}\input{section/quantum_information}\input{section/crypto_tools}\input{section/proof_rev_almost_projective}\input{section/qelwmprf-from-elwmprf}\input{section/proofs_elwmprf_lwe}\input{section/pe_revisited}\input{section/proof_elwmprf_io}\input{section/altogether}\input{section/pe_construction}
	\setcounter{tocdepth}{1}
	\tableofcontents

	\fi
\fi
	
\end{document}

%% file: additional.tex
\let\vec\accentvec

\usepackage[dvipsnames]{xcolor}
\definecolor{darkblue}{rgb}{0,0,0.6}
\definecolor{darkgreen}{rgb}{0,0.5,0}
\definecolor{darkviolet}{RGB}{130,95,141}

\usepackage
[pdfpagelabels=true,
linktocpage=false,
bookmarks=true,bookmarksnumbered=false,bookmarkstype=toc,
pagebackref,
colorlinks=true,linkcolor=darkblue,urlcolor=darkblue,citecolor=darkgreen]
{hyperref}

\usepackage{amscd,amsmath,amssymb,amsfonts}

\usepackage{amsthm}

\ifnum\abbrevref=0
	\usepackage[capitalise,noabbrev]{cleveref}
\else
	\usepackage[capitalise]{cleveref}
\fi
\usepackage[absolute]{textpos}
\usepackage[final]{microtype}
\usepackage[absolute]{textpos}
\usepackage{everypage}

\ifnum\dieordollar=2
	\usepackage{tikz}
\else
	\ifnum\dieordollar=1
		\usepackage{tikz}
	\fi
\fi

\ifnum\LNCSpreview=1
	\usepackage[paperwidth=152mm,paperheight=235mm,textwidth=122mm,textheight=193mm]{geometry}
	\AtBeginDocument{
		\begin{textblock}{5.25}[-0.59,0](0,12)
			\centering\Large
			\textcolor{red}{\textbf{LNCS Preview mode active.}}
		\end{textblock}
	}
	\def\choosepubinfo{0}
\fi

\ifnum\overflow=1
\fi

\ifnum\llncs=0
	\renewcommand*{\backref}[1]{(Cited on page~#1.)}
\fi

\ifnum\allowbreaks=1
	\allowdisplaybreaks
\fi

\ifnum\showlabels=1
	\usepackage{showkeys}
\fi

\ifnum\anonymous=1
	\def\authnotes{0}
\fi

\ifnum\changefont=1
	\usepackage{times}
\fi

\ifnum\LNCSpreview=0
	\ifx\pagelimit\empty
	\else
		\newcounter{pagewarning}
		\AddEverypageHook{
			\ifnum\thepage>\thepagewarning
				\begin{textblock}{1}[0.5,0](0,-13)
					\centering\Large
					\textcolor{red}{\textbf{This page is exceeding the page limit.}}
				\end{textblock}
			\fi
	}
	\fi
\fi

\newcommand{\authnote}[2]{\ifnum\authnotes=1 \begin{center}\fbox{\begin{minipage}{.98\textwidth}
\textbf{#1 says:} #2\end{minipage}}\end{center} \fi}

\newlength{\strutdepth}%
\newcommand{\notes}[3]{
\ifnum\authnotes=1
	\noindent{\bfseries
	\color{#1}{#3}\color{#1}}%
	\strut\vadjust{\kern-\strutdepth%
		\vtop to \strutdepth{%
			\baselineskip\strutdepth%
			\vss\llap{{\large\color{#1}\textbf{#2}\quad\color{black}}}\null%
		}%
	}%
\fi
}

\ifnum\LNCSpreview=0

\else

\fi

\mathchardef\hyphen="2D

\newtheoremstyle{thicktheorem}%
{\ifnum\rmvtheoremspace=1
	0.4
\fi
\topsep}
{\ifnum\rmvtheoremspace=1
	0.4
\fi
\topsep}
{\itshape}{}%
{\bfseries}%
{.}
{ }%
{\thmname{#1}\thmnumber{ #2}%
	\ifnum\dropargs=0
		\thmnote{ (#3)}%
	\fi
}

\newtheoremstyle{remark}
{\ifnum\rmvtheoremspace=1
	0.4
\fi
\topsep}
{\ifnum\rmvtheoremspace=1
	0.4
\fi
\topsep}
	{}
	{}
	{}
	{.}
	{ }
	{\textit{\thmname{#1}}\thmnumber{ #2}
		\ifnum\dropargs=0
			\thmnote{ (#3)}%
		\fi
	}

\ifnum\llncs=0
	\theoremstyle{thicktheorem}
	\newtheorem{theorem}{Theorem}[section]
	\newtheorem{lemma}[theorem]{Lemma}
	\newtheorem{corollary}[theorem]{Corollary}
	\newtheorem{proposition}[theorem]{Proposition}
	\newtheorem{definition}[theorem]{Definition}

	\theoremstyle{remark}
	\newtheorem{claim}[theorem]{Claim}
	\newtheorem{remark}[theorem]{Remark}

	\newtheorem{conjecture}[theorem]{Conjecture}
\fi

\theoremstyle{remark}
\newtheorem{assumption}[theorem]{Assumption}
\newtheorem{observation}[theorem]{Observation}
\newtheorem{fact}[theorem]{Fact}
\newtheorem{experiment}{Experiment}
\newtheorem{construction}[theorem]{Construction}
\newtheorem{counterexample}[theorem]{Counterexample}

\ifnum\abbrevref=0
	\crefname{assumption}{Assumption}{Assumptions}
	\crefname{construction}{Construction}{Constructions}
	\crefname{corollary}{Corollary}{Corollaries}
	\crefname{conjecture}{Conjecture}{Conjectures}
	\crefname{definition}{Definition}{Definitions}
	\crefname{exmaple}{Example}{Examples}
	\crefname{experiment}{Experiment}{Experiments}
	\crefname{counterexample}{Counterexample}{Counterexamples}
	\crefname{lemma}{Lemma}{Lemmata}
	\crefname{observation}{Observation}{Observations}
	\crefname{proposition}{Proposition}{Propositions}
	\crefname{remark}{Remark}{Remarks}
	\crefname{theorem}{Theorem}{Theorems}
\else
	\crefname{assumption}{Ass.}{Ass.}
	\crefname{construction}{Constr.}{Constr.}
	\crefname{corollary}{Cor.}{Cor.}
	\crefname{conjecture}{Conj.}{Conj.}
	\crefname{definition}{Def.}{Def.}
	\crefname{exmaple}{Ex.}{Ex.}
	\crefname{experiment}{Exp.}{Exp.}
	\crefname{counterexample}{Counterex.}{Counterex.}
	\crefname{lemma}{Lem.}{Lem.}
	\crefname{observation}{Obs.}{Obs.}
	\crefname{proposition}{Prop.}{Prop.}
	\crefname{remark}{Rem.}{Rem.}
	\crefname{theorem}{Thm.}{Thms.}
\fi

\crefname{claim}{Claim}{Claims}
\crefname{fact}{Fact}{Facts}
\crefname{note}{Note}{Notes}

\def\YYYSMcoin{\mbox{\begin{tikzpicture}[scale=0.0125]
\definecolor{coinbrown}{HTML}{D89E36}\definecolor{coindarkyellow}{HTML}{F8D81E}\definecolor{coinyellow}{HTML}{F8F800}\fill[coinyellow] (3,-1) rectangle (9,9);\fill(0,0) rectangle (1,8);\fill(1,8) rectangle (2,10);\fill(2,10) rectangle (4,11);\fill(4,11) rectangle (8,12);\fill(8,11) rectangle (10,10);\fill(10,10) rectangle (11,8);\fill(11,8) rectangle (12,0);\fill(10,-2) rectangle (11,0);\fill(8,-3) rectangle (10,-2);\fill(4,-4) rectangle (8,-3);\fill(2,-3) rectangle (4,-2);\fill(1,0) rectangle (2,-2);\fill (5,-1) rectangle (7,0);\fill (7,0) rectangle (8,8);\fill[coinbrown] (9,8) rectangle (10,10);\fill[coinbrown] (10,0) rectangle (11,8);\fill[coinbrown] (9,-2) rectangle (10,0);\fill[coinbrown] (8,-2) rectangle (9,-1);\fill[coinbrown] (4,-3) rectangle (8,-2);\fill[coindarkyellow] (2,-2) rectangle (3,8);\fill[coindarkyellow] (3,-2) rectangle (8,-1);\fill[coindarkyellow] (8,-1) rectangle (9,0);\fill[coindarkyellow] (9,0) rectangle (10,8);\fill[coindarkyellow] (8,8) rectangle (9,10);\fill[coindarkyellow] (4,9) rectangle (8,10);\fill[coindarkyellow] (3,8) rectangle (4,9);\fill[coindarkyellow] (5,0) rectangle (7,2);\fill[coindarkyellow] (6,2) rectangle (7,8);\fill[white] (4,0) rectangle (5,8);\fill[white] (5,8) rectangle (7,9);
\end{tikzpicture}}}

\def\YYYdie{\mbox{\begin{tikzpicture}[scale=0.85,x=1em,y=1em,radius=0.09]
\draw[rounded corners=1,line width=.25pt] (0,0) rectangle (1,1);\fill (0.275,0.275) circle;\fill (0.725,0.725) circle;\fill (0.5,0.5) circle;
\end{tikzpicture}}}

\newcommand{\getsr}{
	\ifnum\dieordollar=0
		\mathrel{\vbox{\offinterlineskip\ialign{
			\hfil##\hfil\cr
			\hspace{0.1em}$\scriptscriptstyle\$$\cr
			$\leftarrow$\cr
		}}}
	\fi
	\ifnum\dieordollar=1
		\mathrel{\vbox{\offinterlineskip\ialign{
			\hfil##\hfil\cr
			{\scalebox{0.5}{\hspace{0.4em}\YYYdie}}\cr
			\noalign{\kern0.05ex}
			$\leftarrow$\cr
		}}}
	\fi
	\ifnum\dieordollar=2
		\mathrel{\vbox{\offinterlineskip\ialign{
			\hfil##\hfil\cr
			\hspace{0.1em}$\YYYSMcoin$\cr
			$\leftarrow$\cr
		}}}
	\fi
}

%% file: stydef.tex
\def\makeuppercase#1{
\expandafter\newcommand\csname sf#1\endcsname{\mathsf{#1}}
\expandafter\newcommand\csname frak#1\endcsname{\mathfrak{#1}}
\expandafter\newcommand\csname bb#1\endcsname{\mathbb{#1}}
\expandafter\newcommand\csname bf#1\endcsname{\textbf{#1}}
}

\def\makelowercase#1{
\expandafter\newcommand\csname frak#1\endcsname{\mathfrak{#1}}
\expandafter\newcommand\csname bf#1\endcsname{\textbf{#1}}
}

\newcounter{char}
\loop
	\edef\letter{\alph{char}}
	\edef\Letter{\Alph{char}}
	\expandafter\makelowercase\letter
	\expandafter\makeuppercase\Letter
	\stepcounter{char}
	\unless\ifnum\thechar>26
\repeat

\newcommand{\etal}{{et~al.\xspace}}


%% file: macros-others.tex
\makeatletter
\newcounter{game}
\def\newGames#1#2#3{%
  \xdef\gameNS{#1}\xdef\gamePrefix{#2}\setcounter{game}{#3}\addtocounter{game}{-1}%
  \immediate\write\@auxout{\string\expandafter\string\xdef\noexpand\csname game-prefix-#1\string\endcsname{#2}}%
}
\def\newGame#1{%
  \xdef\prevGame{\gamePrefix\arabic{game}}\stepcounter{game}\xdef\thisGame{\gamePrefix\arabic{game}}%
  \immediate\write\@auxout{\string\expandafter\string\xdef\noexpand\csname game-\gameNS-#1\string\endcsname{\arabic{game}}}%
}
\makeatother
\def\safecsname#1{\expandafter\ifx\csname#1\endcsname\relax\else\csname#1\endcsname\fi}

\renewcommand{\Game}[1][]{\mathcmd{\textrm{Game\if!#1!\else~\ensuremath{#1}\fi}}}

\usepackage[font=small,
format=plain,
labelformat=simple, 
singlelinecheck=false,
labelfont=bf,
up
]{caption}
\DeclareCaptionFormat{myformat}{#1#2#3\hrulefill}


%% file: macros-Ryo.tex
\newcommand{\qA}{\qalgo{A}}

\DeclareFontFamily{U}{skulls}{}
\DeclareFontShape{U}{skulls}{m}{n}{ <-> skull }{}
\newcommand{\skull}{\text{\usefont{U}{skulls}{m}{n}\symbol{'101}}}

\newcommand{\pirateC}{\qstate{C}_{\skull}}
\newcommand{\qstateq}{\qstate{q}}

\usepackage{mathtools}

\newcommand{\chosen}{\leftarrow}

\renewcommand{\gets}{\leftarrow}
\newcommand{\lrun}{\leftarrow}
\newcommand{\out}{=}

\newcommand{\ra}{\rightarrow}

\newcommand{\seteq}{\coloneqq}

\newcommand{\tensor}{\otimes}
\newcommand{\concat}{\|}

\newcommand{\ceil}[1]{\left\lceil{#1}\right\rceil}

\newcommand{\setbracket}[1]{\{#1\}}

\newcommand{\setbk}[1]{\{#1\}}

\newcommand{\cA}{\mathcal{A}}
\newcommand{\cB}{\mathcal{B}}
\newcommand{\cC}{\mathcal{C}}
\newcommand{\cD}{\mathcal{D}}
\newcommand{\cE}{\mathcal{E}}
\newcommand{\cF}{\mathcal{F}}

\newcommand{\cH}{\mathcal{H}}
\newcommand{\cI}{\mathcal{I}}

\newcommand{\cM}{\mathcal{M}}
\newcommand{\cN}{\mathcal{N}}
\newcommand{\cO}{\mathcal{O}}
\newcommand{\cP}{\mathcal{P}}
\newcommand{\cQ}{\mathcal{Q}}
\newcommand{\cR}{\mathcal{R}}
\newcommand{\cS}{\mathcal{S}}

\newcommand{\cU}{\mathcal{U}}
\newcommand{\cV}{\mathcal{V}}

\newcommand{\tlC}{\widetilde{C}}
\newcommand{\tlD}{\widetilde{D}}

\newcommand{\tlp}{\widetilde{p}}

\newcommand{\oly}{\overline{y}}

\newcommand{\N}{\mathbb{N}}
\newcommand{\Z}{\mathbb{Z}}

\newcommand{\R}{\mathbb{R}}

\newcommand{\Zq}{\mathbb{Z}_q}

\newcommand{\Fs}{\mathcal{F}}

\newcommand{\Cs}{\mathcal{C}}

\newcommand{\Ps}{\mathcal{P}}

\newcommand{\randspace}{\mathcal{R}}

\newcommand{\coin}{\keys{coin}}

\newcommand{\M}{\cM}

\newcommand{\Params}{\algo{\Params}}

\newcommand{\Sampler}{\algo{Samp}}

\newcommand{\secp}{\lambda}

\newcommand{\aux}{\mathsf{aux}}

\newcommand{\ctlen}{{\ell_{\mathsf{ct}}}}
\newcommand{\ptlen}{{\ell_{\mathsf{pt}}}}
\newcommand{\ptxtlen}{{\ell_{\mathsf{p}}}}

\newcommand{\rmOWF}{\textrm{one-way function}\xspace}

\newcommand{\rmIO}{\textrm{IO}\xspace}

\newcommand{\LWE}{\textrm{LWE}}

\newcommand{\sfreal}[2]{\mathsf{Real}^{\mathsf{#1}\textrm{-}\mathsf{#2}}}

\newcommand{\sfrand}{\mathsf{Rand}}

\newcommand{\advt}[2]{\mathsf{Adv}_{#1}^{\mathsf{#2}}}
\newcommand{\adva}[2]{\mathsf{Adv}_{#1}^{\mathsf{#2}}}
\newcommand{\advb}[3]{\mathsf{Adv}_{#1}^{\mathsf{#2} \mbox{-} \mathsf{#3}}}
\newcommand{\advc}[4]{\mathsf{Adv}_{#1}^{\mathsf{#2} \mbox{-} \mathsf{#3} \mbox{-} \mathsf{#4}}}

\newcommand{\expt}[2]{\mathsf{Expt}_{#1}^{\mathsf{#2}}}

\newcommand{\expb}[3]{\mathsf{Exp}_{#1}^{ \mathsf{#2} \mbox{-} \mathsf{#3}}}
\newcommand{\expc}[4]{\mathsf{Exp}_{#1}^{ \mathsf{#2} \mbox{-} \mathsf{#3} \mbox{-} \mathsf{#4}}}

\newcommand{\hybi}[1]{\mathsf{Hyb}_{#1}}
\newcommand{\hybij}[2]{\mathsf{Hyb}_{#1}^{#2}}

\newcommand{\rhybij}[2]{\mathsf{rHyb}_{#1}^{#2}}

\newcommand*{\sk}{\keys{sk}}
\newcommand*{\pk}{\keys{pk}}

\newcommand*{\msk}{\keys{msk}}

\newcommand*{\dk}{\keys{dk}}
\newcommand*{\ck}{\keys{ck}}
\newcommand*{\ek}{\keys{ek}}

\newcommand*{\mk}{\keys{mk}}

\newcommand*{\vk}{\keys{vk}}

\newcommand*{\hk}{\keys{hk}}

\newcommand*{\xk}{\keys{xk}}

\newcommand*{\ct}{\keys{ct}}

\newcommand*{\stinfo}{\keys{st}}

\newcommand*{\peek}{\keys{pe.ek}}
\newcommand*{\pedk}{\keys{pe.dk}}
\newcommand*{\iop}{\tau}

\newcommand*{\msg}{\keys{m}}

\newcommand*{\keys}[1]{\mathsf{#1}}
\newcommand*{\qstate}[1]{\mathpzc{#1}}
\newcommand*{\qreg}[1]{{\color{gray}{\mathsf{#1}}}}

\newcommand*{\algo}[1]{\ensuremath{\mathsf{#1}}}
\newcommand*{\qalgo}[1]{\ensuremath{\mathpzc{#1}}}

\newcommand{\redline}[1]{\textcolor{red}{\underline{\textcolor{black}{#1}}}}

\newenvironment{boxfig}[2]{\begin{figure}[#1]\fbox{\begin{minipage}{0.97\linewidth}
                        \vspace{0.2em}
                        \makebox[0.025\linewidth]{}
                        \begin{minipage}{0.95\linewidth}
            {{
                        #2 }}
                        \end{minipage}
                        \vspace{0.2em}
                        \end{minipage}}}{\end{figure}}

\newcommand{\pprotocol}[4]{
\begin{boxfig}{t!}{\footnotesize 
\centering{\textbf{#1}}
    #4
} \caption{#2}
\label{#3}
\end{boxfig}
}

\newcommand{\protocol}[4]{
\pprotocol{#1}{#2}{#3}{#4} }

 \newcounter{expitem}[table]

\newcommand{\bit}{\{0,1\}}

\newcommand{\mat}[1]{\boldsymbol{#1}}

\newcommand{\mm}[1]{\boldsymbol{#1}}
\newcommand{\mv}[1]{\boldsymbol{#1}}

\newcommand{\Rand}{\algo{R}}

\newcommand{\prf}{\algo{F}}
\newcommand{\prfg}{\algo{G}}

\newcommand{\Setup}{\algo{Setup}}
\newcommand{\Gen}{\algo{Gen}}

\newcommand{\KG}{\algo{KG}}
\newcommand{\Enc}{\algo{Enc}}
\newcommand{\Dec}{\algo{Dec}}

\newcommand{\Sign}{\algo{Sign}}
\newcommand{\Vrfy}{\algo{Vrfy}}

\newcommand\PKE{\algo{PKE}}

\newcommand\SIG{\algo{SIG}}

\newcommand\PE{\algo{PE}}

\newcommand{\iO}{i\cO}

\newcommand{\io}{\mathsf{io}}

\newcommand{\PRG}{\algo{PRG}}
\newcommand{\PRF}{\algo{PRF}}

\newcommand{\prfgen}{\PRF.\Gen}
\newcommand{\Eval}{\algo{Eval}}
\newcommand{\Puncture}{\algo{Puncture}}

\newcommand{\CPRF}{\algo{CPRF}}
\newcommand{\Constrain}{\algo{Constrain}}
\newcommand{\constrain}{\algo{Constrain}}

\newcommand{\CEval}{\algo{CEval}}
\newcommand{\cprf}{\keys{cprf}}

\newcommand{\WM}{\mathsf{WM}}

\newcommand{\Mark}{\algo{Mark}}

\newcommand{\unmarked}{\mathsf{unmarked}}

\newcommand{\qExtract}{\qalgo{Extract}}
\newcommand{\API}{\qalgo{API}}

\newcommand{\shiftdis}[1]{\Delta_{\mathsf{Shift}}^{#1}}

\newcommand{\projimp}{\algo{ProjImp}}
\newcommand{\cproj}{\algo{CProj}}
\newcommand{\Live}{\mathsf{Live}}
\newcommand{\BadExt}{\mathsf{BadExt}}
\newcommand{\GoodExt}{\mathsf{GoodExt}}

\newcommand{\ketisu}{\ket{\mathds{1}_\randspace}}
\newcommand{\braisu}{\bra{\mathds{1}_\randspace}}
\newcommand{\kbisu}{\ketisu\braisu}
\newcommand{\IsU}{\algo{IsU_{\randspace}}}

\newcommand{\Commit}{\algo{Com}}

\newcommand{\com}{\keys{com}}

\newcommand{\cind}{\stackrel{\mathsf{c}}{\approx}}

\newcommand{\negl}{{\mathsf{negl}}}

\newcommand{\poly}{{\mathrm{poly}}}

\newcommand{\zo}[1]{\{0,1\}^{#1}}
\newcommand{\bin}{\{0,1\}}

\newcommand{\xor}{\oplus}

\newcommand{\calR}{\mathcal{R}}

\newcommand{\SUC}{{\tt SUC}}

\newcommand{\Sim}{\algo{Sim}}

\newcommand{\SKE}{\algo{SKE}}

\newcommand{\skekey}{\algo{ske}.\algo{k}}
\newcommand{\E}{\algo{E}}
\newcommand{\D}{\algo{D}}

\newcommand{\PuncPRF}{\algo{PPRF}}



%% file: macros-Fuyuki.tex
\newcommand{\Domprf}{\mathsf{Dom}}
\newcommand{\Ranprf}{\mathsf{Ran}}

\newcommand{\qB}{\qalgo{B}}

\newcommand{\qD}{\qalgo{D}}
\newcommand{\inplen}{n}
\newcommand{\outlen}{m}
\newcommand{\msglen}{k}

\newcommand{\Drev}{D^{\mathtt{rev}}}
\newcommand{\cPrev}{\cP^{\mathtt{rev}}}
\newcommand{\ELWMPRF}{\mathsf{ELWMPRF}}
\newcommand{\QELWMPRF}{\mathsf{QELWMPRF}}

\newcommand{\Find}{\mathtt{Find}}

\newcommand{\Oracle}[1]{O_{\mathtt{#1}}}

\newcommand{\Ev}{\mathtt{Ev}}

%% file: titlepage.tex

\newcount\authorcounter
\newcommand{\provideauthors}{%
		\ifnum\authorcounter<\theauthorcount
			\csname\the\authorcounter name\endcsname
			\expandafter\ifx\csname\the\authorcounter thanks\endcsname\empty
			\else
				\thanks{\csname\the\authorcounter thanks\endcsname}
			\fi%
			\inst{\csname\the\authorcounter institute\endcsname} 
			\and 
			\global\advance\authorcounter by 1 
			\provideauthors
		\else
			\csname\the\authorcounter name\endcsname 
			\expandafter\ifx\csname\the\authorcounter thanks\endcsname\empty 
			\else
				\thanks{\csname\the\authorcounter thanks\endcsname} 
			\fi%
			\inst{\csname\the\authorcounter institute\endcsname} 
		\fi
}

\def\atleastoneauthorplaced{0}
\newcommand{\providerunning}{%
	\ifnum\authorcounter<\theauthorcount%
		\expandafter\ifx\csname\the\authorcounter running\endcsname\empty
		\else
			\ifnum\authorcounter>1
				\ifnum\atleastoneauthorplaced=1
					\and%
				\fi
			\fi
			\csname\the\authorcounter running\endcsname
			\def\atleastoneauthorplaced{1}
		\fi
		\global\advance\authorcounter by 1
		\providerunning%
	\else%
		\expandafter\ifx\csname\the\authorcounter running\endcsname\empty
		\else
			\ifnum\authorcounter>1
				\ifnum\atleastoneauthorplaced=1
					\and%
				\fi
			\fi
			\csname\the\authorcounter running\endcsname
		\fi
	\fi
}

\newcount\institutecounter

\newcommand{\provideinstitutes}{%
	\ifnum\institutecounter<\theinstitutecount%
		\ifnum\llncs=0
			$^{\csname\the\institutecounter number\endcsname}$
		\fi
		\csname\the\institutecounter instname\endcsname
		
		\email{
			\ifx\contactmail\empty
				\csname\the\institutecounter mail\endcsname
			\else
				\href{mailto:\contactmail}{\csname\the\institutecounter mail\endcsname}
			\fi
		}
		
		\and%
			\global\advance\institutecounter by 1
		\provideinstitutes%
	\else%
		\ifnum\llncs=0
			\ifcsname 1name\endcsname
				$^{\csname\the\institutecounter number\endcsname}$
			\fi
		\fi
		\csname\the\institutecounter instname\endcsname
		
		\email{
			\ifx\contactmail\empty
				\csname\the\institutecounter mail\endcsname
			\else
				\href{mailto:\contactmail}{\csname\the\institutecounter mail\endcsname}
			\fi
		}
	\fi
}

\title{
	\ifnum\stuffedtitlepage=1
		\ifnum\llncs=1
			\vspace*{-7ex}
		\else
		\vspace*{-3ex}
		\fi
		\textbf{\titletext}
		\ifnum\llncs=1
			\vspace*{-2ex}
		\else
			\vspace*{-1ex}
		\fi
	\else
		\textbf{\titletext}
	\fi
}
\ifnum\anonymous=1
	\author{}
\else
	\ifnum\llncs=0
		\newcommand{\inst}[1]{{
			\ifcsname 1name\endcsname
				$^{#1}$
			\fi
			}}
	\fi
	\ifcsname 1name\endcsname
		\author{
			\global\authorcounter 1
			\provideauthors
		}
	\fi
\fi

\ifnum\llncs=1
	\titlerunning{\runningtitle}
	\ifnum\anonymous=1
		\institute{}
		\authorrunning{}
	\else
		\ifcsname 1instname\endcsname{
			\institute{
				\global\institutecounter 1
				\provideinstitutes
			}
		\fi
		\ifcsname 1name\endcsname{
			\authorrunning{
				\global \authorcounter 1
				\providerunning
			}
		\fi
	\fi
\fi
\maketitle
\ifnum\stuffedtitlepage=1
	\ifnum\llncs=0
		\vspace{-4ex}
	\fi
\fi

\ifnum\llncs=0
	\ifnum\anonymous=0
		\newcommand{\email}[1]{
			\texttt{
				\ifx\contactmail\empty
					#1
				\else
					\href{mailto:\contactmail}{#1}
				\fi
			}
		}
		\newcommand{\and}{}
		\ifnum\stuffedtitlepage=1
			\ifnum\llncs=0
				\vspace{-2ex}
			\fi
		\fi
		\begin{small}
			\begin{center}
				\global \institutecounter 1
				\provideinstitutes
			\end{center}
		\end{small}
	\fi
\fi

\ifnum\stuffedtitlepage=1
	\ifnum\llncs=1
		\vspace*{-4ex}
	\else
		\vspace*{-2ex}
	\fi
\fi

\begin{abstract}
	
 \input{section/abstract}

	\vspace{1ex}
\ifnum\llncs=1
\else

	\textbf{Keywords\ifnum\llncs=1{.}\else{:}\fi}\keywords
\fi
\end{abstract}
\ifnum\stuffedtitlepage=1
	\ifnum\llncs=1
		\vspace*{-2ex}
	\fi
\fi

\ifnum\llncs=0
	\vspace{1ex}
\fi

\ifnum\choosepubinfo=1
\def\pubinfo{
	\noindent An extended abstract of this paper will appear at
	\ifx\pubinfoCONFERENCE\empty \textcolor{red}{conference missing}\else \pubinfoCONFERENCE\fi.
}
\fi

\ifnum\choosepubinfo=2
	\def\pubinfo{
		\noindent \copyright\ IACR 
		\ifx\pubinfoYEAR\empty \textcolor{red}{year missing}\else \pubinfoYEAR\fi.
		This article is the final version submitted by the author(s) to the IACR and to Springer-Verlag on
		\ifx\pubinfoSUBMISSIONDATE\empty \textcolor{red}{submission date missing}\else \pubinfoSUBMISSIONDATE\fi.
		The version published by Springer-Verlag is available at
		\ifx\pubinfoDOI\empty \textcolor{red}{DOI missing}\else \pubinfoDOI\fi.
	}
\fi

\ifnum\choosepubinfo=3
	\def\pubinfo{
		\noindent \copyright\ IACR
		\ifx\pubinfoYEAR\empty \textcolor{red}{year missing}\else \pubinfoYEAR\fi.
		This article is a minor revision of the version published by Springer-Verlag available at
		\ifx\pubinfoDOI\empty \textcolor{red}{DOI missing}\else \pubinfoDOI\fi.
	}
\fi

\ifnum\choosepubinfo=4
	\def\pubinfo{
		\noindent This article is based on an earlier article:
		\ifx\pubinfoBIBDATA\empty \textcolor{red}{bibliographic data missing}\else \pubinfoBIBDATA\fi,
		\copyright\ IACR
		\ifx\pubinfoYEAR\empty \textcolor{red}{year missing}\else \pubinfoYEAR\fi,
		\ifx\pubinfoDOI\empty \textcolor{red}{DOI missing}\else \pubinfoDOI\fi.
	}
\fi

\ifnum\choosepubinfo=5
		\def\pubinfo{
			\noindent \pubinfoindividual
		}
	\fi

\textblockorigin{0.5\paperwidth}{0.9\paperheight}
\setlength{\TPHorizModule}{\textwidth}

\newlength{\pubinfolength}
\ifnum\choosepubinfo=0
\else
	\settowidth{\pubinfolength}{\pubinfo}
	\begin{textblock}{1}[0.5,0](0,.25)
		 \ifnum\pubinfolength<\textwidth
			\centering
		\fi
		\pubinfo
	\end{textblock}
\fi
\thispagestyle{empty}

%% file: section/abstract.tex

We initiate the study of software watermarking against quantum adversaries.
A quantum adversary generates a \emph{quantum state} as a pirate software that potentially removes an embedded message from a \emph{classical} marked software.
Extracting an embedded message from quantum pirate software is difficult since measurement could irreversibly alter the quantum state.
In software watermarking against classical adversaries, a message extraction algorithm crucially uses the (input-output) behavior of a classical pirate software to extract an embedded message. Even if we instantiate existing watermarking PRFs with quantum-safe building blocks, it is not clear whether they are secure against quantum adversaries due to the quantum-specific property above.
Thus, we need entirely new techniques to achieve software watermarking against quantum adversaries.

In this work, we define secure watermarking PRFs for quantum adversaries (unremovability against quantum adversaries). We also present two watermarking PRFs as follows.
\begin{itemize}
\item We construct a privately extractable watermarking PRF against quantum adversaries from the quantum hardness of the learning with errors (LWE) problem. The marking and extraction algorithms use a public parameter and a private extraction key, respectively. The watermarking PRF is unremovable even if adversaries have (the public parameter and) access to the extraction oracle, which returns a result of extraction for a queried quantum circuit.
\item We construct a publicly extractable watermarking PRF against quantum adversaries from indistinguishability obfuscation (IO) and the quantum hardness of the LWE problem. The marking and extraction algorithms use a public parameter and a public extraction key, respectively. The watermarking PRF is unremovable even if adversaries have the extraction key (and the public parameter).
\end{itemize}

We develop a quantum extraction technique to extract information (a classical string) from a quantum state without destroying the state too much.
We also introduce the notion of extraction-less watermarking PRFs as a crucial building block to achieve the results above by combining the tool with our quantum extraction technique.

%% file: section/intro.tex

\section{Introduction}\label{sec:intro}

\subsection{Background}\label{sec:background}

Software watermarking is a cryptographic primitive that achieves a digital analog of watermarking.
A marking algorithm of software watermarking can embed an arbitrary message (bit string) into a computer software modeled as a circuit.
A marked software almost preserves the functionality of the original software.
An extraction algorithm of software watermarking can extract the embedded message from a marked software.
Secure software watermarking should guarantee that no adversary can remove the embedded message without significantly destroying the functionality of the original software (called unremovability).

Barak, Goldreich, Impagliazzo, Rudich, Sahai, Vadhan, and Yang~\cite{JACM:BGIRSVY12} initiate the study of software watermarking and present the first definition of cryptographically secure software watermarking. Hopper, Molnar, and Wagner~\cite{TCC:HopMolWag07} also study the definition of cryptographically secure watermarking for perceptual objects. However, both works do not present a secure concrete scheme.
A few works study secure constructions of watermarking for cryptographic primitives~\cite{PKC:NacShaSte99,IEICE:YosFuj11,EC:Nishimaki13,IEICE:Nishimaki19}, but they consider only restricted removal strategies.
Cohen, Holmgren, Nishimaki, Wichs, and Vaikuntanathan~\cite{SIAMCOMP:CHNVW18} present stronger definitions for software watermarking and the first secure watermarking schemes for cryptographic primitives \emph{against arbitrary removal strategies}. After the celebrated work, watermarking for cryptographic primitives have been extensively studied~\cite{PKC:BonLewWu17,myJC:KimWu21,TCC:QuaWicZir18,C:KimWu19,AC:YALXY19,C:GKMWW19,C:YAYX20,TCC:Nishimaki20}.

Primary applications of watermarking are identifying ownership of objects and tracing users that distribute illegal copies.
Watermarking for cryptographic primitives also has another exciting application.
 Aaronson, Liu, Liu, Zhandry, and Zhang~\cite{C:ALLZZ21} and Kitagawa, Nishimaki, and Yamakawa~\cite{myTCC:KitNisYam21} concurrently and independently find that we can construct secure software leasing schemes by combining watermarking with quantum cryptography.\footnote{Precisely speaking, Aaronson et al. achieve copy-detection schemes~\cite{C:ALLZZ21}, which are essentially the same as secure software leasing schemes.} Secure software leasing~\cite{EC:AnaLaP21} is a quantum cryptographic primitive that prevents users from generating authenticated pirated copies of leased software.\footnote{Leased software must be a quantum state since classical bit strings can be easily copied.}
Since watermarking has such an exciting application in quantum cryptography and quantum computers might be an imminent threat to cryptography due to rapid progress in research on quantum computing, it is natural and fascinating to study secure software watermarking in the quantum setting.

In quantum cryptography, building blocks must be quantum-safe such as lattice-based cryptography~\cite{JACM:Regev09}.
However, even if we replace building blocks of existing cryptographic primitives/protocols with quantum-safe ones, we do not necessarily obtain quantum-safe cryptographic primitives/protocols~\cite{AC:BDFLSZ11,FOCS:AmbRosUnr14}. We sometimes need new proof techniques which are different from classical ones due to quantum specific properties such as no-cloning and superposition access~\cite{SIAMCOMP:Watrous09,C:Zhandry12,FOCS:Zhandry12,EC:Unruh12,C:Zhandry19,myFOCS:CMSZ21}. Even worse, we must consider entirely different security models in some settings. Zhandry~\cite{TCC:Zhandry20} studies traitor tracing~\cite{C:ChoFiaNao94} in the quantum setting as such an example. In quantum traitor tracing, an adversary can output a \emph{quantum state} as a pirate decoder.
Zhandry shows that we need new techniques for achieving quantum traitor tracing because running a quantum pirate decoder to extract information may irreversibly alter the state due to measurement.

Zhandry~\cite{TCC:Zhandry20} refers to software watermarking as a cryptographic primitive that has a similar issue to quantum traitor tracing. However, his work focuses only on traitor tracing and does not study software watermarking against quantum adversaries.
If we use software watermarking in the quantum setting, an adversary can output a \emph{quantum state} as a pirate circuit where an embedded message might be removed.
However, previous works consider a setting where an adversary outputs a \emph{classical} pirate circuit. It is not clear whether watermarking schemes based on quantum-safe cryptography are secure against quantum adversaries because we need an entirely new extraction algorithm to extract an embedded message from a \emph{quantum} pirate circuit. Thus, the main question in this study is:

\begin{center}
\emph{Can we achieve secure watermarking for cryptographic primitives against quantum adversaries?}
\end{center}
We affirmatively answer this question in this work.

\subsection{Our Result}\label{sec:result}
Our main contributions are two-fold. One is the definitional work. We define watermarking for pseudorandom functions (PRFs) against quantum adversaries, where adversaries output a quantum state as a pirate circuit that distinguishes a PRF from a random function.\footnote{This definitional choice comes from the definition of traceable PRFs~\cite{myAC:GKWW21}. See~\cref{sec:technical_overview,sec:related_work} for the detail.}
The other one is constructing the first secure watermarking PRFs against quantum adversaries.
We present two watermarking PRFs as follows.
\begin{itemize}
\item We construct a privately extractable watermarking PRF against quantum adversaries from the quantum hardness of the learning with errors (LWE) problem. This watermarking PRF is secure in the presence of the extraction oracle and supports public marking. That is, the marking and extraction algorithms use a public parameter and secret extraction key, respectively. The watermarking PRF is unremovable even if adversaries have access to the extraction oracle, which returns a result of extraction for a queried quantum circuit.
\item We construct a publicly extractable watermarking PRF against quantum adversaries from indistinguishability obfuscation (IO) and the quantum hardness of the LWE problem. This watermarking PRF also supports public marking. That is, the marking and extraction algorithms use a public parameter and a public extraction key, respectively. The watermarking PRF is unremovable (we do not need to consider the mark and extraction oracles since it supports public marking and public extraction).
\end{itemize}
The former and latter PRFs satisfy weak pseudorandomness and standard (strong) pseudorandomness even against a watermarking authority, respectively.

We develop a quantum extraction algorithm to achieve the results above.
Zhandry~\cite{TCC:Zhandry20} presents a useful technique for extracting information from quantum states without destroying them too much.
However, we cannot simply apply his technique to the watermarking setting. Embedded information (arbitrary string) is chosen from an exponentially large set in the watermarking setting. On the other hand, in the traitor tracing setting, we embed a user index, which could be chosen from a polynomially large set, in a decryption key. Zhandry's technique is tailored to traitor tracing based on private linear broadcast encryption (PLBE)~\cite{EC:BonSahWat06} where user information is chosen from a polynomially large set with linear structure. Thus, we extend Zhandry's technique~\cite{TCC:Zhandry20} to extract information chosen from an exponentially large set. We also introduce the notion of extraction-less watermarking as a crucial tool to achieve watermarking against quantum adversaries. This tool is a suitable building block for our quantum extraction technique in our watermarking extraction algorithm. These are our technical contributions. See~\cref{sec:technical_overview} for the detail.

Although this paper focuses on watermarking PRFs against quantum adversaries, it is easy to extend our definitions to watermarking public-key encryption (PKE) against quantum adversaries.
In particular, our construction technique easily yields watermarking PKE (where a decryption circuit is marked) schemes. 
We will provide the detail of them in a future version.

We also focus on watermarking PRFs with public marking in this paper. However, we can easily convert our PRFs into ones with private marking.
See~\cref{rmrk:private_marking} for the detail.

\input{section/tech_overview}

\ifnum\submission=0
\input{section/related_work}
\else\fi

\ifnum\submission=1
\subsection{Organization}
Due to the space limitation, we provide preliminaries including notations, basics on quantum informations, and definitions of standard cryptographic tools in \cref{sec:prelim}.
In \cref{sec:measurement_implementation}, we introduce some notions of quantum measurements.
In \cref{sec:def_Q_watermarking}, we define watermarking PRF against quantum adversaries.
In \cref{sec:def_extless_watermarking}, we define extraction-less watermarking PRF.
In \cref{sec:wmprf-from-elwmprf}, we show we can realize watermarking PRF against quantum adversaries from extraction-less watermarking PRF.
In \cref{sec:extless_watermarking_LWE}, we provide an instantiation of extraction-less watermarking PRF with private simulation based on the LWE assumption.
In \cref{sec:pub_ext_watermarking_IO}, we provide an instantiation of extraction-less watermarking PRF with public simulation based on IO and the LWE assumption.

\else\fi

%% file: section/tech_overview.tex

\newcommand{\pp}{\mathsf{pp}}
\newcommand{\prfk}{\mathsf{prfk}}
\newcommand{\WMPRF}{\mathsf{WMPRF}}
\newcommand{\qext}{q_{\mathtt{e}}}
\newcommand{\qdis}{q_{\mathtt{d}}}
\renewcommand{\msglen}{{\ell_{\msg}}}
\newcommand{\cMrev}{\cM^{\mathtt{rev}}}
\newcommand{\Dreal}[1]{D_{\mathtt{real},#1}}

\subsection{Technical Overview}\label{sec:technical_overview}

\paragraph{Syntax of watermarking PRF.}
We first review the syntax of watermarking PRF used in this work.
A watermarking PRF scheme consists of five algorithms $(\Setup,\Gen,\Eval,\Mark,\qExtract)$.\footnote{In this paper, standard math font stands for classical algorithms, and calligraphic font stands for quantum algorithms.}
$\Setup$ outputs a public parameter $\pp$ and an extraction key $\xk$.
$\Gen$ is given $\pp$ and outputs a PRF key $\prfk$ and a public tag $\iop$.
$\Eval$ is the PRF evaluation algorithm that takes as an input $\prfk$ and $x$ in the domain and outputs $y$.
By using $\Mark$, we can generate a marked evaluation circuit that has embedded message $\msg\in \zo{\msglen}$ and can be used to evaluate $\Eval(\prfk,x')$ for almost all $x'$.
Finally, $\qExtract$ is the extraction algorithm supposed to extract the embedded message from a pirated quantum evaluation circuit generated from the marked evaluation circuit.
By default, in this work, we consider the public marking setting, where anyone can execute $\Mark$.
Thus, $\Mark$ takes $\pp$ as an input.
On the other hand, we consider both the private extraction and the public extraction settings.
Thus, the extraction key $\xk$ used by $\qExtract$ is kept secret by an authority in the private extraction setting and made public in the public extraction setting.

In this work, we allow $\qExtract$ to take the public tag $\iop$ generated with the original PRF key corresponding to the pirate circuit.
In reality, we execute $\qExtract$ for a software when a user claims that the software is illegally generated by using her/his PRF key. Thus, it is natural to expect we can use a user's public tag for extraction.
Moreover, pirate circuits are distinguishers, not predictors in this work.
As discussed by Goyal et al.~\cite{myAC:GKWW21}, security against pirate distinguishers is much preferable compared to security against pirate predictors considered in many previous works on watermarking.
In this case, it seems that such additional information fed to $\qExtract$ is unavoidable.
For a more detailed discussion on the syntax, see the discussion in \cref{sec:watermarking_syntax}.

It is also natural to focus on distinguishers breaking weak pseudorandomness of PRFs when we consider pirate distinguishers instead of pirate predictors.
Goyal et al.~\cite{myAC:GKWW21} already discussed this point.
Thus, we focus on watermarking weak PRF in this work.

\paragraph{Definition of unremovability against quantum adversaries.}
We say that a watermarking PRF scheme satisfies unremovability if given a marked evaluation circuit $\tlC$ that has an embedded message $\msg$, any adversary cannot generate a circuit such that it is a ``good enough circuit'', but the extraction algorithm fails to output $\msg$.
In this work, we basically follow the notion of ``good enough circuit'' defined by Goyal et al.~\cite{myAC:GKWW21} as stated above.
Let $D$ be the following distribution for a PRF $\Eval(\prfk,\cdot):\Domprf\ra\Ranprf$.
\begin{description}
\item[$D$:]Generate $b\gets\bit$, $x\gets\Domprf$, and $y_0\gets\Ranprf$. Compute $y_1\gets\Eval(\prfk,x)$. Output $(b,x,y_b)$.
\end{description}
A circuit is defined as good enough circuit with respect to $\Eval(\prfk,\cdot)$ if given $(x,y_b)$ output by $D$, it can correctly guess $b$ with probability significantly greater than $1/2$.
In other words, a circuit is defined as good enough if the circuit breaks weak PRF security.

Below, for a distribution $D'$ whose output is of the form $(b,x,y)$, let $\cM_{D'}=(\mat{M}_{D',0},\mat{M}_{D',1})$ be binary positive operator valued measures (POVMs) that represents generating random $(b,x,y)$ from $D'$ and testing if a quantum circuit can guess $b$ from $(x,y)$.
Then, for a quantum state $\ket{\psi}$, the overall distinguishing advantage of it for the above distribution $D$ is $\bra{\psi}\mat{M}_{D,0}\ket{\psi}$.
Thus, a natural adaptation of the above notion of goodness for quantum circuits might be to define a quantum state $\ket{\psi}$ as good if $\bra{\psi}\mat{M}_{D,0}\ket{\psi}$ is significantly greater than $1/2$.
However, this notion of goodness for quantum circuits is not really meaningful.
The biggest issue is that it does not consider the stateful nature of quantum programs.

This issue was previously addressed by Zhandry~\cite{TCC:Zhandry20} in the context of traitor tracing against quantum adversaries.
In the context of classical traitor tracing or watermarking, we can assume that a pirate circuit is stateless, or can be rewound to its original state.
This assumption is reasonable.
If we have the software description of the pirate circuit, such a rewinding is trivial.
Even if we have a hardware box in which a pirate circuit is built, it seems that such a rewinding is possible by hard reboot or cutting power.
On the other hand, in the context of quantum watermarking, we have to consider that a pirate circuit is inherently stateful since it is described as a quantum state.
Operations to a quantum state can alter the state, and in general, it is impossible to rewind the state into its original state.
Regarding the definition of good quantum circuits above, if we can somehow compute the average success probability $\bra{\psi}\mat{M}_{D,0}\ket{\psi}$ of the quantum state $\ket{\psi}$, the process can change or destroy the quantum state $\ket{\psi}$. Namely, even if we once confirm that the quantum state $\ket{\psi}$ is good by computing $\bra{\psi}\mat{M}_{D,0}\ket{\psi}$, we cannot know the success probability of the quantum state even right after the computation.
Clearly, the above notion of goodness is not the right notion, and we need one that captures the stateful nature of quantum programs.

In the work on traitor tracing against quantum adversaries, Zhandry~\cite{TCC:Zhandry20} proposed a notion of goodness for quantum programs that solves the above issue.
We adopt it. 
For the above POVMs $\cM_D$, let $\cM_D'$ be the projective measurement $\{P_p\}_{p\in[0,1]}$ that projects a state onto the eigenspaces of $\mat{M}_{D,0}$, where each $p$ is an eigenvalue of $\mat{M}_{D,0}$.
$\cM_D'$ is called projective implementation of $\cM_D$ and denoted as $\projimp(\cM_D)$.
Zhandry showed that the following process has the same output distribution as $\cM_{D}$:
\begin{enumerate}
\item Apply the projective measurement $\cM_D'=\projimp(\cM_D)$ and obtain $p$.
\item Output $0$ with probability $p$ and output $1$ with probability $1-p$.
\end{enumerate}
Intuitively, $\cM_D'$ project a state to an eigenvector of $\mat{M}_{D,0}$ with eigenvalue $p$, which can be seen as a quantum state with success probability $p$.
Using $\cM_D'$, Zhandry defined that a quantum circuit is $\Live$ if the outcome of the measurement $\cM_D'$ is significantly greater than $1/2$.
The notion of $\Live$ is a natural extension of the classical goodness since it collapses to the classical goodness for a classical decoder.
Moreover, we can ensure that a quantum state that is tested as $\Live$ still has a high success probability. On the other hand, the above notion of goodness cannot say anything about the post-tested quantum state's success probability even if the test is passed.
In this work, we use the notion of $\Live$ quantum circuits as the notion of good quantum circuits.

\paragraph{Difficulty of quantum watermarking PRF.}
From the above discussion, our goal is to construct a watermarking PRF scheme that guarantees that we can extract the embedded message correctly if a pirated quantum circuit is $\Live$.
In watermarking PRF schemes, we usually extract an embedded message by applying several tests on success probability to a pirate circuit.
When a pirate circuit is a quantum state, the set of tests that we can apply is highly limited compared to a classical circuit due to the stateful nature of quantum states.

One set of tests we can apply without destroying the quantum state is $\projimp(\cM_{D'})$ for distributions $D'$ that are indistinguishable from $D$ from the view of the pirate circuit.\footnote{In the actual extraction process, we use an approximation of projective implementation introduced by Zhandry~\cite{TCC:Zhandry20} since applying a projective implementation is inefficient. In this overview, we ignore this issue for simplicity.
}
We denote this set as $\{\projimp(\cM_{D'}) \mid D'\cind D\}$.
Zhandry showed that if distributions $D_1$ and $D_2$ are indistinguishable, the outcome of $\projimp(\cM_{D_1})$ is close to that of $\projimp(\cM_{D_2})$.
By combining this property with the projective property of projective implementations, as long as the initial quantum state is $\Live$ and we apply only tests contained in $\{\projimp(\cM_{D'})\mid D'\cind D\}$, the quantum state remains $\Live$.
On the other hand, if we apply a test outside of $\{\projimp(\cM_{D'}) \mid D'\cind D\}$, the quantum state might be irreversibly altered.
This fact is a problem since the set $\{\projimp(\cM_{D'}) \mid D'\cind D\}$ only is not sufficient to implement the existing widely used construction method for watermarking PRF schemes.

To see this, we briefly review the method.
In watermarking PRF schemes, the number of possible embedded messages is super-polynomial, and thus we basically need to extract an embedded message in a bit-by-bit manner.
In the method, such a bit-by-bit extraction is done as follows.
For every $i\in[\msglen]$, we define two distributions $S_{i,0}$ and $S_{i,1}$ whose output is of the form $(b,x,y)$ as $D$ above.
Then, we design a marked circuit with embedded message $\msg\in \zo{\msglen}$ so that it can be used to guess $b$ from $(x,y)$ with probability significantly greater than $1/2$ only for $S_{i,0}$ (resp. $S_{i,1}$) if $\msg[i]=0$ (resp. $\msg[i]=1$).
The extraction algorithm can extract $i$-th bit of the message $\msg[i]$ by checking for which distributions of $S_{i,0}$ and $S_{i,1}$ a pirate circuit has a high distinguishing advantage.

As stated above, we cannot use this standard method to extract a message from quantum pirate circuits.
The reason is that $S_{i,0}$ and $S_{i,1}$ are typically distinguishable.
This implies that at least either one of $\projimp(\cM_{S_{i,0}})$ or $\projimp(\cM_{S_{i,1}})$ is not contained in $\{\projimp(\cM_{D'}) \mid D'\cind D\}$.
Since the test outside of $\{\projimp(\cM_{D'})\mid D'\cind D\}$ might destroy the quantum state, we might not be able to perform the process for all $i$, and fail to extract the entire bits of the embedded message.

It seems that to perform the bit-by-bit extraction for a quantum state, we need to extend the set of applicable tests and come up with a new extraction method.

\paragraph{Our solution: Use of reverse projective property.}
We find that as another applicable set of tests, we have  $\projimp(\cM_{D'})$ for distributions $D'$ that are indistinguishable from $\Drev$, where $\Drev$ is the following distribution.
\begin{description}
\item[$\Drev$:]Generate $b\gets\bit$, $x\gets\Domprf$, and $y_0\gets\Ranprf$. Compute $y_1\gets\Eval(\prfk,x)$. Output $(1\oplus b,x,y_b)$.
\end{description}
We denote the set as $\{\projimp(\cM_{D'})\mid D'\cind \Drev\}$.
$\Drev$ is the distribution the first bit of whose output is flipped from that of $D$.
Then, $\cM_{\Drev}$ can be seen as POVMs that represents generating random $(b,x,y_b)$ from $D$ and testing if a quantum circuit \emph{cannot} guess $b$ from $(x,y_b)$.
Thus, we see that $\cM_{\Drev}=(\mat{M}_{D,1},\mat{M}_{D,0})$.
Recall that $\cM_D=(\mat{M}_{D,0},\mat{M}_{D,1})$.

Let $D_1\in\{\projimp(\cM_{D'})\mid D'\cind D\}$ and $\Drev_1$ be the distribution that generates $(b,x,y)\gets D_1$ and outputs $(1\oplus b,x,y)$.
$\Drev_1$ is a distribution contained in $\{\projimp(\cM_{D'})\mid D'\cind \Drev\}$.
Similarly to the relation between $D$ and $\Drev$, if $\cM_{D_1}=(\mat{M}_{D_1,0},\mat{M}_{D_1,1})$, we have $\cM_{\Drev_1}=(\mat{M}_{\Drev_1,1},\mat{M}_{\Drev_1,0})$.
Since $\mat{M}_{D_1,0}+\mat{M}_{D_1,1}=\mat{I}$, $\mat{M}_{D_1,0}$ and $\mat{M}_{D_1,1}$ share the same set of eigenvectors, and if a vector is an eigenvector of $\mat{M}_{D_1,0}$ with eigenvalue $p$, then it is also an eigenvector of $\mat{M}_{D_1,1}$ with eigenvalue $1-p$.
Thus, if apply $\projimp(\cM_{D_1})$ and $\projimp(\cM_{\Drev_1})$ successively to a quantum state and obtain the outcomes $\tlp_1$ and $\tlp_1'$, it holds that $\tlp_1'=1-\tlp_1$.
We call this property the reverse projective property of the projective implementation.

Combining projective and reverse projective properties and the outcome closeness for indistinguishable distributions of the projective implementation, we see that the following key fact holds.

\begin{description}
\item[Key fact:] As long as the initial quantum state is $\Live$ and we apply tests contained in $\{\projimp(\cM_{D'})\mid D'\cind D\}$ or $\{\projimp(\cM_{D'})|D'\cind \Drev\}$, the quantum state remains $\Live$.
Moreover, if the outcome of applying $\projimp(\cM_D)$ to the initial state is $p$, we get the outcome close to $p$ every time we apply a test in $\{\projimp(\cM_{D'})\mid D'\cind D\}$, and we get the outcome close to $1-p$ every time we apply a test in $\{\projimp(\cM_{D'})\mid D'\cind \Drev\}$.
\end{description}

In this work, we perform bit-by-bit extraction of embedded messages by using the above key fact of the projective implementation.
To this end, we introduce the new notion of extraction-less watermarking PRF as an intermediate primitive.

\paragraph{Via extraction-less watermarking PRF.}
An extraction-less watermarking PRF scheme has almost the same syntax as a watermarking PRF scheme, except that it does not have an extraction algorithm $\qExtract$ and instead has a simulation algorithm $\Sim$.
$\Sim$ is given the extraction key $\xk$, the public tag $\iop$, and an index $i\in[\msglen]$, and outputs a tuple of the form $(\gamma,x,y)$.
$\Sim$ simulates outputs of $D$ or $\Drev$ for a pirate circuit depending on the message embedded to the marked circuit corresponding to the pirate circuit.
More concretely, we require that from the view of the pirate circuit generated from a marked circuit with embedded message $\msg\in\zo{\msglen}$, outputs of $\Sim$ are indistinguishable from those of $D$ if $\msg[i]=0$ and are indistinguishable from those of $\Drev$ if $\msg[i]=1$ for every $i\in[\msglen]$.
We call this security notion simulatability for mark-dependent distributions (SIM-MDD security).

By using an extraction-less watermarking PRF scheme $\ELWMPRF$, we construct a watermarking PRF scheme $\WMPRF$ against quantum adversaries as follows.
We use $\Setup,\Gen,\Eval,\Mark$ of $\ELWMPRF$ as $\Setup,\Gen,\Eval,\Mark$ of $\WMPRF$, respectively.
We explain how to construct the extraction algorithm $\qExtract$ of $\WMPRF$ using $\Sim$ of $\ELWMPRF$.
For every $i\in[\msglen]$, we define $D_{\iop,i}$ as the distribution that outputs randomly generated $(\gamma,x,y)\gets\Sim(\xk,\iop,i)$.
Given $\xk$, $\iop$, and a quantum state $\ket{\psi}$, $\qExtract$ extracts the embedded message in the bit-by-bit manner by repeating the following process for every $i\in[\msglen]$.
\begin{itemize}
\item Apply $\projimp(\cM_{D_{\iop,i}})$ to $\ket{\psi_{i-1}}$ and obtain the outcome $\tlp_i$, where $\ket{\psi_0}=\ket{\psi}$ and $\ket{\psi_{i-1}}$ is the state after the $(i-1)$-th loop for every $i\in[\msglen]$.
\item Set $\msg'_i=0$ if $\tlp_i>1/2$ and otherwise $\msg'_i=1$.
\end{itemize}
The extracted message is set to $\msg'_1\|\cdots\|\msg'_\msglen$.

We show that the above construction satisfies unremovability.
Suppose an adversary is given marked circuit $\tlC\gets\Mark(\pp,\prfk,\msg)$ and generates a quantum state $\ket{\psi}$, where $(\pp,\xk)\gets\Setup(1^\secp)$ and $(\prfk,\iop)\gets\Gen(\pp)$.
Suppose also that $\ket{\psi}$ is $\Live$.
This assumption means that the outcome $p$ of applying $\projimp(\cM_D)$ to $\ket{\psi}$ is $1/2+\epsilon$, where $\epsilon$ is an inverse polynomial.
For every $i\in[\msglen]$, from the SIM-MDD security of $\ELWMPRF$, $D_{\iop,i}$ is indistinguishable from $D$ if $\msg[i]=0$ and is indistinguishable from $\Drev$ if $\msg[i]=1$.
This means that $D_{\iop,i}\in\{\projimp(\cM_{D'})\mid D'\cind D\}$ if $\msg[i]=0$ and $D_{\iop,i}\in\{\projimp(\cM_{D'})\mid D'\cind \Drev\}$ if $\msg[i]=1$.
Then, from the above key fact of the projective implementation, it holds that $\tlp_i$ is close to $1/2+\epsilon>1/2$ if $\msg[i]=0$ and is close to $1/2-\epsilon<1/2$ if $\msg[i]=1$.
Therefore, we see that $\qExtract$ correctly extract $\msg$ from $\ket{\psi}$.
This means that $\WMPRF$ satisfies unremovability.

The above definition, construction, and security analysis are simplified and ignore many subtleties.
The most significant point is that we use approximated projective implementations introduced by Zhandry~\cite{TCC:Zhandry20} instead of projective implementations in the actual construction since applying a projective implementation is an inefficient process.
Moreover, though the outcomes of (approximate) projective implementations for indistinguishable distributions are close, in the actual analysis, we have to take into account that the outcomes gradually change every time we apply an (approximate) projective implementation.
These issues can be solved by doing careful parameter settings.

\paragraph{Comparison with the work by Zhandry~\cite{TCC:Zhandry20}.}
Some readers familiar with Zhandry's work~\cite{TCC:Zhandry20} might think that our technique contradicts the lesson from Zhandry's work since it essentially says that once we find a large gap in success probabilities, the tested quantum pirate circuit might self-destruct. However, this is not the case. What Zhandry's work really showed is the following. Once a quantum pirate circuit itself detects that there is a large gap in success probabilities, it might self-destruct. Even if an extractor finds a large gap in success probabilities, if the tested quantum pirate circuit itself cannot detect the large gap, the pirate circuit cannot self-destruct. In Zhandry's work, whenever an extractor finds a large gap, the tested pirate circuit also detects the large gap. In our work, the tested pirate circuit cannot detect a large gap throughout the extraction process while an extractor can find it.

The reason why a pirate circuit cannot detect a large gap in our scheme even if an extractor can find it is as follows. Recall that in the above extraction process of our scheme based on an extraction-less watermarking PRF scheme, we apply $\projimp(\cM_{D_{\iop,i}})$ to the tested pirate circuit for every $i\in[\msglen]$.
Each $D_{\iop,i}$ outputs a tuple of the form $(b,x,y)$ and is indistinguishable from $D$ or $\Drev$ depending on the embedded message.
In the process, we apply $\projimp(\cM_{D_{\iop,i}})$ for every $i\in[\msglen]$, and we get the success probability $p$ if $D_{\iop,i}$ is indistinguishable from $D$ and we get $1-p$ if $D_{\iop,i}$ is indistinguishable from $\Drev$.
The tested pirate circuit needs to know which of $D$ or $\Drev$ is indistinguishable from the distribution $D_{\iop,i}$ behind the projective implementation to know which of $p$ or $1-p$ is the result of an application of a projective implementation.
However, this is impossible.
The tested pirate circuit receives only $(x,y)$ part of $D_{\iop,i}$'s output and not $b$ part. (Recall that the task of the pirate circuit is to guess $b$ from $(x,y)$.)
The only difference between $D$ and $\Drev$ is that the first-bit $b$ is flipped.
Thus, if the $b$ part is dropped, $D_{\iop,i}$ is, in fact, indistinguishable from both $D$ and $\Drev$.
 As a result, the pirate program cannot know which of $p$ or $1-p$ is the result of an application of a projective implementation. In other words, the pirate circuit cannot detect a large gap in our extraction process.

\paragraph{Instantiating extraction-less watermarking PRF.}
In the rest of this overview, we will explain how to realize extraction-less watermarking PRF.

We consider the following two settings similar to the ordinary watermarking PRF. Recall that we consider the public marking setting by default.
\begin{description}
\item[Private-simulatable:] In this setting, the extraction key $\xk$ fed into $\Sim$ is kept secret. We require that SIM-MDD security hold under the existence of the simulation oracle that is given a public tag $\iop'$ and an index $i'\in[\msglen]$ and returns $\Sim(\xk,\iop',i')$. An extraction-less watermarking PRF scheme in this setting yields a watermarking PRF scheme against quantum adversaries in private-extractable setting where unremovability holds for adversaries who can access the extraction oracle.
\item[Public-simulatable:] In this setting, the extraction key $\xk$ is publicly available.
An extraction-less watermarking PRF scheme in this setting yields a watermarking PRF scheme against quantum adversaries in the public-extractable setting.
\end{description}
We provide a construction in the first setting using private constrained PRF based on the hardness of the LWE assumption.
Also, we provide a construction in the second setting based on IO and the hardness of the LWE assumption.

To give a high-level idea behind the above constructions, in this overview, we show how to construct a public-simulatable extraction-less watermarking PRF in the token-based setting~\cite{SIAMCOMP:CHNVW18}.
In the token-based setting, we treat a marked circuit $\tlC\gets\Mark(\pp,\prfk,\msg)$ as a tamper-proof hardware token that an adversary can only access in a black-box way.

Before showing the actual construction, we explain the high-level idea.
Recall that SIM-MDD security requires that an adversary $\qA$ who is given $\tlC\gets\Mark(\pp,\prfk,\msg)$ cannot distinguish $(\gamma^\ast,x^\ast,y^\ast)\gets\Sim(\xk,\iop,i^\ast)$ from an output of $D$ if $\msg[i^\ast]=0$ and from that of $\Drev$ if $\msg[i^\ast]=1$.
This is the same as requiring that $\qA$ cannot distinguish $(\gamma^\ast,x^\ast,y^\ast)\gets\Sim(\xk,\iop,i^\ast)$ from that of the following distribution $\Dreal{i^\ast}$.
We can check that $\Dreal{i^\ast}$ is identical with $D$ if $\msg[i^\ast]=0$ and with $\Drev$ if $\msg[i^\ast]=1$.
\begin{description}
\item[$\Dreal{i^\ast}$:]Generate $\gamma\gets\bit$ and $x\gets\Domprf$. Then, if $\gamma=\msg[i^\ast]$, generate $y\gets\Ranprf$, and otherwise, compute $y\gets\Eval(\prfk,x)$. Output $(\gamma,x,y)$.
\end{description}
Essentially, the only attack that $\qA$ can perform is to feed $x^\ast$ contained in the given tuple $(\gamma^\ast,x^\ast,y^\ast)$ to $\tlC$ and compares the result $\tlC(x^\ast)$ with $y^\ast$, if we ensure that $\gamma^\ast$, $x^\ast$ are pseudorandom.
In order to make the construction immune to this attack, letting $\tlC\gets\Mark(\pp,\prfk,\msg)$ and $(\gamma^\ast,x^\ast,y^\ast)\gets\Sim(\xk,\iop,i^\ast)$, we have to design $\Sim$ and $\tlC$ so that
\begin{itemize}
\item If $\gamma=\msg[i^\ast]$, $\tlC(x^\ast)$ outputs a value different from $y^\ast$.
\item If $\gamma\ne\msg[i^\ast]$, $\tlC(x^\ast)$ outputs $y^\ast$.
\end{itemize}
We achieve these conditions as follows.
First, we set $(\gamma^\ast,x^\ast,y^\ast)$ output by $\Sim(\xk,\iop,i^\ast)$ so that $\gamma^\ast$ and $y^\ast$ is random values and $x^\ast$ is an encryption of $y^\ast\|i^\ast\|\gamma^\ast$ by a public-key encryption scheme with pseudorandom ciphertext property, where the encryption key $\pk$ is included in $\iop$.
Then, we set $\tlC$ as a token such that it has the message $\msg$ and the decryption key $\sk$ corresponding to $\pk$ hardwired, and it outputs $y^\ast$ if the input is decryptable and $\gamma^\ast\ne\msg[i^\ast]$ holds for the decrypted $y^\ast\|i^\ast\|\gamma^\ast$, and otherwise behaves as $\Eval(\prfk,\cdot)$.
The actual construction is as follows.

Let $\PRF$ be a PRF family consisting of functions $\{\prf_\prfk(\cdot):\bit^n\ra\bit^\secp|\prfk\}$, where $\secp$ is the security parameter and $n$ is sufficiently large.
Let $\PKE=(\KG,\E,\D)$ be a CCA secure public-key encryption scheme satisfying pseudorandom ciphertext property.
Using these ingredients, We construct an extraction-less watermarking PRF scheme $\ELWMPRF=(\Setup,\Gen,\Eval,\Mark,\Sim)$ as follows.
\begin{description}
\item[$\Setup(1^\secp)$:] In this construction, $\pp:=\bot$ and $\xk:=\bot$.
\item[$\Gen(\pp)$:]  It generates a fresh PRF key $\prfk$ of $\PRF$ and a key pair $(\pk,\sk)\gets\KG(1^\secp)$.
The PRF key is $(\prfk,\sk)$ and the corresponding public tag is $\pk$.
\item[$\Eval((\prfk,\sk),x)$:] It simply outputs $\prf_\prfk(x)$.
\item[$\Mark(\pp,(\prfk,\sk),\msg)$:] It generates the following taken $\tlC[\prfk,\sk,\msg]$.

\begin{figure}[!ht]
\centering{\small
	\fbox{
		\begin{minipage}[t]{0.9\textwidth}
			\textbf{Hard-Coded Constants}: $\prfk,\sk,\msg$.\\
			\textbf{Input:} $x\in\bit^n$. 
			\begin{enumerate}
				\setlength{\parskip}{0.3mm} 
				\setlength{\itemsep}{0.3mm} 
				\item Try to decrypt $y\|i\|\gamma\gets\D(\sk,x)$ with $y\in\bit^\secp$, $i\in[\msglen]$, and $\gamma\in\bit$.
				\item If decryption succeeds, output $y$ if $\gamma\ne\msg[i]$ and $\prf_\prfk(x)$ otherwise.
				\item Otherwise, output $\prf_\prfk(x)$.
			\end{enumerate}
		\end{minipage}
	}}
\end{figure}

\item[$\Sim(\xk,\iop,i)$:]It first generates $\gamma\gets\bit$ and $y\gets\bit^\secp$. Then, it parses $\iop:=\pk$ and generates $x\gets\E(\pk,y\|i\|\gamma)$. Finally, it outputs $(\gamma,x,y)$.

\end{description}

We check that $\ELWMPRF$ satisfies SIM-MDD security.
For simplicity, we fix the message $\msg\in[\msglen]$ embedded into the challenge PRF key.
Then, for any adversary $\qA$ and $i^\ast\in[\msglen]$, SIM-MDD security requires that given $\tlC[\prfk,\sk,\msg]\gets\Mark(\mk,\prfk,\msg)$ and $\iop=\pk$, $\qA$ cannot distinguish $(\gamma^\ast,x^\ast=\E(\pk,y^\ast\|i^\ast\|\gamma^\ast),y^\ast)\gets\Sim(\xk,\iop,i^\ast)$ from an output of $D$ if $\msg[i^\ast]=0$ and is indistinguishable from $\Drev$ if $\msg[i^\ast]=1$.

We consider the case of $\msg[i^\ast]=0$.
We can finish the security analysis by considering the following sequence of mutually indistinguishable hybrid games, where $\qA$ is given $(\gamma^\ast,x^\ast=\E(\pk,y^\ast\|i^\ast\|\gamma^\ast),y^\ast)\gets\Sim(\xk,\iop,i^\ast)$ in the first game, and on the other hand, is given $(\gamma^\ast,x^\ast,y^\ast)\gets D$ in the last game.
We first change the game so that $x^\ast$ is generated as a uniformly random value instead of $x^\ast\gets\E(\pk,y^\ast\|i^\ast\|\gamma^\ast)$ by using the pseudorandom ciphertext property under CCA of $\PKE$.
This is possible since the CCA oracle can simulate access to the marked token $\tlC[\prfk,\sk,\msg]$ by $\qA$.
Then, we further change the security game so that if $\gamma^\ast=1$, $y^\ast$ is generated as $\prf_\prfk(x^\ast)$ instead of a uniformly random value by using the pseudorandomness of $\PRF$.
Note that if $\gamma^\ast=0$, $y^\ast$ remains uniformly at random.
We see that if $\gamma^\ast=1$, the token $\tlC[\prfk,\sk,\msg]$ never evaluate $\prf_\prfk(x^\ast)$ since $\msg[i^\ast]\ne\gamma^\ast$.
Thus, this change is possible.
We see that now the distribution of $(\gamma^\ast,x^\ast,y^\ast)$ is exactly the same as that output by $D$.
Similarly, in the case of $\msg[i^\ast]=1$, we can show that an output of $\Sim(\xk,\iop,i^\ast)$ is indistinguishable from that output by $\Drev$. The only difference is that in the final step, we change the security game so that $y^\ast$ is generated as $\prf_\prfk(x^\ast)$ if $\gamma^\ast=0$.

In the actual public-simulatable construction, we implement this idea using iO and puncturable encryption~\cite{SIAMCOMP:CHNVW18} instead of token and CCA secure public-key encryption.
Also, in the actual secret-simulatable construction, we basically follow the same idea using private constrained PRF and secret-key encryption.

%% file: section/related_work.tex

\ifnum\submission=0
\subsection{More on Related Work}\label{sec:related_work}
\else
\section{More on Related Work}\label{sec:related_work}
\fi

\paragraph{Watermarking against classical adversaries.}
Cohen et al.~\cite{SIAMCOMP:CHNVW18} present a publicly extractable watermarking PRF from IO and injective OWFs. It is unremovable against adversaries who can access the mark oracle only before a target marked circuit is given. The mark oracle returns a marked circuit for a queried \emph{arbitrary} polynomial-size circuit. Suppose we additionally assume the hardness of the decisional Diffie-Hellman or LWE problem. In that case, their watermarking PRF is unremovable against adversaries that can access the mark oracle before and after a target marked circuit is given. However, adversaries can query a valid PRF key to the mark oracle in that case. They also present definitions and constructions of watermarking for public-key cryptographic primitives.

Boneh, Lewi, and Wu~\cite{PKC:BonLewWu17} present a privately extractable watermarking PRF from privately programmable PRFs, which are variants of private constrained PRFs~\cite{PKC:BonLewWu17,EC:CanChe17}. It is unremovable in the presence of the mark oracle. However, it is not secure in the presence of the extraction oracle and does not support public marking. They instantiate a privately programmable PRF with IO and OWFs, but later, Peikert and Shiehian~\cite{PKC:PeiShi18} instantiate it with the LWE assumption.

Kim and Wu~\cite{myJC:KimWu21} (KW17), Quach, Wichs, and Zirdelis~\cite{TCC:QuaWicZir18} (QWZ), and Kim and Wu~\cite{C:KimWu19} (KW19) present privately extractable watermarking PRFs from the LWE assumption. They are secure in the presence of the mark oracle. KW17 construction is not secure in the presence of the extraction oracle and does not support public marking. QWZ construction is unremovable in the presence of the extraction oracle and supports public marking. However, it does not have pseudorandomness against an authority that generates a marking and extraction key. KW19 construction is unremovable in the presence of the extraction oracle and has some restricted pseudorandomness against an authority (see the reference~\cite{C:KimWu19} for the detail). However, it does not support public marking.\footnote{Their construction supports public marking in the random oracle model.}

Yang et al.~\cite{AC:YALXY19} present a collusion-resistant watermarking PRF from IO and the LWE assumption. Collusion-resistant watermarking means unremovability holds even if adversaries receive multiple marked circuits with different embedded messages generated from one target circuit.

Goyal, Kim, Manohar, Waters, and Wu~\cite{C:GKMWW19} improve the definitions of watermarking for public-key cryptographic primitives and present constructions. In particular, they introduce collusion-resistant watermarking and more realistic attack strategies for public-key cryptographic primitives.
Nishimaki~\cite{TCC:Nishimaki20} present a general method for equipping many existing public-key cryptographic schemes with the watermarking functionality.

Goyal, Kim, Waters, and Wu~\cite{myAC:GKWW21} introduce the notion of traceable PRFs, where we can identify a user that creates a pirate copy of her/his authenticated PRF. The difference between traceable PRF and (collusion-resistant) watermarking PRF is that there is only one target original PRF and multiple authenticated copies of it with different identities in traceable PRF. In (collusion-resistant) watermarking PRF, we consider many different PRF keys.
In addition, Goyal et al. introduce a refined attack model. Adversaries in previous watermarking PRF definitions output a pirate PRF circuit that correctly computes the original PRF values for $1/2 +\epsilon$ fraction of inputs. However, adversaries in traceable PRFs output a pirate circuit that \emph{distinguishes} whether an input pair consists of a random input and a real PRF value or a random input and output value. This definition captures wide range of attacks. For example, it captures adversaries who create a pirate PRF circuit that can compute the first quarter bits of the original PRF output. Such an attack is not considered in previous watermarking PRFs. We adopt the refined attack model in our definitions.

\paragraph{Learning information from adversarial entities in the quantum setting.}
Zhandry~\cite{TCC:Zhandry20} introduces the definition of secure traitor tracing against quantum adversaries. In traitor tracing, each legitimate user receives a secret key that can decrypt broadcasted ciphertexts and where identity information is embedded. An adversary outputs a pirate decoder that can distinguish whether an input is a ciphertext of $m_0$ or $m_1$ where $m_0$ and $m_1$ are adversarially chosen plaintexts. A tracing algorithm must identify a malicious user's identity such that its secret decryption key is embedded in the pirate decoder. Thus, we need to extract information from adversarially generated objects. Such a situation also appears in security proofs of interactive proof systems~\cite{SIAMCOMP:Watrous09,EC:Unruh12,FOCS:AmbRosUnr14,myFOCS:CMSZ21} (but not in real cryptographic algorithms) since we rewind a verifier.

Zhandry presents how to estimate the success decryption probability of a \emph{quantum} pirate decoder without destroying the decoding (distinguishing) capability. He achieves a quantum tracing algorithm that extracts a malicious user identity by combining the probability approximation technique above with PLBE~\cite{EC:BonSahWat06}. However, his technique is limited to the setting where user identity spaces are only polynomially large while there are several traitor tracing schemes with exponentially large identity spaces~\cite{EC:NisWicZha16,TCC:GoyKopWat19}. As observed in previous works~\cite{C:GKMWW19,TCC:Nishimaki20,myAC:GKWW21}, traitor tracing and watermarking have similarities since an adversary outputs a pirate circuit in the watermarking setting and an extraction algorithm tries to retrieve information from it. However, a notable difference is that we must consider exponentially large message spaces by default in the (message-embedding) watermarking setting.

\paragraph{Application of (classical) watermarking.}
As we explained above, Aaronson et al.~\cite{C:ALLZZ21}, and Kitagawa et al.~\cite{myTCC:KitNisYam21} achieve secure software leasing schemes by using watermarking.
A leased software consists of a quantum state and watermarked circuit.
Although they use watermarking schemes in the quantum setting, it is sufficient for their purpose to use secure watermarking against adversaries that output a \emph{classical} pirate circuit. This is because a returned software is verified by a checking algorithm and must have a specific format in secure software leasing.\footnote{A valid software must run on a legitimate platform. For example, a video game title of Xbox must run on Xbox.} That is, a returned software is rejected if it does not have a classical circuit part that can be tested by an extraction algorithm of the building block watermarking.


%% file: section/prelim.tex

\ifnum\submission=0
\input{section/notation}
\input{section/quantum_information}
\else\fi

\ifnum\submission=0
\subsection{Measurement Implementation}\label{sec:measurement_implementation}
\else
\section{Measurement Implementation}\label{sec:measurement_implementation}
\fi

\begin{definition}[Projective Implementation]\label{def:projective_implementation}
Let:
\begin{itemize}
 \item $\cP=(\mat{P},\mat{I}-\mat{P})$ be a binary outcome POVM
 \item $D$ be a finite set of distributions over outcomes $\zo{}$
 \item $\cE = \setbk{\mat{E}_D}_{D\in\cD}$ be a projective measurement with index set $\cD$.
 \end{itemize}
 We define the following measurement.
 \begin{enumerate}
 \item Measure under the projective measurement $\cE$ and obtain a distribution $D$ over $\zo{}$.
 \item Output a bit sampled from the distribution $D$.
 \end{enumerate}
 We say this measurement is a projective implementation of $\cP$, denoted by $\projimp(\cP)$ if it is equivalent to $\cP$.
\end{definition}

\begin{theorem}[{\cite[Lemma 1]{TCC:Zhandry20}}]\label{lem:commutative_projective_implementation}
Any binary outcome POVM $\cP=(\mat{P},\mat{I}-\mat{P})$ has a projective implementation $\projimp(\cP)$.
\end{theorem}


\begin{definition}[Shift Distance]\label{def:shift_distance}
For two distributions $D_0,D_1$, the shift distance with parameter $\epsilon$, denoted by $\shiftdis{\epsilon}(D_0,D_1)$, is the smallest quantity $\delta$ such that for all $x \in \R$:
\begin{align}
\Pr[D_0\le x] & \le \Pr[D_1\le x + \epsilon] + \delta,&& \Pr[D_0\ge x]  \le \Pr[D_1\ge x - \epsilon] + \delta,\\
\Pr[D_1\le x] & \le \Pr[D_0\le x + \epsilon] + \delta,&& \Pr[D_1\ge x]  \le \Pr[D_0\ge x - \epsilon] + \delta.
\end{align}
For two real-valued measurements $\cM$ and $\cN$ over the same quantum system, the shift distance between $\cM$ and $\cN$ with parameter $\epsilon$ is
\[
\shiftdis{\epsilon}(\cM,\cN)\seteq \sup_{\ket{\psi}}\shiftdis{\epsilon}(\cM(\ket{\psi}),\cN(\ket{\psi})).
\]
\end{definition}

\begin{definition}[{$(\epsilon,\delta)$-Almost Projective~\cite{TCC:Zhandry20}}]\label{def:almost_projective}
A real-valued quantum measurement $\cM= \setbk{\mat{M}_i}_{i\in\cI}$ is $(\epsilon,\delta)$-almost projective if the following holds. For any quantum state $\ket{\psi}$, we apply $\cM$ twice in a row to $\ket{\psi}$ and obtain measurement outcomes $x$ and $y$, respectively. Then, $\Pr[\abs{x-y}\le \epsilon]\ge 1-\delta$.
\end{definition}

\begin{theorem}[{\cite[Theorem 2]{TCC:Zhandry20}}]\label{thm:api_property}
Let $D$ be any probability distribution and $\cP$ be a collection of projective measurements. For any $0<\epsilon,\delta<1$, there exists an algorithm of measurement $\API_{\cP,\cD}^{\epsilon,\delta}$ that satisfies the following.
\begin{itemize}
\item $\shiftdis{\epsilon}(\API_{\cP,D}^{\epsilon,\delta},\projimp(\cP_D))\le \delta$.
\item $\API_{\cP,D}^{\epsilon,\delta}$ is $(\epsilon,\delta)$-almost projective.
\item The expected running time of $\API_{\cP,D}^{\epsilon,\delta}$ is $T_{\cP,D}\cdot \poly(1/\epsilon,\log(1/\delta))$ where $T_{\cP,D}$ is the combined running time of $D$, the procedure mapping $i \ra (\mat{P}_i,\mat{I}- \mat{P}_i)$, and the running time of measurement $(\mat{P}_i,\mat{I}-\mat{P}_i)$.
\end{itemize}
\end{theorem}


\begin{theorem}[{\cite[Corollary 1]{TCC:Zhandry20}}]\label{cor:cind_sample_api}
Let $\qstateq$ be an efficiently constructible, potentially mixed state, and $D_0,D_1$ efficiently sampleable distributions.
If $D_0$ and $D_1$ are computationally indistinguishable, for any inverse polynomial $\epsilon$ and any function $\delta$, we have $\shiftdis{3\epsilon}(\API_{\cP,D_0}^{\epsilon,\delta},\API_{\cP,D_1}^{\epsilon,\delta}) \le 2\delta + \negl(\secp)$.
\end{theorem}

Note that the indistinguishability of $D_0$ and $D_1$ needs to hold against distinguishers who can construct $\qstateq$ in the theorem above. However, this fact is not explicitly stated in \cite{TCC:Zhandry20}.
We need to care about this condition if we need secret information to construct $\qstateq$, and the secret information is also needed to sample an output from $D_0$ or $D_1$.
We handle such a situation when analyzing the unremovability of our privately extractable watermarking PRF.
In that situation, we need a secret extraction key to construct $\qstateq$ and sample an output from $D_0$ and $D_1$.

\medskip

We also define the notion of the reverse almost projective property of API.

\begin{definition}[$(\epsilon,\delta)$-Reverse Almost Projective]\label{def:rev_almost_projective}
Let $\cP=\setbk{(\Pi_i,\mat{I} -\Pi_i)}_i$ be a collection of binary outcome projective measurements.
Let $D$ be a distribution.
We also let $\cPrev=\setbk{(\mat{I} -\Pi_i,\Pi_i)}_i$.
We say $\API$ is $(\epsilon,\delta)$-reverse almost projective if the following holds. For any quantum state $\ket{\psi}$, we apply $ \API_{\cP,D}^{\epsilon ,\delta}$ and $\API_{\cPrev,D}^{\epsilon,\delta}$ in a row to $\ket{\psi}$ and obtain measurement outcomes $x$ and $y$, respectively. Then, $\Pr[\abs{(1-x)-y}\le \epsilon]\ge 1-\delta$.
\end{definition}

\ifnum\submission=0
\input{section/proof_rev_almost_projective}

\else
We show that the measurement algorithm$\API_{\cP,D}^{\epsilon,\delta}$ in~\cref{thm:api_property} also satisfies~\cref{def:rev_almost_projective} in \cref{sec:proof_rev_almost_projective}
\fi

\ifnum\submission=0
\input{section/crypto_tools}

\else\fi

%% file: section/notation.tex

\section{Preliminaries}\label{sec:prelim}

\paragraph{Notations and conventions.}
In this paper, standard math or sans serif font stands for classical algorithms (e.g., $C$ or $\algo{Gen}$) and classical variables (e.g., $x$ or $\keys{pk}$).
Calligraphic font stands for quantum algorithms (e.g., $\qalgo{Gen}$) and calligraphic font and/or the bracket notation for (mixed) quantum states (e.g., $\qstateq$ or $\ket{\psi}$).
For strings $x$ and $y$, $x \concat y$ denotes the concatenation of $x$ and $y$.
Let $[\ell]$ denote the set of integers $\{1, \cdots, \ell \}$, $\secp$ denote a security parameter, and $y \seteq z$ denote that $y$ is set, defined, or substituted by $z$.

In this paper, for a finite set $X$ and a distribution $D$, $x \chosen X$ denotes selecting an element from $X$ uniformly at random, $x \chosen D$ denotes sampling an element $x$ according to $D$. Let $y \gets \algo{A}(x)$ and $y \gets \qalgo{A}(\qstate{x})$ denote assigning to $y$ the output of a probabilistic or deterministic algorithm $\algo{A}$ and a quantum algorithm $\qalgo{A}$ on an input $x$ and $\qstate{x}$, respectively. When we explicitly show that $\algo{A}$ uses randomness $r$, we write $y \gets \algo{A}(x;r)$.
PPT and QPT algorithms stand for probabilistic polynomial time algorithms and polynomial time quantum algorithms, respectively.
Let $\negl$ denote a negligible function.

%% file: section/quantum_information.tex
\subsection{Quantum Information}\label{sec:quantum_info}

Let $\cH$ be a finite-dimensional complex Hilbert space. A (pure) quantum state is a vector $\ket{\psi}\in \cH$.
Let $\cS(\cH)$ be the space of Hermitian operators on $\cH$. A density matrix is a Hermitian operator $\qstate{X} \in \cS(\cH)$ with $\Trace(\qstate{X})=1$, which is a probabilistic mixture of pure states.
A quantum state over $\cH=\bbC^2$ is called qubit, which can be represented by the linear combination of the standard basis $\setbk{\ket{0},\ket{1}}$. More generally, a quantum system over $(\bbC^2)^{\tensor n}$ is called an $n$-qubit quantum system for $n \in \bbN \setminus \setbk{0}$.

A Hilbert space is divided into registers $\cH= \cH^{\qreg{R}_1} \tensor \cH^{\qreg{R}_2} \tensor \cdots \tensor \cH^{\qreg{R}_n}$.
We sometimes write $\qstate{X}^{\qreg{R}_i}$ to emphasize that the operator $\qstate{X}$ acts on register $\cH^{\qreg{R}_i}$.\footnote{The superscript parts are gray colored.}
When we apply $\qstate{X}^{\qreg{R}_1}$ to registers $\cH^{\qreg{R}_1}$ and $\cH^{\qreg{R}_2}$, $\qstate{X}^{\qreg{R}_1}$ is identified with $\qstate{X}^{\qreg{R}_1} \tensor \mat{I}^{\qreg{R}_2}$.

A unitary operation is represented by a complex matrix $\mat{U}$ such that $\mat{U}\mat{U}^\dagger = \mat{I}$. The operation $\mat{U}$ transforms $\ket{\psi}$ and $\qstate{X}$ into $\mat{U}\ket{\psi}$ and $\mat{U}\qstate{X}\mat{U}^\dagger$, respectively.
A projector $\mat{P}$ is a Hermitian operator ($\mat{P}^\dagger =\mat{P}$) such that $\mat{P}^2 = \mat{P}$.

For a quantum state $\qstate{X}$ over two registers $\cH^{\qreg{R}_1}$ and $\cH^{\qreg{R}_2}$, we denote the state in $\cH^{\qreg{R}_1}$ as $\qstate{X}[\qreg{R}_1]$, where $\qstate{X}[\qreg{R}_1]= \Trace_2[\qstate{X}]$ is a partial trace of $\qstate{X}$ (trace out $\qreg{R}_2$).

Given a function $F: X\ra Y$, a quantum-accessible oracle $O$ of $F$ is modeled by a unitary transformation $\mat{U}_F$ operating on two registers $\cH^{\qreg{in}}$ and $\cH^{\qreg{out}}$, in which $\ket{x}\ket{y}$ is mapped to $\ket{x}\ket{y\oplus F(x)}$, where $\oplus$ denotes XOR group operation on $Y$.
We write $\qA^{\ket{O}}$ to denote that the algorithm $\qA$'s oracle $O$ is a quantum-accessible oracle.

\ifnum\submission=0
\begin{definition}[Quantum Program with Classical Inputs and Outputs~\cite{C:ALLZZ21}]\label{def:Q_program_C_IO}
A quantum program with classical inputs is a pair of quantum state $\qstateq$ and unitaries $\setbk{\mat{U}_x}_{x\in[N]}$ where $[N]$ is the domain, such that the state of the program evaluated on input $x$ is equal to $\mat{U}_x \qstateq \mat{U}_x^\dagger$. We measure the first register of $\mat{U}_x \qstateq \mat{U}_x^\dagger$ to obtain an output. We say that $\setbk{\mat{U}_x}_{x\in[N]}$ has a compact classical description $\mat{U}$ when applying $\mat{U}_x$ can be efficiently computed given $\mat{U}$ and $x$.
\end{definition}
\else\fi

\begin{definition}[Positive Operator-Valued Measure]\label{def:POVM}
Let $\cI$ be a finite index set. A positive operator valued measure (POVM) $\cM$ is a collection $\setbk{\mat{M}_i}_{i\in\cI}$ of Hermitian positive semi-define matrices $\mat{M}_i$ such that $\sum_{i\in \cI}\mat{M}_i = \mat{I}$. When we apply POVM $\cM$ to a quantum state $\qstate{X}$, the measurement outcome is $i$ with probability $p_i =\Trace(\qstate{X}\mat{M}_i)$.
We denote by $\cM(\ket{\psi})$ the distribution obtained by applying $\cM$ to $\ket{\psi}$.
\end{definition}

\begin{definition}[Quantum Measurement]\label{def:quantum_measurement}
A quantum measurement $\cE$ is a collection $\setbk{\mat{E}_i}_{i\in\cI}$ of matrices $\mat{E}_i$ such that $\sum_{i\in\cI}\mat{E}_i^\dagger \mat{E}_i=\mat{I}$.
When we apply $\cE$ to a quantum state $\qstate{X}$, the measurement outcome is $i$ with probability $p_i =\Trace(\qstate{X}\mat{E}_i^\dagger \mat{E}_i)$. Conditioned on the outcome being $i$, the post-measurement state is $\mat{E}_i \qstate{X} \mat{E}_i^\dagger/p_i$.
\end{definition}
We can construct a POVM $\cM$ from any quantum measurement $\cE$ by setting $\mat{M}_i \seteq \mat{E}_i^\dagger \mat{E}_i$.
We say that $\cE$ is an implementation of $\cM$. The implementation of a POVM may not be unique.

\begin{definition}[Projective Measurement/POVM]\label{def:projective_measurement}
A quantum measurement $\cE=\setbk{\mat{E}_i}_{i\in \cI}$ is projective if for all $i \in \cI$, $\mat{E}_i$ is a projector.
This implies that $\mat{E}_i\mat{E}_j = \mat{0}$ for distinct $i,j\in \cI$.
In particular, two-outcome projective measurement is called a binary projective measurement, and is written as $\cE=(\mat{P},\mat{I}-\mat{P})$, where $\mat{P}$ is associated with the outcome $1$, and $\mat{I}-\mat{P}$ with the outcome $0$.
Similarly, a POVM $\cM$ is projective if for all $i\in \cI$, $\mat{M}_i$ is a projector. This also implies that $\mat{M}_i\mat{M}_j = \mat{0}$ for distinct $i,j\in \cI$.
\end{definition}

\begin{definition}[Controlled Projection]\label{def:controlled_projection}
Let $\cP = \setbk{\cM_i}_{i\in\cI}$ be a collection of projective measurement over a Hilbert space $\cH$, where $\cM_i = (\Pi_i,\mat{I}- \Pi_i)$ for $i\in \cI$.
Let $D$ be a distribution whose randomness space is $\cR$. The controlled projection $\cproj_{\cP,D} =(\cproj_{\cP,D}^1, \cproj_{\cP,D}^0)$ is defined as follows.\footnote{We use superscript $b$ to denote that it is associated with the outcome $b$ here.}
\begin{align}
\cproj_{\cP,D}^1 &\seteq \sum_{r\in \cR}\ket{r}\bra{r}\tensor \Pi_{D(r)}, && \cproj_{\cP,D}^0 \seteq \sum_{r\in \cR}\ket{r}\bra{r}\tensor (\mat{I}-\Pi_{D(r)})
\end{align}
\end{definition}

%% file: section/proof_rev_almost_projective.tex

\ifnum\submission=1
\section{The Reverse Almost Projective Property of API}\label{sec:proof_rev_almost_projective}
\else\fi

We show that the measurement algorithm $\API_{\cP,D}^{\epsilon,\delta}$ in~\cref{thm:api_property} also satisfies~\cref{def:rev_almost_projective}.
First, we describe the detail of $\API_{\cP,D}^{\epsilon,\delta}$ in~\cref{fig:API_algo}.
$\API$ uses an ancilla register $\cH^{\qreg{R}}$ besides the original Hilbert space $\cH^{\qreg{H}}$.
Let $\randspace$ be the randomness space of distribution $D$. We define $\IsU \seteq (\kbisu,\mat{I} -\kbisu)$ where
\[
\ketisu \seteq \frac{1}{\sqrt{\abs{\randspace}}}\sum_{r\in \randspace}\ket{r}.
\]


\protocol
{Algorithm $\API_{\cP,D}^{\epsilon,\delta}$}
{The description of $\API$.}
{fig:API_algo}
{
\begin{description}
\setlength{\parskip}{0.3mm} 
\setlength{\itemsep}{0.3mm} 
\item[Parameter:] Collection of projective measurement $\cP$, distribution $D$, real values $\epsilon,\delta$.
\item[Input:] A quantum state $\ket{\psi}$.
\end{description}
\begin{enumerate}
\item Initialize a state $\ketisu \ket{\psi}$.
\item Initialize a classical list $L \seteq (1)$.
\item Repeat the following $T \seteq \ceil{\frac{\ln{4/\delta}}{\epsilon^2}}$.
\begin{enumerate}
\item Apply $\cproj_{\cP,D}$ to register $\cH^{\qreg{R}} \tensor \cH^{\qreg{H}}$. Let $b_{2i-1}$ be the measurement outcome and set $L \seteq (L,b_{2i-1})$.
\item Apply $\IsU$ to register $\cH^{\qreg{R}}$. Let $b_{2i}$ be the measurement outcome and set $L \seteq (L,b_{2i})$.
\end{enumerate}
\item Let $t$ be the number of index $i$ such that $b_{i-1} = b_i$ in the list $L =(0,b_1,\ldots,b_{2T})$, and $\tlp \seteq t/2T$.
\item If $b_{2T}=0$, repeat the loop again until $b_{2i}=1$.
\item Discard $\cH^{\qreg{R}}$ register, and output $\tlp$.
\end{enumerate}
}




We use the following lemma to analyze $\API_{\cP,D}^{\epsilon,\delta}$.

\begin{lemma}[\cite{Jordan75}]\label{lem:jordan}
For any two Hermitian projectors $\Pi_v$ and $\Pi_w$ on a Hilbelt space $\cH$, there exists an orthogonal decomposition of $\cH$ into one-dimensional and two-dimensional subspaces (the Jordan subspaces) that are invariant under both $\Pi_v$ and $\Pi_w$. Moreover:
\begin{itemize}
\item in each one-dimensional space, $\Pi_v$ and $\Pi_w$ act as identity or rank-zero projectors; and
\item in each two-dimensional subspace $S_j$, $\Pi_v$ and $\Pi_w$ are rank-one projectors: there exists $\ket{v_j},\ket{w_j}\in S_j$ such that $\Pi_v$ projects onto $\ket{v_j}$ and $\Pi_w$ projects onto $\ket{w_j}$.
\end{itemize}
\end{lemma}

For each two-dimensional subspace $S_j$, we call $p_j \seteq \abs{\bra{v_j}\ket{w_j}}^2$ the eigenvalue of the $j$-th subspace. It is easy to see that $\ket{v_j}$ is an eigenvector of the Hermitian matrix $\Pi_v\Pi_w\Pi_v$ with eigenvalue $p_j$.

As previous works observed~\cite{CC:MarWat05,LecNot:Regev06,myFOCS:CMSZ21}, we obtain the following by~\cref{lem:jordan}. There exists orthogonal vectors $\ket{v_j},\ket{v_j^\perp}$ that span $S_j$, such that $\Pi_v \ket{v_j}=\ket{v_j}$ and $\Pi_v \ket{v_j^\perp}=0$. Similarly, $\Pi_w \ket{w_j}=\ket{w_j}$ and $\Pi_w \ket{w_j^\perp}=0$. By setting appropriate phases, we have
\begin{align}
&& \ket{w_j} &= \sqrt{p_j}\ket{v_j} + \sqrt{1-p_j}\ket{v_j^\perp}, & \ket{w_j^\perp} &= \sqrt{1- p_j}\ket{v_j} - \sqrt{p_j}\ket{v_j^\perp},\\
&& \ket{v_j} &= \sqrt{p_j}\ket{w_j} + \sqrt{1-p_j}\ket{w_j^\perp}, & \ket{v_j^\perp} &= \sqrt{1- p_j}\ket{w_j} - \sqrt{p_j}\ket{w_j^\perp},
\end{align}
where $\ket{v_j^\perp},\ket{w_j^\perp}\in S_j$ such that $\bra{v_j}\ket{v_j^\perp}=\bra{w_j}\ket{w_j^\perp}=0$.
We also have
\begin{align}
\Pi_v \ket{w_j} &=\sqrt{p_j}\ket{v_j}, && \Pi_w \ket{v_j} =\sqrt{p_j}\ket{w_j}.
\end{align}


\usetikzlibrary{matrix}
\begin{figure}[!t]
\centering{
\begin{tikzpicture}
  \matrix (m) [matrix of math nodes,row sep=3em,column sep=5em,minimum width=4em]
  {
     \ket{v_j} & \ket{w_j} & \ket{v_j} & \ket{w_j} & \cdots\\
     \ket{v_j^\perp} & \ket{w_j^\perp} & \ket{v_j^\perp} & \ket{w_j^\perp}& \cdots\\};
  \path[-stealth]
    (m-1-1) edge [double] node [above] {$p_j$} (m-1-2)
    (m-1-1) edge [dashed,double] node [above] {$1-p_j$} (m-2-2)

    (m-2-2) edge node [above] {$1-p_j$} (m-1-3)
    (m-2-2) edge [dashed] node [below] {$p_j$} (m-2-3)
    (m-1-2) edge [dashed] node [above] {} (m-2-3)
    (m-1-2) edge node [above] {$p_j$} (m-1-3)

    (m-2-3) edge [double] node [above] {$1-p_j$} (m-1-4)
    (m-2-3) edge [dashed,double] node [below] {$p_j$} (m-2-4)
    (m-1-3) edge [dashed,double] node [above] {} (m-2-4)
    (m-1-3) edge [double] node [above] {$p_j$} (m-1-4);
\end{tikzpicture}
}
\caption{Solid lines denote that the measurement outcome is $1$. Dashed line denote that the measurement outcome is $0$. Double lines denote we apply $\cproj_{\cP,D} =(\cproj_{\cP,D}^{1},\cproj_{\cP,D}^{0})$. Single lines denote we apply $\IsU = (\kbisu,\mat{I}-\kbisu)$.}
\label{fig:jordan_bernoulli}
\end{figure}

\begin{figure}[!th]
\centering{
\begin{tikzpicture}
  \matrix (m) [matrix of math nodes,row sep=3em,column sep=5em,minimum width=4em]
  {
     \ket{v_j} & \ket{w_j^\perp} & \ket{v_j} & \ket{w_j^\perp} &  \cdots\\
     \ket{v_j^\perp} & \ket{w_j} & \ket{v_j^\perp} & \ket{w_j}& \cdots\\};
  \path[-stealth]
    (m-1-1) edge [double] node [above] {$1-p_j$} (m-1-2)
    (m-1-1) edge [dashed,double] node [above] {$p_j$} (m-2-2)

    (m-2-2) edge node [above] {$p_j$} (m-1-3)
    (m-2-2) edge [dashed] node [below] {$1-p_j$} (m-2-3)
    (m-1-2) edge [dashed] node [above] {} (m-2-3)
    (m-1-2) edge node [above] {$1-p_j$} (m-1-3)

    (m-2-3) edge [double] node [above] {$p_j$} (m-1-4)
    (m-2-3) edge [dashed,double] node [below] {$1-p_j$} (m-2-4)
    (m-1-3) edge [dashed,double] node [above] {} (m-2-4)
    (m-1-3) edge [double] node [above] {$1-p_j$} (m-1-4);
\end{tikzpicture}
}
\caption{Solid lines denote that the measurement outcome is $1$. Dashed line denote that the measurement outcome is $0$. Double lines denote we apply $\cproj_{\cP,D}^{\mathtt{rev}}=(\cproj_{\cPrev,D}^{1},\cproj_{\cPrev,D}^{0}) =(\cproj_{\cP,D}^{0},\cproj_{\cP,D}^{1})$. Single lines denote we apply $\IsU = (\kbisu,\mat{I}-\kbisu)$.}
\label{fig:jordan_bernoulli_rev}
\end{figure}

\begin{theorem}\label{thm:API_is_rev_almost_projective}
$\API_{\cP,D}^{\epsilon,\delta}$ in~\cref{fig:API_algo} is $(\epsilon,\delta)$-reverse almost projective.
\end{theorem}

\begin{proof}[Proof of~\cref{thm:API_is_rev_almost_projective}]
To analyze $\API_{\cP,D}^{\epsilon,\delta}$, 
we set $\Pi_w \seteq \cproj_{\cP,D}^1$, $\Pi_v \seteq \kbisu\tensor \mat{I}$, and apply~\cref{lem:jordan}.
Then, we have the following relationships:
\begin{align}
&& \cproj_{\cP,D}^{1} \ket{v_j} &= \sqrt{p_j}\ket{w_j}, & (\kbisu) \ket{w_j} &= \sqrt{p_j} \ket{v_j} & \\
&& \cproj_{\cP,D}^{0} \ket{v_j} &= \sqrt{1- p_j}\ket{w_j^\perp}, & (\mat{I} - \kbisu) \ket{w_j} &= \sqrt{1- p_j} \ket{v_j^\perp} &\\
&& \cproj_{\cP,D}^{1} \ket{v_j^\perp} &= \sqrt{1- p_j}\ket{w_j}, & (\kbisu) \ket{w_j^\perp} &= \sqrt{1-p_j} \ket{v_j} & \\
&& \cproj_{\cP,D}^{0} \ket{v_j^\perp} &= -\sqrt{p_j}\ket{w_j^\perp}, & (\mat{I} - \kbisu) \ket{w_j^\perp} &= -\sqrt{p_j} \ket{v_j^\perp} &,
\end{align}
where $\ket{w_j}$ and $\ket{v_j}$ are decompositions of $\Pi_w$ and $\Pi_v$, and $p_j=\abs{\bra{v_j}\ket{w_j}}^2$.

Suppose we apply $\API_{\cP,D}^{\epsilon,\delta}$ to a $\ket{\psi}$ on $\cH^{\qreg{H}}$.
We can write the initial state in~\cref{fig:API_algo} as $\ketisu\ket{\psi} = \sum_j \alpha_j \ket{v_j} +\alpha_j^\perp \ket{v_j^\perp}$ since $\setbk{\ket{v_j},\ket{v_j^\perp}}_j$ is a basis. We also have that $\kbisu^{\qreg{R}} \tensor \mat{I} =\sum_j \ket{v_j}\bra{v_j}$ since $\Pi_v$ projects onto $\ket{v_j}$ in each decomposed subspace $S_j$. It is easy to see that $(\kbisu \tensor\mat{I}) \ketisu \ket{\psi} =\ketisu\ket{\psi}$. Thus, for all $j$, $\alpha_j^\perp =0$.
Therefore, for any $\ket{\psi}$, we can write $\ketisu\ket{\psi} = \sum_j \alpha_j \ket{v_j}$.
As we see in~\cref{fig:API_algo}, when we run $\API_{\cP,D}^{\epsilon,\delta}(\ket{\psi})$, the initial state is $\ketisu \ket{\psi}= \sum_j \alpha_j \ket{v_j}$ and we apply $\cproj_{\cP,D}$ and $\IsU$ alternately.
Therefore, the quantum state $\ket{v_j}$ in each decomposed subspace $S_j$ changes as in~\cref{fig:jordan_bernoulli} when we run $\API_{\cP,D}^{\epsilon,\delta}(\ket{\psi})$.
%

Next, suppose we apply $\API_{\cPrev,D}^{\epsilon,\delta}$ to the quantum state immediately after applying $\API_{\cP,D}^{\epsilon,\delta}$ to $\ket{\psi}$.
$\API_{\cP,D}^{\epsilon,\delta}$ ensures that the final measurement is $\IsU$ and its result is $1$.
This means that the state going into the main loop of $\API_{\cPrev,D}^{\epsilon,\delta}$ (the third item in~\cref{fig:API_algo}) is identical to the state before $\cH^{\qreg{R}}$ is discarded at the application of $\API_{\cP,D}^{\epsilon,\delta}$ to $\ket{\psi}$.
By the definition of $\cPrev$, we have $\cproj_{\cP,D}^{b} = \cproj_{\cPrev,D}^{1-b}$ for $b\in\zo{}$.
Thus, the quantum state $\ket{v_j}$ in each decomposed subspace $S_j$ changes as in~\cref{fig:jordan_bernoulli_rev} when we apply $\API_{\cPrev,D}^{\epsilon,\delta}$.

From the above discussions, we can view a successive execution of $\API_{\cP,D}^{\epsilon,\delta}$ and $\API_{\cPrev,D}^{\epsilon,\delta}$ to $\ket{\psi}$ as the following single experiment.

\begin{itemize}
\item Sample $p_j$ from $\{p_j\}_j$ with the probability $\alpha_j^2$.
\item Flip $2T$ biased random coins whose probability of outputting $1$ is $p_j$.
\item Flip an even number of additional random coins until $0$ is found.
\item Flip $2T$ biased random coins whose probability of outputting $1$ is $1-p_j$.
\item Let $K$ be the overall list of coin flips.
\end{itemize}

Let $\tlp_x$ and $\tlp_y$ be the outcome of $\API_{\cP,D}^{\epsilon,\delta}$ and $\API_{\cPrev,D}^{\epsilon,\delta}$, respectively.
$\tlp_x$ is the fraction of $1$'s in the first $2T$ bits of $K$.
Also, $\tlp_y$ is the fraction of $1$'s in the last $2T$ bits of $K$.
Then, we have
\begin{align}
\Pr[\abs{\tlp_x - p_j} \ge \epsilon/2] &\le \delta/2 \label{eq:projective_bound}\\
\Pr[\abs{\tlp_y - (1-p_j)} \ge \epsilon/2] &\le \delta/2 \label{eq:reverse_projective_bound}.
\end{align}
It is easy to see that \cref{eq:reverse_projective_bound} is equivalent to
$\Pr[\abs{(1-\tlp_y) -p_j} \ge \epsilon/2] \le \delta/2$ due to $\abs{a} = \abs{-a}$. Therefore, by combining it with \cref{eq:projective_bound}, we obtain
\begin{align}
\Pr[\abs{\tlp_x - (1-\tlp_y)} \ge \epsilon] & \le \delta.
\end{align}


This completes the proof.
\end{proof}

%% file: section/crypto_tools.tex

\subsection{Cryptographic Tools}\label{sec:crypto_tools}

\begin{definition}[Learning with Errors]\label{def:LWE}
Let $n,m,q \in \N$ be integer functions of the security parameter $\secp$. Let $\chi = \chi(\secp)$ be an error distribution over $\Z$.
The LWE problem $\LWE_{n,m,q,\chi}$ is to distinguish the following two distributions.
\[
D_0 \seteq \setbracket{(\mm{A},\mv{s}^{\intercal}\mm{A}+\mv{e}) \mid \mm{A} \chosen \Zq^{n\times m}, \mv{s}\chosen \Zq^{n}, \mv{e}\chosen \chi^m}  \text{ and } D_1 \seteq \setbracket{(\mm{A},\mv{u}) \mid \mm{A} \chosen \Zq^{n\times m}, \mv{u}\chosen \Zq^m}.
\]

When we say we assume the quantum hardness of the LWE problem or the QLWE assumption holds, we assume that for any QPT adversary $\qalgo{A}$, it holds that
\[
\abs{ \Pr[\qA (D_0)\out  1 ] - \Pr[\qA (D_1)\out 1]} \le \negl(\secp).
\]
\end{definition}


\begin{definition}[Pseudorandom Generator]\label{def:prg}
A pseudorandom generator (PRG) $\PRG: \zo{\secp} \ra \zo{\secp + \ell(\secp)}$ with stretch $\ell(\secp)$ ($\ell$ is some polynomial function) is a polynomial-time computable function that satisfies the following.
For any QPT adversary $\qA$, it holds that
\[\abs{\Pr[\qA(\PRG(s))=1 \mid s \chosen \cU_\secp]- \Pr[\qA(r)\mid r \chosen \cU_{\secp+\ell(\secp)}]}\le \negl(\secp),\]
where $\cU_{m}$ denotes the uniform distribution over $\zo{m}$.
\end{definition}

\begin{theorem}[\cite{SIAMCOMP:HILL99}]\label{thm:owf_prg}
If there exists a OWF, there exists a PRG.
\end{theorem}

\begin{definition}[Quantum-Accessible Pseudo-Random Function]\label{def:prf}
Let $\{\PRF_{K}: \bin^{\ell_1} \ra \allowbreak \bin^{\ell_2} \mid K \in \bin^\secp\}$ be a family of polynomially computable functions, where $\ell_1$ and $\ell_2$ are some polynomials of $\secp$.
We say that $\PRF$ is a quantum-accessible pseudo-random function (QPRF) family if for any QPT adversary $\qA$, it holds that
\begin{align}
\advt{\qA}{prf}(\secp)
= \abs{\Pr[\qA^{\ket{\PRF_{K}(\cdot)}}(1^\secp) \out 1 \mid K \gets \bit^{\secp}]
-\Pr[\qA^{\ket{\Rand(\cdot)}}(1^\secp) \out 1 \mid \Rand \gets \cU]
}\leq\negl(\secp),
\end{align}
where $\cU$ is the set of all functions from $\bit^{\ell_1}$ to $\bit^{\ell_2}$. 
\end{definition}

\begin{theorem}[\cite{FOCS:Zhandry12}]\label{thm:qprf}
If there exists a OWF, there exists a QPRF.
\end{theorem}

\begin{definition}[Puncturable PRF]\label{def:pprf}
A puncturable PRF (PPRF) is a tuple of algorithms $\PuncPRF = (\prfgen, \prf,\Puncture)$ where $\{\prf_{K}: \bin^{\ell_1} \ra \zo{\ell_2} \mid K \in \zo{\secp}\}$ is a PRF family and satisfies the following two conditions. Note that $\ell_1$ and $\ell_2$ are polynomials of $\secp$.
    \begin{description}
        \item[Punctured correctness:] For any polynomial-size set $S \subseteq \zo{\ell_1}$ and any $x\in \zo{\ell_1} \setminus S$, it holds that
        \begin{align}
        \Pr[\prf_{K}(x) = \prf_{K_{\notin S}}(x)  \mid K \gets \prfgen(1^{\secp}),
        K_{\notin S} \gets \Puncture(K,S)]=1.
        \end{align}
        \item[Pseudorandom at punctured point:] For any polynomial-size set $S \subseteq\zo{\ell_1}$
        and any QPT distinguisher $\qA$, it holds that
        \begin{align}
        \vert
        \Pr[\qA(\prf_{K_{\notin S}},\{\prf_{K}(x_i)\}_{x_i\in S}) \out 1] -
        \Pr[\qA(\prf_{K_{\notin S}}, (\cU_{\ell_2})^{\abs{S}}) \out 1]
        \vert \leq \negl(\secp),
        \end{align}
        where $K\gets \prfgen(1^{\secp})$,
        $K_{\notin S} \gets \Puncture(K,S)$ and $\cU_{\ell_2}$ denotes the uniform distribution over $\zo{\ell_2}$.
    \end{description}
    If $S = \setbk{x^\ast}$ (i.e., puncturing a single point), we simply write $\prf_{\ne x^\ast}(\cdot)$ instead of $\prf_{K_{\notin S}}(\cdot)$.
\end{definition}


It is easy to see that the Goldwasser-Goldreich-Micali tree-based construction of PRFs (GGM PRF)~\cite{JACM:GolGolMic86} from \rmOWF yield puncturable PRFs where the size of the punctured key grows polynomially with the size of the set $S$ being punctured~\cite{AC:BonWat13,PKC:BoyGolIva14,CCS:KPTZ13}. Thus, we have:
\begin{theorem}[\cite{JACM:GolGolMic86,AC:BonWat13,PKC:BoyGolIva14,CCS:KPTZ13}]\label{thm:pprf-owf} If OWFs exist, then for any polynomials $\ell_1(\secp)$ and $\ell_2(\secp)$, there exists a PPRF that maps $\ell_1$-bits to $\ell_2$-bits.
\end{theorem}

\begin{definition}[SKE]\label{def:ske_syntax}
An SKE scheme with plaintext space $\Ps = \setbk{\Ps_\secp}_{\secp \in \bbN}$ and ciphertext space $\Cs=\setbk{\Cs_{\secp}}_{\secp \in \bbN}$, where $\Cs_{\secp} \subseteq \zo{\ctlen}$ for some $\ctlen=\ctlen(\secp)$, is a tuple of three algorithms.
\begin{description}
\item[$\Gen(1^\secp) \ra k$:] The key generation algorithm takes as input the security parameter $\secp$, and outputs an encryption key $k$.
\item[$\Enc(k,m) \ra \ct$:] The encryption algorithm takes as input $k$ and a plaintext $m \in \Ps_{\secp}$, and outputs a ciphertext $\ct \in \Cs_{\secp}$.
\item[$\Dec(k,\ct) \ra m^\prime$:] The decryption algorithm takes as input $k$ and $\ct \in \Cs_{\secp}$, and outputs a plaintext $m^\prime \in \Ps_{\secp} \cup \setbk{\bot}$.
\item[Correctness:] An SKE scheme is correct if for all $\secp\in \bbN$ and $m\in \Ps_{\secp}$,
\[
\Pr[\Dec(k,\ct)=m \mid k \lrun \Gen(1^\secp,1^\kappa),\ct \lrun \Enc(k,m)]=1.
\]

\item[Sparseness:]
In this work, we also require that most strings are not valid ciphertexts under a randomly generated key of an SKE scheme:
\begin{align}
\Pr \left [
\Dec(k, c) \neq \bot ~\left|~
k \lrun \Gen(1^\secp), c \chosen \zo{\ctlen} \right.
\right ] \le \negl(\secp).
\end{align}

\end{description}
\end{definition}

\begin{definition}[Ciphertext Pseudorandomness for SKE]\label{def:ske_pseudorandomct}
An SKE scheme satisfies ciphertext pseudorandomness if for any (stateful) QPT $\qA$, it holds that
\[
2\abs{\Pr\left[
\qA^{\Enc(k,\cdot)}(\ct_b)=b
 \ \middle |
\begin{array}{rl}
 &1^\kappa \lrun \qA(1^\secp), k\lrun \Gen(1^\secp,1^\kappa),\\
 &m \lrun \qA^{\Enc(k,\cdot)}, b \chosen \zo{}, \\
 &\ct_0 \lrun \Enc(k,m), \ct_1 \chosen \zo{\ctlen}
\end{array}
\right] -\frac{1}{2}} \le  \negl(\secp).
\]
\end{definition}

\begin{theorem}\label{thm:pseudorandom_ske}
If OWFs exist, there exists an SKE scheme with sparseness and ciphertext pseudorandomness.
\end{theorem}
The well-known PRF-based SKE satisfies ciphertext pseudorandomness. However, we need padding for sparseness.
That is, a ciphertext is $(r,\PRF_k (r)\xor 0^{\ell -\ell_\mu}\concat m)$ where $r \in \zo{n}$ is randomness of encryption, $k$ is a PRF key, $\PRF: \zo{n}\ra \zo{\ell}$ is a PRF, and $\abs{m}=\ell_\mu$. We check that the first $\ell-\ell_\mu$ bits of $m^\prime = \Dec(k,(c_1,c_2))$ equals to $0^{\ell-\ell_\mu}$. If $\ell$ is sufficiently long, the scheme has sparseness.


%
%

\begin{definition}[Constrained PRF (Syntax)]\label{def:cprf_syntax}
A constrained PRF (CPRF) with domain $\Domprf$, range $\Ranprf$, and constraint family $\Fs = \setbk{\Fs_{\secp,\kappa}}_{\secp,\kappa\in \bbN}$ where $\Fs_{\secp,\kappa}= \setbk{f \colon \Domprf \ra \zo{}}$ is a tuple of four algorithms.
\begin{description}
\item[$\Setup(1^\secp,1^\kappa) \ra \msk$:] The setup algorithm takes as input the security parameter $\secp$ and a constraint-family parameter $\kappa$, and outputs a master PRF key $\msk$.
\item[$\constrain(\msk,f)\ra \sk_f$:] The constrain algorithm takes as input $\secp$ and a constraint $f\in \Fs_{\secp,\kappa}$, and outputs a constrained key $\sk_f$.
\item[$\Eval(\msk,x) \ra y$:] The evaluation algorithm takes as input $\msk$ and an input $x \in \Domprf$, and outputs a value $y\in \Ranprf$.
\item[$\CEval(\sk_f,x) \ra y$:] The constrained evaluation algorithm takes as input $\sk_f$ and $x \in \Domprf$, and outputs a value $y \in \Ranprf$.
\end{description}
\end{definition}

\begin{definition}[Security for CPRF]\label{def:cprf_security}
A private CPRF should satisfy correctness, pseudorandomness, and privacy.
\begin{description}
\item[Correctness:] A CPRF is correct if for any (stateful) QPT adversary $\qA$, it holds that
\[
\Pr\left[
\begin{array}{rl}
&\Eval(\msk,x) \ne \CEval(\sk_f,x)\\
&\land \ x\in\Domprf \land f(x)=0
\end{array}
 \ \middle |
\begin{array}{rl}
 &(1^\kappa,f) \lrun \qA(1^\secp),\\
 & \msk\lrun \Setup(1^\secp,1^\kappa),\\
 &\sk_f \lrun \Constrain(\msk,f), \\
 &x \lrun \qA^{\Eval(\msk,\cdot)}(\sk_f)
\end{array}
\right]\le  \negl(\secp).
\]
\item[Selective single-key pseudorandomness:]
A CPRF is selectively single-key pseudorandom if for any (stateful) QPT adversary $\qA$, it hods that
\[
2\abs{\Pr\left[
\begin{array}{rl}
&\qA^{\Eval(\msk,\cdot)}(y_b)=b\\
&\land \ x\notin \cQ_e \\
&\land f(x)\ne0
\end{array}
 \ \middle |
\begin{array}{rl}
 &(1^\kappa,f) \lrun \qA(1^\secp),\\
 & \msk\lrun \Setup(1^\secp,1^\kappa),\\
 & sk_f \lrun \Constrain(\msk,f)\\
 &x \lrun \qA^{\Eval(\msk,\cdot)}(\sk_f)\\
 &y_0 \seteq \Eval(\msk,x),y_1\chosen \Ranprf,\\
 &b\chosen \zo{}
\end{array}
\right] -\frac{1}{2}} \le \negl(\secp),
\]
where $\cQ_e$ is the sets of queries to $\Eval(\msk,\cdot)$.

\item[Selective single-key privacy:]
A CPRF is selectively single-key private if for any (stateful) QPT adversary $\qA$, there exists a stateful PPT simulator $\Sim= (\Sim_1,\Sim_2)$ that satisfying that
\[
2\abs{\Pr\left[
\qA^{\cO_b(\cdot)}(\sk_b)=b
 \ \middle |
\begin{array}{rl}
 &(1^\kappa,f) \lrun \qA(1^\secp),\\
 & \msk\lrun \Setup(1^\secp,1^\kappa), b\chosen \zo{},\\
 &\sk_0 \lrun \Constrain(\msk,f),\\
 &(\stinfo_\Sim,\sk_1) \lrun \Sim_1(1^\kappa,1^\secp)
\end{array}
\right] -\frac{1}{2}} \le  \negl(\secp),
\]
where $\cO_0(\cdot) \seteq \Eval(\msk,\cdot)$ and $\cO_1(\cdot) \seteq \Sim_2(\stinfo_\Sim,\cdot,f(\cdot))$.
\end{description}
We say that a CPRF is a selectively single-key private CPRF if it satisfies correctness, selective single-key pseudorandomness, and selective single-key privacy.
\end{definition}

\begin{theorem}[\cite{TCC:BTVW17,PKC:PeiShi18}]\label{thm:pcprf_lwe}
If the QLWE assumption holds, there exists a selectively signle-key private CPRF for polynomial-size classical circuits.
\end{theorem}

\begin{definition}[PKE]\label{def:pke_syntax}
A PKE with plaintext space $\Ps =\setbk{\Ps_\secp}_{\secp\in \bbN}$, ciphertext space $\Cs= \setbk{\Cs_\secp}_{\secp\in\bbN}$ is a tuple of three algorithms.
\begin{description}
\item[$\Gen(1^\secp)\ra (\pk,\sk)$:] The key generation algorithm takes as input the security parameter $\secp$ and outputs a key pair $(\pk,\sk)$.
\item[$\Enc(\pk,m)\ra \ct$:] The encryption algorithm takes as input $\pk$, a plaintext $m\in\Ps$, and outputs a ciphertext $\ct\in \Cs$.
\item[$\Dec(\sk,\ct)\ra m^\prime/\bot$:] The decryption algorithm takes as input $\sk$ and $\ct \in \Cs$, and outputs a plaintext $m^\prime \in \Ps$ or $\bot$.
\item[Correctness:] A PKE scheme is correct if for all $\secp \in \bbN$ and $m\in \Ps_{\secp}$, it holds that
\[
\Pr[\Dec(\sk,\ct)=m \mid (\pk,\sk) \lrun \Gen(1^\secp),\ct \lrun \Enc(\pk,m)]=1.
\]
\end{description}
\end{definition}

\begin{definition}[CCA Security for PKE]\label{def:pke_CCA}
A PKE scheme is CCA secure if for any (stateful) QPT adversary $\qA$, it holds that
\[
2\abs{\Pr\left[
\begin{array}{rl}
&b^\ast=b\\
\land &\ct_b\notin\cQ
\end{array}
 \ \middle |
\begin{array}{rl}
 &(\pk,\sk) \lrun \Gen(1^\secp),\\
 & (m_0,m_1)\lrun \qA^{\Dec(\sk,\cdot)}(1^\secp,\pk),\\
 &b\chosen \zo{}, \ct_b \lrun \Enc(\pk,m_b),\\
 & b' \lrun \qA^{\Dec(\sk,\cdot)}(1^\secp,\ct_b)
\end{array}
\right] -\frac{1}{2}} \le  \negl(\secp),
\]
where $\cQ$ is the set of queries to $\Dec(\sk,\cdot)$ after $\qA$ is given $\ct_b$.

\end{definition}

\begin{theorem}[\cite{STOC:Peikert09}]\label{thm:CCA_QLWE}
If the QLWE assumption holds, there exists a PKE scheme that satisfies CCA security.
\end{theorem}

\begin{definition}[Indistinguishability Obfuscator~\cite{JACM:BGIRSVY12}]\label{def:io}
A PPT algorithm $\iO$ is a secure \rmIO for a classical circuit class $\{\cC_\secp\}_{\secp \in \bbN}$ if it satisfies the following two conditions.

\begin{description}
\item[Functionality:] For any security parameter $\secp \in \bbN$, circuit $C \in \cC_\secp$, and input $x$, we have that
\begin{align}
\Pr[C^\prime(x)=C(x) \mid C^\prime \lrun \iO(C)] = 1\enspace.
\end{align}

\item[Indistinguishability:] For any PPT $\Sampler$ and QPT distinguisher $\qD$, the following holds:

If $\Pr[\forall x,\ C_0(x)=C_1(x) \mid (C_0,C_1,\aux)\gets \Sampler(1^\secp)]> 1 - \negl(\secp)$, then we have
 \begin{align}
\adva{\iO,\qD}{io}(\secp) &\seteq \left|
 \Pr\left[\qD(\iO(C_0),\aux) = 1 \mid (C_0,C_1,\aux)\gets \Sampler(1^\secp) \right] \right.\\
&~~~~~~~ \left. - \Pr\left[\qD(\iO(C_1),\aux)= 1\mid (C_0,C_1,\aux)\gets \Sampler(1^\secp)  \right] \right|
  \leq \negl(\secp).
 \end{align}
\end{description}
\end{definition}

There are a few candidates of secure IO for polynomial-size classical circuits against quantum adversaries~\cite{TCC:BGMZ18,TCC:CHVW19,EC:AgrPel20,myTCC:DQVWW21}. In some candidates~\cite{EC:WeeWic21,STOC:GayPas21}, the assumptions behind the constructions were found to be false~\cite{C:HopJaiLin21}.

%% file: section/def_q_watermarking.tex

\section{Definition of Quantum Watermarking}\label{sec:def_Q_watermarking}
We introduce definitions for watermarking PRFs against quantum adversaries in this section.
\subsection{Syntax and Pseudorandomness}\label{sec:watermarking_syntax}
\begin{definition}[Watermarking PRF]\label{def:watermarking_prf}
A watermarking PRF $\WMPRF$ for the message space $\M \seteq \zo{\msglen}$ with domain $\Domprf$ and range $\Ranprf$ is a tuple of five algorithms $(\Setup,\Gen,\Eval,\Mark,\qExtract)$.
\begin{description}

 \item[$\Setup(1^\secp)\ra (\pp,\xk)$:] The setup algorithm takes as input the security parameter and outputs a public parameter $\pp$ and an extraction key $\xk$.
 \item[$\Gen(\pp)\ra(\prfk,\iop)$:] The key generation algorithm takes as input the public parameter $\pp$ and outputs a PRF key $\prfk$ and a public tag $\iop$.
 \item[$\Eval(\prfk,x)\ra y$:] The evaluation algorithm takes as input a PRF key $\prfk$ and an input $x \in \Domprf$ and outputs $y \in \Ranprf$.
\item[$\Mark (\pp,\prfk, \msg) \ra \tlC$:] The mark algorithm takes as input the public parameter $\pp$, a PRF key $\prfk$, and a message $\msg \in \bit^\msglen$, and outputs a marked evaluation circuit $\tlC$.

\item[$\qExtract(\xk,\iop,\qstate{C}^{\prime},\epsilon) \ra \msg^\prime$:] The extraction algorithm takes as input an extraction key $\xk$, a tag $\iop$, a quantum circuit with classical inputs and outputs $\qstate{C}'=(\qstateq,\mat{U})$, and a parameter $\epsilon$, and outputs $\msg^\prime$ where $\msg' \in \bit^\msglen \cup \{\unmarked\}$. 
\end{description}

\begin{description}
 \item[Evaluation Correctness:] For any message $\msg \in \bit^\msglen$, it holds that
\[
 \Pr\left[\widetilde{C}(x)=\Eval(\prfk,x)~\left|~ \begin{array}{c} (\pp,\xk) \chosen \Setup(1^\secp)\\ (\prfk,\iop)\gets\Gen(\pp)\\ \tlC \gets\Mark (\pp,\prfk,\msg)\\x\gets\Domprf\end{array}\right.\right] \geq 1 - \negl(\secp).
\]
\end{description}
\end{definition}

\begin{remark}[On extraction correctness]
Usually, a watermarking PRF scheme is required to satisfy extraction correctness that ensures that we can correctly extract the embedded mark from an honestly marked circuit.
However, as observed by Quach et al.~\cite{TCC:QuaWicZir18}, if we require the extraction correctness to hold for a randomly chosen PRF key, it is implied by unremovability defined below.
Note that the unremovability defined below considers a distinguisher as a pirate circuit. However, it implies the extraction correctness since we can easily transform an honestly marked circuit into a successful distinguisher.
Thus, we do not explicitly require a watermarking PRF scheme to satisfy extraction correctness in this work.
\end{remark}

\begin{remark}[On public marking]
We consider only watermarking PRFs with public marking as in~\cref{def:watermarking_prf} since we can achieve public marking by default.
The reason is as follows. Suppose that we generate $\pp$, $\xk$, and a marking key $\mk$ at the setup.
When we generate a PRF key and a public tag at $\Gen$, we can first generate $(\pp^\prime,\xk^\prime,\mk^\prime)\gets \Setup(1^\secp)$ from scratch (ignoring the original $(\pp,\xk,\mk)$) and set a PRF key $\widehat{\prfk} \seteq (\prfk^\prime,\mk^\prime)$ and a public tag $\widehat{\iop}\seteq (\pp^\prime,\xk^\prime,\iop^\prime)$ where $(\prfk^\prime,\iop^\prime) \gets \Gen(\pp^\prime)$. That is, anyone can generate a marked circuit from $\widehat{\prfk} =(\prfk^\prime,\mk^\prime)$ by $\Mark(\mk^\prime,\prfk^\prime,\msg)$. Therefore, we consider public marking by default in our model.
\end{remark}

\begin{remark}[On private marking]\label{rmrk:private_marking}
We might prefer private marking in some settings since we might want to prevent adversaries from forging a watermarked PRF.
We can convert watermarking PRFs in~\cref{def:watermarking_prf} into ones with private marking by using signatures.
Below, we assume that a PRF key $\prfk$ includes its public tag $\iop$ since it does not harm security.
At the setup phase, we also generate a signature key pair $(\vk,\sk)\gets \SIG.\Gen(1^\secp)$ and set a mark key $\mk^\prime \seteq (\vk,\sk)$ and an extraction key $\xk^\prime \seteq (\xk,\vk)$. To embed a message $\msg$ into $\prfk$, we generate a signature $\sigma \gets \SIG.\Sign(\sk,\iop \concat \msg)$ and generate $\tlC \gets \Mark(\pp,\prfk, \msg\concat \sigma)$. To extract a message, we run $\msg^\prime \gets \qExtract(\xk,\tau,\qstate{C}^\prime,\epsilon)$, parse $\msg^\prime = \msg \concat \sigma$, and run $\SIG.\Vrfy(\vk,\iop \concat \msg,\sigma)$. If the verification result is $\top$, we output $\msg$. This conversion is the same as what Goyal et al.~\cite{C:GKMWW19} proposed.
Adversaries cannot forge a signature for $\tau^\ast\concat \msg^\ast \ne \tau\concat \msg$ by the unforgeability of $\SIG$.
Intuitively, if an adversary can forge a watermarked PRF whose functionality is different from those of watermarked PRFs given from a mark oracle, $\tau^\ast \ne \tau$ should hold since public tags are related to PRF keys. This breaks the unforgeability of $\SIG$.
Thus, we expect that adversaries cannot break the unforgeability of watermarking. However, we do not formally define watermarking unforgeability against quantum adversaries since it is not a scope of this work. We leave it as future work.
\end{remark}

\paragraph{Discussion on syntax.}
\Cref{def:watermarking_prf} is a natural quantum variant of classical watermarking PRFs except that the key generation algorithm outputs a public tag $\iop$, and the extraction algorithm uses it. Such a public tag is not used in previous works on watermarking PRFs~\cite{SIAMCOMP:CHNVW18,myJC:KimWu21,TCC:QuaWicZir18,C:KimWu19,AC:YALXY19}. A public tag should not harm watermarking PRF security.
We justify using $\iop$ as follows.

First, we need to obtain many pairs of input and output to extract an embedded message from a marked PRF in almost all known (classical) watermarking constructions~\cite{SIAMCOMP:CHNVW18,PKC:BonLewWu17,myJC:KimWu21,TCC:QuaWicZir18,C:KimWu19,AC:YALXY19,C:GKMWW19,TCC:Nishimaki20}. This is because we must check whether a tested PRF circuit outputs particular values for particular inputs which \emph{depends on the target PRF} (such particular inputs are known as marked points).
Suppose marked points are fixed and do not depend on a PRF that will be marked. In that case, an adversary can easily remove an embedded message by destroying functionalities at the fixed marked points that could be revealed via a (non-target) marked PRF that an adversary generated.
Recall that we consider the public marking setting.
The attack was already observed by Cohen et al.~\cite{SIAMCOMP:CHNVW18}.

Second, we consider a stronger adversary model than that in most previous works as the definition of traceable PRFs by Goyal et al.~\cite{myAC:GKWW21}.
An adversary outputs a distinguisher-based pirate circuit in our security definition rather than a pirate circuit that computes an entire output of a PRF. This is a refined and realistic model, as Goyal et al.~\cite{myAC:GKWW21} argued (and we explain in~\cref{sec:related_work}).
In this model, we cannot obtain a valid input-output pair from a pirate circuit anymore. Such a pair is typical information related to a target PRF.
Goyal et al. resolve this issue by introducing a tracing key that is generated from a target PRF. Note that parameters of watermarking ($\pp$ and $\xk$) should \emph{not} be generated from a PRF since we consider many different PRF keys in the watermarking PRF setting.

Thus, if we would like to achieve an extraction algorithm and the stronger security notion simultaneously, an extraction algorithm should somehow take information related to a target PRF as input to correctly extract an embedded message.
In the weaker adversary model, an extraction algorithm can easily obtain many valid input and output pairs by running a tested circuit many times. However, in the stronger distinguisher-based pirate circuit model, a pirate circuit outputs a single decision bit.

To resolve this issue, we introduce public tags. We think it is natural to have information related to the original PRF key in an extraction algorithm. In reality, we check a circuit when a user claims that her/his PRF key (PRF evaluation circuit) is illegally used. Thus, it is natural to expect we can use a user's public tag for extraction. This setting resembles watermarking for public-key cryptographic primitives, where a user public key is available in an extraction algorithm. In addition, public tags do not harm PRF security in our constructions. It is unclear whether we can achieve unremovability in the stronger distinguisher-based model without any syntax change (even in the classical setting).
\footnote{Even if we consider the weaker adversary model, the same issue appears in the quantum setting in the end. If we run a quantum circuit for an input and measure the output, the measurement could irreversibly alter the quantum state and we lost the functionality of the original quantum state.
That is, there is no guarantee that we can correctly check whether a tested quantum circuit is marked or not \emph{after} we obtain a single valid pair of input and output by running the circuit. However, as we explained above, we want to obtain information related to a target PRF for extraction. Thus, we need a public tag in the syntax in either case.}

\paragraph{Extended pseudorandomness.}
We consider extended weak pseudorandomness, where weak pseudorandomness holds even if the adversary generates $\pp$.
This notion is the counterpart of extended pseudorandomness by Quach et al.~\cite{TCC:QuaWicZir18}, where pseudorandomness holds in the presence of the extraction oracle. However, our pseudorandomness holds even against an authority unlike extended pseudorandomness by Quach et al. since we allow adversaries to generate a public parameter.
\begin{definition}[Extended Weak Pseudorandomness against Authority]\label{def:extednded_weak_pseudorandomness}
To define extended weak pseudorandomness for watermarking PRFs, we define the game $\expb{\qA,\WMPRF}{ext}{wprf}(\secp)$ as follows.
\begin{enumerate}
\item $\qA$ first sends $\pp$ to the challenger.
\item The challenger generates $(\prfk,\iop) \gets \Gen(\pp)$ and sends $\iop$ to $\qA$.
\item  The challenger chooses $\coin\gets\bit$. $\qA$ can access to the following oracles.
\begin{description}
\item[$\Oracle{wprf}$:] When this is invoked (no input), it returns $(a,b)$ where $a\chosen \Domprf$ and $b\seteq \Eval(\prfk,a)$.
\item[$\Oracle{chall}$:] When this is invoked (no input), it returns:
\begin{itemize}
 \item $(a,b)$ where $a\chosen \Domprf$ and $b\seteq \Eval(\prfk,a)$ if $\coin =0$, 
 \item $(a,b)$ where $a\chosen \Domprf$ and $b\chosen \Ranprf$ if $\coin=1$.
 \end{itemize}
 This oracle is invoked only once.
\end{description}
\item When $\qA$ terminates with output $\coin^\prime$, the challenger outputs $1$ if $\coin=\coin^\prime$ and $0$ otherwise.
\end{enumerate}

We say that $\WMPRF$ is extended weak pseudorandom if for every QPT $\qA$, we have
\begin{align}
\advb{\qA,\WMPRF}{ext}{wprf}(\secp)=
2\abs{\Pr[
\expb{\qA,\WMPRF}{ext}{wprf}(\secp)=1
] -\frac{1}{2}} =\negl(\secp).
\end{align}

\end{definition}

\subsection{Unremovability against Quantum Adversaries}
We define unremovability for watermarking PRFs against quantum adversaries.
\ifnum\submission=1
We first define quantum program with classical inputs and outputs and then define unremovability.
\begin{definition}[Quantum Program with Classical Inputs and Outputs~\cite{C:ALLZZ21}]\label{def:Q_program_C_IO}
A quantum program with classical inputs is a pair of quantum state $\qstateq$ and unitaries $\setbk{\mat{U}_x}_{x\in[N]}$ where $[N]$ is the domain, such that the state of the program evaluated on input $x$ is equal to $\mat{U}_x \qstateq \mat{U}_x^\dagger$. We measure the first register of $\mat{U}_x \qstateq \mat{U}_x^\dagger$ to obtain an output. We say that $\setbk{\mat{U}_x}_{x\in[N]}$ has a compact classical description $\mat{U}$ when applying $\mat{U}_x$ can be efficiently computed given $\mat{U}$ and $x$.
\end{definition}
\else\fi

\begin{definition}[Unremovability for private extraction]\label{def:unrem_prf}
We consider the public marking and secret extraction setting here. Let $\epsilon\geq0$.
We define the game $\expt{\qA,\WMPRF}{nrmv}(\secp,\epsilon)$ as follows.
\begin{enumerate}
\item The challenger generates $(\pp,\xk) \chosen \Setup(1^\secp)$ and gives $\pp$ to the adversary $\qA$.
$\qA$ send $\msg\in\bit^\msglen$ to the challenger.
The challenger generates $(\prfk,\iop)\gets\Gen(\pp)$, computes $\tlC\gets\Mark(\pp,\prfk,\msg)$, and sends $\iop$ and $\tlC$ to $\qA$.

\item $\qA$ can access to the following oracle.
\begin{description}
\item[$\Oracle{ext}$:] On input $\iop'$ and a quantum circuit $\qstate{C}$, it returns $\qExtract(\xk,\qstate{C},\iop',\epsilon)$.
\end{description}

\item Finally, the adversary outputs a ``pirate'' quantum circuit $\pirateC =(\qstateq,\mat{U})$, where $\pirateC$ is a quantum program with classical inputs and outputs whose first register (i.e., output register) is $\bbC^2$ and $\mat{U}$ is a compact classical description of $\{\mat{U}_{x,y}\}_{x\in\Domprf,y\in\Ranprf}$.
\end{enumerate}

Let $D$ be the following distribution.
\begin{description}
\item[$D$:]Generate $b\gets\bit$, $x\gets\Domprf$, and $y_0\gets\Ranprf$. Compute $y_1\gets\Eval(\prfk,x)$. Output $(b,x,y_b)$.
\end{description}
We also let $\cP=(\mat{P}_{b,x,y},\mat{Q}_{b,x,y})_{b,x,y}$ be a collection of binary outcome projective measurements, where
\begin{align}
\mat{P}_{b,x,y}=\mat{U}_{x,y}^\dagger\ket{b}\bra{b}\mat{U}_{x,y}\textrm{~~~~and~~~~}\mat{Q}_{b,x,y}=\mat{I}-\mat{P}_{b,x,y}.
\end{align}
Moreover, we let $\cM_D=(\mat{P}_D,\mat{Q}_D)$ be binary outcome POVMs, where
\begin{align}
\mat{P}_D=\sum_{r\in\calR}\frac{1}{\abs{\calR}}\mat{P}_{D(r)}\textrm{~~~~and~~~~}\mat{Q}_D =\mat{I}-\mat{P}_D.
\end{align}

\begin{description}
\item[$\Live$:] When applying the measurement $\projimp(\cM_D)$ to $\qstateq$, we obtain a value $p$ such that $p\geq\frac{1}{2}+\epsilon$.
\item[$\GoodExt$:] When Computing $\msg^\prime \lrun \qExtract(\xk,\pirateC,\iop,\epsilon)$, it holds that $\msg^\prime \ne \unmarked$.
\item[$\BadExt$:] When Computing $\msg^\prime \lrun \qExtract(\xk, \pirateC,\iop,\epsilon)$, it holds that $\msg^\prime \notin \{\msg,\unmarked\}$.
\end{description}

We say that $\WMPRF$ satisfies unremovability if for every $\epsilon>0$ and QPT $\qA$, we have
\begin{align}
\Pr[\BadExt]\leq\negl(\secp)\textrm{~~~~and~~~~}\Pr[\GoodExt]\geq\Pr[\Live]-\negl(\secp).
\end{align}
\end{definition}

Intuitively, $(\mat{P}_{b,x,y},\mat{Q}_{b,x,y})$ is a projective measurement that feeds $(x,y)$ to $\pirateC$ and checks whether the outcome is $b$ or not (and then uncomputes).
Then, $\cM_D$ can be seen as POVMs that results in $0$ with the probability that $\pirateC$ can correctly guess $b$ from $(x,y_b)$ for $(b,x,y_b)$ generated randomly from $D$.

\begin{remark}[On attack model]\label{remark:attack_model}
We check whether $\pirateC$ correctly distinguishes a real PRF value from a random value or not by applying $\projimp(\cM_D)$ to $\qstateq$.
This attack model follows the refined and more realistic attack model by Goyal et al.~\cite{myAC:GKWW21}. The adversary outputs a pirate circuit that computes an entire PRF value in all previous works except their work.

The distinguisher-based pirate circuit model is compatible with the (quantum) pirate decoder model of traitor tracing.
Thus, our attack model also follows the attack model of quantum traitor tracing (the black box projection model) by Zhandry~\cite[Section 4.2]{TCC:Zhandry20}.\footnote{In the watermarking setting, an extraction algorithm can take the description of a pirate circuit as input (corresponding to the software decoder model~\cite[Section 4.2]{TCC:Zhandry20}), unlike the black-box tracing model of traitor tracing. However, we use a pirate circuit in the black box way for our extraction algorithms. Thus, we follow the black box projection model by Zhandry~\cite{TCC:Zhandry20}.}

As in the traitor tracing setting~\cite{TCC:Zhandry20}, $\projimp(\cM_D)$ is inefficient in general. We can handle this issue as Zhandry did. We will use an approximate version of $\projimp(\cM_D)$ to achieve an efficient reduction.
In addition, we cannot apply both $\projimp(\cM_D)$ and $\qExtract$ to $\pirateC$ simultaneously.
However, the condition $\Pr[\GoodExt]\geq\Pr[\Live]-\negl(\secp)$ claims that an embedded mark cannot be removed as long as the pirate circuit is alive.
This fits the spirit of watermarking.
See Zhandry's paper~\cite[Section 4]{TCC:Zhandry20} for more discussion on the models.
\end{remark}

\begin{remark}[On selective message]\label{remark:na_ada_message_watermarking}
As we see in~\cref{def:unrem_prf}, we consider the selective setting for private extraction case, where $\qA$ must send the target message $\msg$ to the challenger before $\qA$ accesses to the oracle $\Oracle{ext}$ and after $\pp$ is given. This is the same setting as that by Quach et al.~\cite{TCC:QuaWicZir18}. We can consider the fully adaptive setting, where $\qA$ can send the target message $\msg$ after it accesses to the oracle $\Oracle{ext}$, as Kim and Wu~\cite{C:KimWu19}. However, our privately extractable watermarking PRF satisfies only selective security. Thus, we write only the selective variant for the private extraction case.
\end{remark}

\begin{definition}[Unremovability for Public Extraction]\label{def:pub_ext_unrem_prf}
This is the same as~\cref{def:unrem_prf} except we use the game $\expc{\qA,\WMPRF}{pub}{ext}{nrmv}(\secp,\epsilon)$ defined in the same way as $\expt{\qA,\WMPRF}{nrmv}(\secp,\epsilon)$ except the following differences.
\begin{itemize}
\item In item 1, $\qA$ is given $\xk$ together with $\pp$.
\item Item 2 is removed.
\end{itemize}
\end{definition}

%% file: section/def_extless_watermarking.tex
\newcommand{\KeyExtract}{\mathsf{KeyExt}}
\newcommand{\challenge}{\mathsf{ch}}

\section{Definition of Extraction-Less Watermarking}\label{sec:def_extless_watermarking}

We introduce the notion of extraction-less watermarking PRF as an intermediate primitive towards watermarking PRFs secure against quantum adversaries.

\subsection{Syntax and Pseudorandomness}
\begin{definition}[Extraction-Less Watermarking PRF]\label{def:extless_watermarking_prf}
An extraction-less watermarking PRF $\WMPRF$ for the message space $\bit^\msglen$ with domain $\Domprf$ and range $\Ranprf$ is a tuple of five algorithms $(\Setup,\Gen,\Eval,\Mark,\Sim)$, where the first four algorithms have the same input/output behavior as those defined in \cref{def:watermarking_prf} and $\Sim$ has the following input/output behavior.
\begin{description}



\item[$\Sim(\xk,\iop,i)\ra (\gamma,x,y)$:]The simulation algorithm $\Sim$ takes as input the extraction key $\xk$, a tag $\iop$, and an index $i$, and outputs a tuple $(\gamma,x,y)$.

\end{description}

\begin{description}
 \item[Evaluation Correctness:] It is defined in exactly the same way as the evaluation correctness for watermarking PRF defined in \cref{def:watermarking_prf}.

\end{description}
\end{definition}

\paragraph{Extended pseudorandomness.}
Extended pseudorandomness for extraction-less watermarking PRF is defined in exactly the same way as that for watermarking PRF, that is \cref{def:extednded_weak_pseudorandomness}.

\subsection{Simulatability for Mark-Dependent Distributions (SIM-MDD Security)}\label{sec:extraction-less_watermarking_simmdd}

We introduce the security notion for extraction-less watermarking PRF that we call simulatability for mark-dependent distributions.
Let $D$ and $\Drev$ be the following distributions.
\begin{description}
\item[$D$:]Generate $b\gets\bit$, $x\gets\Domprf$, and $y_0\gets\Ranprf$. Compute $y_1\gets\Eval(\prfk,x)$. Output $(b,x,y_b)$.
\item[$\Drev$:]Generate $(b,x,y)\gets D$. Output $(1\oplus b,x,y)$.
\end{description}
Namely, $D$ is the distribution that outputs a random value if the first bit $b=0$ and a PRF evaluation if the first bit $b=1$, and $\Drev$ is its opposite (i.e., a PRF evaluation if $b=0$ and a random value if $b=1$).
SIM-MDD security is a security notion that guarantees that an adversary given $\tlC\gets\Mark(\mk,\prfk,\msg)$ cannot distinguish an output of $\Sim(\xk,\iop,i)$ from that of $D$ if $\msg[i]=0$ and from that of $\Drev$ if $\msg[i]=1$.

\begin{definition}[SIM-MDD Security with Private Simulation]\label{def:SIM-MDD}
To define SIM-MDD security with private simulation, we define the game $\expt{i^\ast,\qA,\WMPRF}{sim\textrm{-}mdd}(\secp)$ as follows, where $i^\ast\in[\msglen]$.
\begin{enumerate}
\item The challenger generates $(\pp,\xk) \chosen \Setup(1^\secp)$ and sends $\pp$ to $\qA$.
$\qA$ sends $\msg\in \bit^\msglen$ to the challenger.
The challenger generates $(\prfk,\iop)\gets\Gen(\pp)$ and computes $\tlC\gets\Mark(\mk,\prfk,\msg)$.
The challenger sends $\iop$ and $\tlC$ to $\qA$.
\item $\qA$ can access to the following oracle.
\begin{description}
\item[$\Oracle{sim}$:] On input $\iop'$ and $i'\in[\msglen]$, it returns $\Sim(\xk,\iop',i')$.
\end{description}

\item Let $\Dreal{i^\ast}$ be the following distribution. Note that $\Dreal{i^\ast}$ is identical with $D$ if $\msg[i^\ast]=0$ and with $\Drev$ if $\msg[i^\ast]=1$.
\begin{description}
\item[$\Dreal{i^\ast}$:]Generate $\gamma\gets\bit$ and $x\gets\Domprf$. Then, if $\gamma=\msg[i^\ast]$, generate $y\gets\Ranprf$, and otherwise, compute $y\gets\Eval(\prfk,x)$. Output $(\gamma,x,y)$.
\end{description}
The challenger generates $\coin\gets\bit$.
If $\coin=0$, the challenger samples $(\gamma,x,y)\gets\Dreal{i^\ast}$.
If $\coin=1$, the challenger generates $(\gamma,x,y)\gets\Sim(\xk,\iop,i^\ast)$.
The challenger sends $(\gamma,x,y)$ to $\qA$.

\item When $\qA$ terminates with output $\coin^\prime$, the challenger outputs $1$ if $\coin=\coin^\prime$ and $0$ otherwise.
\end{enumerate}
Note that $\qA$ is not allowed to access to $\Oracle{sim}$ after $\qA$ is given $(\gamma,x,y)$.

We say that $\WMPRF$ is SIM-MDD secure if for every $i^\ast\in[\msglen]$ and QPT $\qA$, we have
\begin{align}
\adva{i^\ast,\qA,\WMPRF}{sim\textrm{-}mdd}(\secp)=
2\abs{\Pr[
\expt{i^\ast,\qA,\WMPRF}{sim\textrm{-}mdd}(\secp)=1
] -\frac{1}{2}} =\negl(\secp).
\end{align}
\end{definition}

We consider the selective setting above as unremovability for private extraction in~\cref{def:unrem_prf} since we use SIM-MDD security with private simulation to achieve unremovability for private simulation.


\begin{remark}[On multi challenge security]
We can prove that the above definition implies the multi-challenge variant where polynomially many outputs of $\Sim(\xk,\iop,i^\ast)$ are required to be indistinguishable from those of $\Dreal{i^\ast}$.
This is done by hybrid arguments where outputs of $\Sim(\xk,\iop,i^\ast)$ are simulated using $\Oracle{sim}$ and those of $\Dreal{i^\ast}$ are simulated using $\tlC$.
To apply \cref{cor:cind_sample_api}, we need the multi challenge variant. However, we consider the single challenge variant due to the implication above.
A similar remark is applied to the variants of SIM-MDD security introduced below.
\end{remark}

\paragraph{SIM-MDD security with private simulation under the $\API$ oracle.}
Let the $\API$ oracle be an oracle that is given $(\epsilon,\delta,\iop',i')$ and a quantum state $\qstateq$, and returns the result of $\API^{\epsilon,\delta}_{\cP,\D_{\iop',i'}}(\qstateq)$ and the post measurement state, where $\cP$ is defined in the same way as that in \cref{def:unrem_prf} and $\D_{\iop',i'}$ be the distribution that outputs randomly generated $(\gamma,x,y)\gets\Sim(\xk,\iop',i')$.
The $\API$ oracle cannot be simulated using the simulation oracle $\Oracle{sim}$ since we need superposition of outputs of $\Sim$ to compute $\API^{\epsilon,\delta}_{\cP,\D_{\iop',i'}}(\qstateq)$.
When constructing watermarking PRFs with private simulation from extraction-less watermarking PRFs, the underlying extraction-less watermarking PRF scheme needs to satisfy SIM-MDD security with private simulation under the $\API$ oracle that we call QSIM-MDD security with private simulation.
The reason is as follows.
In the security analysis of the construction, the indistinguishability guarantee provided by SIM-MDD security needs to hold for an adversary against the resulting watermarking scheme who can access the extraction oracle.
This means that it also needs to hold for an adversary who can access the $\API$ oracle since $\API$ is repeatedly invoked in the extraction algorithm of the resulting scheme.

Fortunately, as we will see, we can generically convert an extraction-less watermarking PRF scheme satisfying SIM-MDD security with private simulation into one satisfying QSIM-MDD security with private simulation, using QPRFs.
Thus, when realizing an extraction-less watermarking PRF scheme as an intermediate step towards privately extractable watermarking PRFs, we can concentrate on realizing one satisfying SIM-MDD security with private simulation.

\begin{remark}
There is a similar issue in the traitor tracing setting. If PLBE is a secret-key based one, we need a counterpart of QSIM-MDD in secret-key based PLBE to achieve traitor tracing with a secret tracing algorithm against quantum adversaries by using Zhandry's framework~\cite{TCC:Zhandry20}. Note that Zhandry focuses on public-key based PLBE in his work~\cite{TCC:Zhandry20}.
\end{remark}

\begin{definition}[QSIM-MDD Security with Private Simulation]\label{def:qsim-mdd}
Let $\D_{\iop,i}$ be a distribution defined as follows.
\begin{description}
\item[$\D_{\iop,i}$:] Output $(\gamma,x,y)\gets\Sim(\xk,\iop,i)$.
\end{description}
Then, we define the game $\expc{i^\ast,\qA,\WMPRF}{q}{sim}{mdd}(\secp)$ in the same way as $\expb{i^\ast,\qA,\WMPRF}{sim}{mdd}(\secp)$ except that in addition to $\Oracle{sim}$, $\qA$ can access to the following oracle in the step 2.
\begin{description}
\item[$\Oracle{api}$:]On input $(\epsilon,\delta,\iop',i')$ and a quantum state $\qstateq$, it returns the result of $\API^{\epsilon,\delta}_{\cP,\D_{\iop',i'}}(\qstateq)$ and the post measurement state, where $\cP$ is defined in the same way as that in \cref{def:unrem_prf}.
\end{description}

We say that $\WMPRF$ is QSIM-MDD secure with private simulation if
for every $i^\ast\in[\msglen]$ and QPT $\qA$, we have
\begin{align}
\advc{i^\ast,\qA,\WMPRF}{q}{sim}{mdd}(\secp)=
2 \abs{\Pr[
\expc{i^\ast,\qA,\WMPRF}{q}{sim}{mdd}(\secp)=1
]- \frac{1}{2}} =\negl(\secp).
\end{align}
\end{definition}

We have the following theorem.
\begin{theorem}\label{thm:qelwmprf-from-elwmprf}
Assume there exists an extraction-less watermarking PRF scheme satisfying SIM-MDD security with private simulation and a QPRF.
Then, there exists an extraction-less watermarking PRF scheme satisfying QSIM-MDD security with private simulation.
\end{theorem}
We prove this theorem in \cref{sec:qelwmprf-from-elwmprf}.

\begin{definition}[SIM-MDD Security with Public Simulation]\label{def:SIM-MDD_public_simulation} 
We define the game $\expc{i^\ast,\qA,\WMPRF}{sim}{mdd}{pub}(\secp)$ in the same way as $\expt{i^\ast,\qA,\WMPRF}{sim\textrm{-}mdd}(\secp)$ except the following differences, where $i^\ast\in[\msglen]$.
\begin{itemize}
\item In item 1, $\qA$ is given $\xk$ together with $\pp$.
\item Item 2 is removed.
\end{itemize}

We say that $\WMPRF$ satisfies SIM-MDD security with public simulation if for every $i^\ast\in[\msglen]$ and QPT $\qA$, we have
\begin{align}
\advc{i^\ast,\qA,\WMPRF}{sim}{mdd}{pub}(\secp)=
2\abs{\Pr[\expc{i^\ast,\qA,\WMPRF}{sim}{mdd}{pub}(\secp)=1]-\frac{1}{2}} =\negl(\secp).
\end{align}
\end{definition}

%% file: section/wmprf-from-elwmprf.tex
\section{Watermarking PRF from Extraction-Less Watermarking PRF}\label{sec:wmprf-from-elwmprf}
We show how to construct watermarking PRF secure against quantum adversaries from extraction-less watermarking PRF.

Let $\ELWMPRF=(\Setup,\Gen,\Eval,\Mark,\Sim)$ be an extraction-less watermarking PRF scheme whose message space is $\bit^{\msglen+1}$.
We construct a watermarking PRF scheme $\WMPRF=(\WM.\Setup,\allowbreak\WM.\Gen,\WM.\Eval,\WM.\Mark,\qExtract)$ whose message space is $\bit^{\msglen}$ as follows. We use $\Setup$, $\Gen$, and $\Eval$ as $\WM.\Setup$, $\WM.\Gen$, and $\WM.\Eval$, respectively. Thus, the domain and range of $\WMPRF$ are the same as those of $\ELWMPRF$.
Also, we construct $\WM.\Mark$ and $\qExtract$ as follows.

\begin{description}

%
\item[$\WM.\Mark (\pp,\prfk, \msg)$:] $ $
\begin{itemize}
\item Output $\tlC \gets \Mark(\pp,\prfk,\msg\|0)$.
\end{itemize}
\item[$\qExtract(\xk,\qstate{C},\iop,\epsilon)$:]$ $
\begin{itemize}
\item Let $\epsilon'=\epsilon/4(\msglen+1)$ and $\delta'=2^{-\lambda}$.
\item Parse $(\qstateq,\mat{U})\gets\qstate{C}$.
\item Let $\cP$ be defined in the same way as that in \cref{def:unrem_prf} and $D_{\iop,i}$ be the following distribution for every $i\in[\msglen+1]$.
\begin{description}
\item[$D_{\iop,i}$:] Output $(\gamma,x,y)\gets\Sim(\xk,\iop,i)$.
\end{description}
\item Compute $\tlp_{\msglen+1} \gets \API_{\cP,D_{\iop,\msglen+1}}^{\epsilon' ,\delta'}(\qstateq)$. If $\tlp_{\msglen+1}<\frac{1}{2}+\epsilon-4\epsilon'$, return $\unmarked$. Otherwise, letting $\qstateq_0$ be the post-measurement state, go to the next step.
\item For all $i \in [\msglen]$, do the following.
\begin{enumerate}
\item Compute $\tlp_{i} \gets \API_{\cP,D_{\iop,i}}^{\epsilon' ,\delta'}(\qstateq_{i-1})$. Let $\qstateq_i$ be the post-measurement state.
\item If $\tlp_i>\frac{1}{2}+\epsilon-4(i+1)\epsilon'$, set $\msg'_i=0$. If $\tlp_i<\frac{1}{2}-\epsilon+4(i+1)\epsilon'$, set $\msg'_i=1$. Otherwise, exit the loop and output $\msg'=0^{\msglen}$.
\end{enumerate}
\item Output $\msg'=\msg'_1 \concat \cdots \concat \msg'_{\msglen}$.
\end{itemize}
\end{description}

We have the following theorems.

\begin{theorem}\label{thm:watermarking_from_extraction-less_pseudorandomness}
If $\ELWMPRF$ satisfies extended weak pseudorandomness against authority, then so does $\WMPRF$.
\end{theorem}

\begin{theorem}\label{thm:watermarking_from_extraction-less}
If $\ELWMPRF$ is an extraction-less watermarking PRF that satisfies QSIM-MDD security, $\WMPRF$ is a privately extractable watermarking PRF.
\end{theorem}

\begin{theorem}\label{thm:watermarking_from_extraction-less_public}
If $\ELWMPRF$ is an extraction-less watermarking PRF that satisfies SIM-MDD security with public simulation, $\WMPRF$ is a publicly extractable watermarking PRF.
\end{theorem}

It is clear that \cref{thm:watermarking_from_extraction-less_pseudorandomness} holds since the evaluation algorithm of $\WMPRF$ is the same as that of $\ELWMPRF$ and extended weak pseudorandomness is insensitive to how the marking and extraction algorithms are defined.
Thus, we omit a formal proof.

The proofs of \cref{thm:watermarking_from_extraction-less,thm:watermarking_from_extraction-less_public} are almost the same.
Thus, we only provide the proof for the former, and omit the proof for the latter.

\begin{proof}[Proof of~\cref{thm:watermarking_from_extraction-less}]
Let $\epsilon>0$.
Let $\qA$ be a QPT adversary attacking the unremovability of $\WMPRF$.
The description of  $\expt{\qA,\WMPRF}{nrmv}(\secp,\epsilon)$ is as follows.

\begin{enumerate}
\item The challenger generates $(\pp,\xk) \chosen \Setup(1^\secp)$ and gives $\pp$ to the adversary $\qA$.
$\qA$ sends $\msg\in\bit^\msglen$ to the challenger.
The challenger generates $(\prfk,\iop)\gets\Gen(\pp)$, computes $\tlC\gets\Mark(\pp,\prfk,\msg\|0)$, and sends $\tlC$ to $\qA$.

\item $\qA$ can access to the following oracle.
\begin{description}
\item[$\Oracle{ext}$:]On input $\iop'$ and a quantum circuit $\qstate{C}$, it returns $\qExtract(\xk,\qstate{C},\iop',\epsilon)$.
\end{description}

\item Finally, the adversary outputs a quantum circuit $\pirateC=(\qstateq,\mat{U})$.
\end{enumerate}

We define $D$, $\cP$, $\cM_D$, and the three events $\Live$, $\GoodExt$, and $\BadExt$ in the same way as \cref{def:unrem_prf}.

\paragraph{The proof of $\Pr[\GoodExt]\geq\Pr[\Live]-\negl(\secp)$.}
$\qExtract$ outputs $\unmarked$ if and only if $\tlp_{\ell+1}<\frac{1}{2}+\epsilon-4\epsilon'$, that is we have $\Pr[\GoodExt]=\Pr[\tlp_{\ell+1}\geq\frac{1}{2}+\epsilon-4\epsilon']$.
Let $p$ the probability obtained by applying $\projimp(\cM_D)$ to $\qstateq$.
Then, we have $\Pr[\Live]=\Pr[p\geq\frac{1}{2}+\epsilon]$.
Let $\tlp$ be the outcome obtained if we apply $\API_{\cP,D}^{\epsilon' ,\delta'}$ to $\qstateq$.
From the property of $\API$, we have 
\begin{align}
\Pr[\Live]=\Pr[p\geq\frac{1}{2}+\epsilon]\leq\Pr[\tlp\geq\frac{1}{2}+\epsilon-\epsilon']+\negl(\secp).
\end{align}
$D$ and $D_{\iop,\msglen+1}$ are computationally indistinguishable from the QSIM-MDD security of $\ELWMPRF$ since outputs of $\Sim(\xk,\iop,i)$ is indistinguishable from those of $D$ if $\msg[i]=0$.
This indistinguishability holds even under the existence of $\Oracle{api}$.
Then, from \cref{cor:cind_sample_api}, we have
\begin{align}
\Pr[\tlp\geq\frac{1}{2}+\epsilon-\epsilon']\leq\Pr[\tlp_{\ell+1}\geq\frac{1}{2}+\epsilon-4\epsilon']+\negl(\secp)=\Pr[\GoodExt]+\negl(\secp).
\end{align}
By combining the above two equations, we obtain $\Pr[\GoodExt]\geq\Pr[\Live]-\negl(\secp)$.

The reason $D$ and $D_{\iop,\ell+1}$ need to be computationally indistinguishable under the existence of $\Oracle{api}$ to apply \cref{cor:cind_sample_api} is as follows.
In this application of \cref{cor:cind_sample_api}, the quantum state appeared in the statement of it is set as $\qstateq$ contained in the quantum circuit $\qstate{C}$ output by $\qA$.
Then, \cref{cor:cind_sample_api} (implicitly) requires that $D$ and $D_{\iop,\ell+1}$ be indistinguishable for distinguishers who can construct $\qstateq$.
To construct $\qstateq$, we need to execute $\qA$ who can access to $\Oracle{ext}$ in which $\API$ is repeatedly executed.
This is the reason $D$ and $D_{\iop,\ell+1}$ need to be indistinguishable under the existence of $\Oracle{api}$.

\paragraph{The proof of $\Pr[\BadExt]\leq\negl(\secp)$.}
We define the event $\BadExt_i$ as follows for every $i\in[\msglen]$.
\begin{description}
\item[$\BadExt_i$:]When Running $\qExtract(\xk,\pirateC,\iop^*,\epsilon)$, the following conditions hold.
\begin{itemize}
\item $\tlp_{\ell+1}\geq\frac{1}{2}+\epsilon-4\epsilon'$ holds.
\item $\msg'_j=\msg_j$ holds for every $j\in[i-1]$.
\item $\qExtract$ exits the $i$-th loop or $\msg'_i\neq\msg_i$ holds.
\end{itemize}
\end{description}
Then, we have $\Pr[\BadExt]\leq\sum_{i\in[\ell]}\Pr[\BadExt_i]$.
Below, we estimate $\Pr[\BadExt_i]$.

We first consider the case of $\msg_{i-1}=0$ and $\msg_i=0$.
Assume $\msg'_{i-1}=\msg_{i-1}=0$ holds.
Then, we have $\tlp_{i-1}>\frac{1}{2}+\epsilon-4i\epsilon'$.
Let $\tlp'_{i-1} \gets \API_{\cP,D_{\iop,i-1}}^{\epsilon' ,\delta'}(\qstateq_{i-1})$.
From, the almost-projective property of $\API$, we have
\begin{align}
\Pr[\tlp'_{i-1}>\frac{1}{2}+\epsilon-4i\epsilon'-\epsilon']\geq1-\delta'.
\end{align}
When $\msg_{i-1}=0$ and $\msg_i=0$, $D_{\iop,i-1}$ and $D_{\iop,i}$ are computationally indistinguishable since both of them are computationally indistinguishable from $D$ by the QSIM-MDD security of $\ELWMPRF$.
This indistinguishability holds under the existence of $\Oracle{api}$.
Thus, from \cref{cor:cind_sample_api}, we have
\begin{align}
1-\delta'\leq
\Pr[\tlp'_{i-1}>\frac{1}{2}+\epsilon-(4i+1)\epsilon']
\leq\Pr[\tlp_{i}>\frac{1}{2}+\epsilon-4(i+1)\epsilon']+\negl(\secp).
\end{align}
This means that $\Pr[\BadExt_i]=\negl(\secp)$ in this case.
Note that the reason the indistinguishability of $D_{\iop,i-1}$ and $D_{\iop,i}$ needs to hold under $\Oracle{api}$ is that \cref{cor:cind_sample_api} requires it hold for distinguishers who can construct $\qstateq_{i-1}$.

Next, we consider the case of $\msg_{i-1}=0$ and $\msg_i=1$.
Assume $\msg'_{i-1}=\msg_{i-1}=0$ holds.
Then, we have $\tlp_{i-1}>\frac{1}{2}+\epsilon-4i\epsilon'$.
We then define an additional distribution $\Drev_{\iop,i}$ as follows.
\begin{description}
\item[$\Drev_{\iop,i}$:] Generate $(\gamma,x,y)\gets\Sim(\xk,\iop,i)$. Output $(1\oplus\gamma,x,y)$.
\end{description}
That is, the first bit of the output is flipped from $\D_{\iop,i}$.
Then, for any random coin $r$, we have $(\mat{P}_{\Drev_{\iop,i}(r)},\mat{Q}_{\Drev_{\iop,i}(r)})=(\mat{Q}_{\D_{\iop,i}(r)},\mat{P}_{\D_{\iop,i}(r)})$.
This is because we have $\mat{Q}_{b,x,y}=\mat{I}-\mat{P}_{b,x,y}=\mat{P}_{1\oplus b,x,y}$ for any tuple $(b,x,y)$.
Therefore, $ \API_{\cP,\Drev_{\iop,i-1}}^{\epsilon' ,\delta'}$ is exactly the same process as  $\API_{\cPrev,\D_{\iop,i-1}}^{\epsilon' ,\delta'}$.
Let $\tlp'_{i-1} \gets \API_{\cP,\Drev_{\iop,i-1}}^{\epsilon' ,\delta'}(\qstateq_{i-1})$.
From, the reverse-almost-projective property of $\API$, we have
\begin{align}
\Pr[\tlp'_{i-1}<\frac{1}{2}-\epsilon+4i\epsilon'+\epsilon']\geq1-\delta'.
\end{align}
When $\msg_{i-1}=0$ and $\msg_i=1$, $\Drev_{\iop,i-1}$ and $D_{\iop,i}$ are computationally indistinguishable since both of them are computationally indistinguishable from the following distribution $\Drev$ by the QSIM-MDD security of $\ELWMPRF$.
\begin{description}
\item[$\Drev$:] Generate $(\gamma,x,y)\gets D$. Output $(1\oplus\gamma,x,y)$.
\end{description}
This indistinguishability holds under the existence of $\Oracle{api}$.
Thus, from \cref{cor:cind_sample_api}, we have
\begin{align}
1-\delta'\leq
\Pr[\tlp'_{i-1}<\frac{1}{2}-\epsilon+(4i+1)\epsilon']
\leq\Pr[\tlp_{i}<\frac{1}{2}-\epsilon+4(i+1)\epsilon']+\negl(\secp).
\end{align}
This means that $\Pr[\BadExt_i]=\negl(\secp)$ also in this case.
Note that the reason the indistinguishability of $\Drev_{\iop,i-1}$ and $D_{\iop,i}$ needs to hold under $\Oracle{api}$ is that \cref{cor:cind_sample_api} requires it hold for distinguishers who can construct $\qstateq_{i-1}$.

Similarly, we can prove that $\Pr[\BadExt_i]=\negl(\secp)$ holds in the case of $(\msg_{i-1},\msg_i)=(1,0)$ and $(\msg_{i-1},\msg_i)=(1,1)$.

Overall, we see that $\Pr[\BadExt]=\negl(\secp)$ holds in all cases.
\end{proof}

%% file: section/elwmprf_from_LWE_simple_no_mk.tex
\newcommand{\tbectlen}{\ell_{\mathsf{outct}}}
\newcommand{\tberlen}{\ell_{\mathsf{outr}}}
\newcommand{\tlprf}{\widetilde{\prf}}
\newcommand{\sfin}{\mathsf{in}}
\newcommand{\sfout}{\mathsf{out}}
\newcommand{\inrand}{r_{\sfin}}
\newcommand{\inctlen}{\ell_\mathsf{inct}}
\renewcommand{\SKE}{\algo{SKE}}
\renewcommand{\skekey}{\keys{ske}.\mathsf{k}}
\newcommand{\aeplen}{\ell_{\mathsf{ske}}}
\newcommand{\simprf}{\widehat{\prf}}
\newcommand{\simprfg}{\widehat{\prfg}}
\newcommand{\Diop}{D_{\mathtt{auth}}}
\newcommand{\cViop}{\cV_{\mathtt{auth}}}

\section{Extraction-Less Watermarking PRF from LWE}\label{sec:extless_watermarking_LWE}

We present an extraction-less watermarking PRF, denoted by $\PRF_\cprf$, whose message space is $\zo{\msglen}$ with domain $\zo{\inplen}$ and range $\zo{\outlen}$.
We use the following tools, which can be instantiated with the QLWE assumption (See~\cref{thm:CCA_QLWE,thm:pcprf_lwe,thm:pseudorandom_ske}):
\begin{itemize}
 \item Private CPRF $\CPRF=(\CPRF.\Setup,\CPRF.\Eval,\CPRF.\Constrain,\CPRF.\CEval)$. For ease of notation, we denote CPRF evaluation circuit $\CPRF.\Eval(\msk,\cdot)$ and constrained evaluation circuits $\CPRF.\CEval(\sk_f,\cdot)$ by $\prfg: \zo{\inplen}\ra \zo{\outlen}$ and $\prfg_{\notin \cV}: \zo{\inplen}\ra\zo{\outlen}$, respectively, where $x \in \cV$ iff $f(x)=1$.
 \item SKE scheme $\SKE=(\SKE.\Gen,\SKE.\Enc,\SKE.\Dec)$. The plaintext space and ciphertext space of $\SKE$ are $\zo{\aeplen}$ and $\zo{\inplen}$, respectively, where $\aeplen = \log{\msglen}+1$.
 \item PKE scheme $\PKE=(\Gen,\Enc,\Dec)$. The plaintext space of PKE is $\zo{2\secp}$.
 \end{itemize}
 
\paragraph{Construction overview.}
We already explained the high-level idea for how to realize extraction-less watermarking PRFs in \cref{sec:technical_overview}.
However, the construction of $\PRF_\cprf$ requires some additional efforts.
Thus, before providing the actual construction, we provide a high-level overview of $\PRF_\cprf$.

Recall that letting $\tlC\gets\Mark(\pp,\prfk,\msg)$ and $(\gamma^\ast,x^\ast,y^\ast)\gets\Sim(\xk,\iop,i^\ast)$, we have to design $\Sim$ and $\tlC$ so that
\begin{itemize}
\item If $\gamma=\msg[i^\ast]$, $\tlC(x^\ast)$ outputs a value different from $y^\ast$.
\item If $\gamma\ne\msg[i^\ast]$, $\tlC(x^\ast)$ outputs $y^\ast$.
\end{itemize}
In the token-based construction idea, we achieve these conditions by setting $x^\ast$ as an encryption of $y^\ast\|i^\ast\|\gamma^\ast$ and designing $\tlC$ as a token such that it outputs $y^\ast$ if the input is decryptable and $\gamma^\ast\ne\msg[i^\ast]$ holds for the decrypted value $y^\ast\|i^\ast\|\gamma^\ast$, and otherwise behaves as the original evaluation circuit.
However, in $\PRF_\cprf$, we use a constrained evaluation circuit of $\CPRF$ as $\tlC$, and thus we cannot program output values for specific inputs.
Intuitively, it seems that $\Sim$ needs to use the original PRF key $\prfk$ to achieve the above two conditions.

To solve the issue, we adopt the idea used by Quach et al.~\cite{TCC:QuaWicZir18}.
In $\PRF_\cprf$, the setup algorithm $\Setup$ generates $(\pk,\sk)\gets\Gen(1^\secp)$ of $\PKE$, and sets $\pp=\pk$ and $\xk=\sk$.
Then, the PRF key generation algorithm is given $\pk$, generates $\prfg \lrun \CPRF.\Setup(1^\secp,1^\kappa)$ along with $\skekey\gets \SKE.\Gen(1^
\secp)$, and sets the public tag $\iop$ as an encryption of $(\prfg,\skekey)$ under $\pk$.
The evaluation algorithm of $\PRF_\cprf$ is simply that of $\CPRF$.

Now, we explain how to design $\Sim$ and $\tlC\gets\Mark(\pp,\prfk,\msg)$ to satisfy the above two conditions.
Given $\xk=\sk$, $\iop=\Enc(\pk,\prfk)$ and $i$,  $\Sim$ is able to extract $\prfk=(\prfg,\skekey)$.
Then, $\Sim$ generates $\gamma\gets\bit$ and sets $x\gets\SKE.\Enc(\skekey,i\|\gamma)$ and $y\gets\prfg(x)$.
We set $\tlC$ as a constrained version of $\prfg$ for a circuit $D$ that outputs $1$ if the input $x$ is decryptable by $\skekey$ and $\gamma=\msg[i]$ holds for decrypted value $i\|\gamma$, and otherwise outputs $0$.
For an input $x$, the constrained version of $\prfg$ outputs the correct output $\prfg(x)$ if and only if $D(x)=0$.
We can check that $\PRF_\cprf$ satisfies the above two conditions.

The above construction does not satisfy extended weak pseudorandomness against authority since the authority can extract the original CPRF key $\prfg$ by $\xk=\sk$.
However, this problem can be fixed by constraining $\prfg$.
We see that $\Sim$ needs to evaluate $\prfg$ for valid ciphertexts of $\SKE$.
Thus, to implement the above mechanism, it is sufficient to set the public tag $\iop$ as an encryption of $\skekey$ and a constrained version of $\prfg$ for a circuit $\Diop$ that output $0$ if and only if the input is decryptable by $\skekey$.
Then, the authority can only extract such a constrained key.
By requiring sparseness for $\SKE$, the constrained key cannot be used to break the pseudorandomness of $\PRF_\cprf$ for random inputs.
This means that $\PRF_\cprf$ satisfies extended weak pseudorandomness against an authority.
Note that we only need a single-key CPRF for $\PRF_\cprf$ since either a user or the authority (not both) is a malicious entity in security games.

The description of $\PRF_\cprf$ is as follows.
\begin{description}

 \item[$\Setup(1^\secp)$:] $ $
 \begin{itemize}
 \item Generate $(\pk,\sk)\lrun \Gen(1^\secp)$.
 \item Output $(\pp,\xk)\seteq (\pk,\sk)$.
 \end{itemize}
 \item[$\Gen(\pp)$:] $ $
 \begin{itemize}
 \item Parse $\pp=\pk$.
 \item Generate $\prfg \lrun \CPRF.\Setup(1^\secp,1^\kappa)$. In our construction, $\kappa$ is the size of circuit $D[\skekey,\msg]$ described in~\cref{fig:description_D_lwe}, which depends on $\msglen$ (and $\secp$).
 \item Generate $\skekey \lrun \SKE.\Gen(1^\secp)$.
 \item Construct a circuit $\Diop[\skekey]$ described in~\cref{fig:description_Diop}.
 \item Compute $\prfg_{\notin \cViop} \seteq \CPRF.\Constrain(\prfg,\Diop[\skekey])$, where $\cViop \subset \zo{\inplen}$ is a set such that $x \in \cViop$ iff $\Diop[\skekey](x)=1$.
 \item Output $\prfk \seteq (\prfg,\skekey)$ and $\iop\chosen \Enc(\pk,(\prfg_{\notin \cViop},\skekey))$.
 \end{itemize}
 \item[$\Eval(\prfk,x\in\zo{\inplen})$:] Recall that $\prfg$ is a keyed CPRF evaluation circuit.
 \begin{itemize}
\item Parse $\prfk= (\prfg,\skekey)$.
 \item Output $y\seteq \prfg(x)$.
 \end{itemize}
\item[$\Mark (\pp,\prfk, \msg)$:] $ $
\begin{itemize}
\item Parse $\pp=\pk$ and $\prfk =(\prfg,\skekey)$.
\item Construct a circuit $D[\skekey,\msg]$ described in~\cref{fig:description_D_lwe}.
\item Compute $\prfg_{\notin \cV} \gets \CPRF.\Constrain(\prfg,D[\skekey,\msg])$, where $\cV \subset \zo{\inplen}$ is a set such that $x \in \cV$ iff $D[\skekey,\msg](x)=1$.
\item Output $\tlC = \prfg_{\notin \cV}$.
\end{itemize}
\item[$\Sim(\xk,\iop,i)$:] $ $
\begin{itemize}
\item Parse $\xk=\sk$.
\item Compute $(\prfg_{\notin \cViop},\skekey) \lrun \Dec(\sk,\iop)$.
\item Choose $\gamma \chosen \zo{}$.
\item Compute $x \lrun \SKE.\Enc(\skekey, i\concat \gamma)$ and $y \gets \prfg_{\notin \cViop}(x)$.
\item Output $(\gamma,x,y)$.
\end{itemize}
\end{description}

\protocol
{Circuit $\Diop[\skekey]$}
{The description of $\Diop$}
{fig:description_Diop}
{
\begin{description}
\setlength{\parskip}{0.3mm} 
\setlength{\itemsep}{0.3mm} 
\item[Constants:] An SKE key $\skekey$, and a message $\msg$.
\item[Input:] A string $x \in \zo{\inplen}$.
\end{description}
\begin{enumerate}
\item Compute $d \lrun \SKE.\Dec(\skekey,x)$.
\item Output $0$ if $d\ne \bot$ and $1$ otherwise.
\end{enumerate}
}

\protocol
{Circuit $D[\skekey,\msg]$}
{The description of $D$}
{fig:description_D_lwe}
{
\begin{description}
\setlength{\parskip}{0.3mm} 
\setlength{\itemsep}{0.3mm} 
\item[Constants:] An SKE key $\skekey$, and a message $\msg$.
\item[Input:] A string $x \in \zo{\inplen}$.
\end{description}
\begin{enumerate}
\item Compute $d \lrun \SKE.\Dec(\skekey,x)$.
\item If $d\ne \bot$, do the following
\begin{enumerate}
\item Parse $d= i\concat \gamma$, where $i\in [\msglen]$ and $\gamma\in\zo{}$.
\item If $\gamma=\msg[i]$, output $1$. Otherwise, output $0$.
\end{enumerate}
\item Otherwise output $0$.
\end{enumerate}
}

The evaluation correctness of $\PRF_\cprf$ follows from the sparseness of $\SKE$ and the correctness of $\CPRF$.
For the security of $\PRF_\cprf$, we have the following theorems.

\begin{theorem}\label{thm:extraction-less_watermarking_pcprf}
$\SKE$ is a secure SKE scheme with pseudorandom ciphertext, $\CPRF$ is a selectively single-key private CPRF, $\PKE$ is a CCA secure PKE scheme, then $\PRF_\cprf$ is an extraction-less watermarking PRF satisfying SIM-MDD security.
\end{theorem}

\begin{theorem}\label{thm:pcprf_extended_weak_pseudorandomness}
If $\CPRF$ is a selective single-key private CPRF, $\PRF_\cprf$ satisfies extended weak pseudorandomness.
\end{theorem}

\ifnum\submission=0
\input{section/proofs_elwmprf_lwe}

\else
We prove \cref{thm:extraction-less_watermarking_pcprf,thm:pcprf_extended_weak_pseudorandomness} in \cref{sec:proofs_elwmprf_lwe}.
\fi

%% file: section/proofs_elwmprf_lwe.tex

\ifnum\submission=1
\section{Security Proof for $\PRF_\cprf$}\label{sec:proofs_elwmprf_lwe}
\else\fi

\paragraph{SIM-MDD security.} First, we prove the SIM-MDD security of $\PRF_\cprf$.
\begin{proof}[Proof of~\cref{thm:extraction-less_watermarking_pcprf}]
We define a sequence of hybrid games to prove the theorem.
\begin{description}
\item[$\hybi{0}$:] This is the same as the case $\coin=1$ in $\expb{i^\ast,\qA,\PRF_\cprf}{sim}{mdd}(\secp)$. In this game, $\qA$ is given $\iop\gets \Enc(\pk,(\prfg_{\notin \cViop},\skekey))$ and $\prfg_{\notin \cV} \gets \CPRF.\Constrain(\prfg,D[\skekey,\msg])$ as a public tag and a marked circuit.
After $\iop$ and $\prfg_{\notin \cV}$ are given, $\qA$ can access to $\Oracle{sim}$.
Finally, after finishing the access to $\Oracle{sim}$, $\qA$ is given $(\gamma^\ast,x^\ast,y^\ast))$ as the challenge tuple and outputs $\coin'\in\bit$, where $\gamma^\ast \chosen \zo{}$, $x^\ast \lrun \SKE.\Enc(\skekey, i^\ast \concat \gamma^\ast)$, and $y^\ast \gets\prfg_{\notin \cViop}(x^\ast)$.
\item[$\hybi{1}$:] This is the same as $\hybi{0}$ except for the following two changes.
First, $\prfg$ is used instead of $\prfg_{\notin \cViop}$ when generating the challenge tuple $(\gamma^\ast,x^\ast,y^\ast)$.
Second, we change the behavior of $\Oracle{sim}$ as follows.
When $\qA$ sends $\iop'$ and $i'$ to $\Oracle{sim}$, if $\iop'=\iop$, $\Oracle{sim}$ performs the remaining procedures by using $(G,\skekey)$ (without decrypting $\iop'=\iop$).
 \item[$\hybi{2}$:] This is the same as $\hybi{1}$ except that we use $\iop \lrun \Enc(\pk,0^{2\secp})$ instead of $\iop \gets \Enc(\pk,(\prfg_{\notin \cViop},\skekey))$.
 
 \item[$\hybi{3}$:] This is the same as $\hybi{2}$ except that if $\msg[i^\ast] = \gamma^\ast$, we use $y^\ast \gets \zo{\outlen}$ instead of $y^\ast \gets \prfg(x^\ast)$.

 \item[$\hybi{4}$:] This is the same as $\hybi{3}$ except that we use a simulated $(\stinfo_\Sim,\simprfg) \lrun \CPRF.\Sim_1(1^\kappa,1^\secp)$ instead of $\prfg_{\notin\cV} \lrun \CPRF.\Constrain(\prfg,D[\skekey,\msg])$ for the challenge marked circuit.
 Also, if $\msg[i^\ast]\ne \gamma^\ast$, the challenger computes $y^\ast \gets \Sim_2(\stinfo_\Sim,x^\ast,0)$.
 In addition, we also change the behavior of $\Oracle{sim}$ as follows. 
Given $\iop'$ and $i'$, if $\iop'\ne\iop$, $\Oracle{sim}$ answers in the same way as $\hybi{3}$.
Otherwise, it returns $(\gamma,x,y)$, where $\gamma\gets\bit$, $x \lrun \SKE.\Enc(\skekey,i'\concat \gamma)$, and $y\gets\Sim_2(\stinfo_\Sim,x,1)$ if $\msg[i']=\gamma$ and $y\gets\Sim_2(\stinfo_\Sim,x,0)$ otherwise.




\item[$\hybi{5}$:] This is the same as $\hybi{4}$ except that we use $x^\ast \chosen \zo{\inplen}$ instead of $x^\ast \lrun \SKE.\Enc(\skekey, i^\ast \concat \gamma^\ast)$.


\item[$\hybi{6}$:]
We undo the change at $\hybi{4}$.

 \item[$\hybi{7}$:]We undo the change at $\hybi{2}$.
 \item[$\hybi{8}$:]We undo the change at $\hybi{1}$. This is the same as the case $\coin=0$ in $\expb{i^\ast,\qA,\PRF_\cprf}{sim}{mdd}(\secp)$.
\end{description}

\begin{proposition}\label{prop:pcprf_zero_one}
If $\CPRF$, $\SKE$, and $\PKE$ are correct, it holds that $\abs{\Pr[\hybi{0}=1]- \Pr[\hybi{1}=1]}\le \negl(\secp)$.
\end{proposition}
\begin{proof}[Proof of~\cref{prop:pcprf_zero_one}]
For the first change, $x^\ast\notin\cViop$ holds since $\bot\ne\SKE.\Dec(\skekey,x^\ast)$ from the correctness of $\SKE$.
Then, from the correctness of $\CPRF$, we have $\prfg_{\notin \cViop}(x^\ast)=\prfg(x^\ast)$, and thus the first change does not affect the view of $\qA$.
For the second change, from the correctness of $\PKE$, $(\prfg_{\notin \cViop},\skekey)\gets\Dec(\sk,\iop')$ if $\iop'=\iop$.
Then, similarly to the first change, we can see that the second change does not affect the view of $\qA$.
\end{proof}

\begin{proposition}\label{prop:pcprf_one_two}
If $\PKE$ is CCA secure, it holds that $\abs{\Pr[\hybi{1}=1]- \Pr[\hybi{2}=1]}\le \negl(\secp)$.
\end{proposition}
\begin{proof}[Proof of~\cref{prop:pcprf_one_two}]
We construct an algorithm $\qB$ that breaks CCA security of $\PKE$ by using $\qA$. 
\begin{description}
\item[$\qB$:] $\qB$ is given $\pk$ from the challenger. $\qB$ generates $\skekey \gets \SKE.\Gen(1^\secp)$, and $\prfg \gets \CPRF.\Setup(1^\secp)$, and sets $\pp \seteq \pk$. $\qB$ sends $\pp$ to $\qA$ and obtain $\msg$ from $\qA$. $\qB$ then generate $\prfg_{\notin \cViop} \seteq \CPRF.\Constrain(\prfg,\Diop[\skekey])$, sets $(m_0,m_1)\seteq ((\prfg_{\notin \cViop},\skekey), 0^{2\secp})$ as the challenge plaintext of the CCA game and receives $\iop$ from its challenger.
$\qB$ also constructs $D[\skekey,\msg]$, generates $\prfg_{\notin \cV} \gets \CPRF.\Constrain(\prfg,D[\skekey,\msg])$, and sends $\iop$ and $\prfg_{\notin \cV}$ to $\qA$ as the challenge public tag and marked circuit.
\begin{description}
\item[$\Oracle{sim}$:] When $\qA$ sends $\iop'$ and $i'$ to $\Oracle{sim}$, $\qB$ simulates the answer by using $\prfg$, $\skekey$, and the decryption oracle $\Dec(\sk,\cdot)$.
\end{description}

After finishing $\qA$'s oracle access to $\Oracle{sim}$, $\qB$ chooses $\gamma^\ast \chosen \zo{}$, generates $x^\ast \gets \SKE.\Enc(\skekey,i^\ast\concat \gamma^\ast)$ and $y^\ast\gets\prfg_{\notin \cViop}(x^\ast)=\prfg(x^\ast)$, and sends $(\gamma^\ast,x^\ast,y^\ast)$ to $\qA$. Note that $\prfg_{\notin \cViop}(x^\ast)=\prfg(x^\ast)$ holds since $\bot\ne\SKE.\Dec(\skekey,x^\ast)$ and thus $x^\ast\notin\cViop$.

Finally, when $\qA$ terminates with output $\coin'$, $\qB$ outputs $\coin'$ and terminates.
\end{description}

$\qB$ perfectly simulates if $\hybi{1}$ if $\iop\gets \Enc(\pk,(\prfg_{\notin \cViop},\skekey))$, and $\hybi{2}$ if $\iop\gets \Enc(\pk,0^{2\secp})$.
 This completes the proof.
\end{proof}

\begin{proposition}\label{prop:pcprf_two_three}
If $\CPRF$ satisfies selective pseudorandomness, it holds that
\[\abs{\Pr[\hybi{2}=1]- \Pr[\hybi{3}=1]}\le \negl(\secp).\]
\end{proposition}
\begin{proof}[Proof of~\cref{prop:pcprf_two_three}]
We use selective single-key pseudorandomness of $\prfg$.
%
%
We construct an algorithm $\qB$ that breaks the selective single-key pseudorandomness of $\prfg$ by using $\qA$.

\begin{description}
\item[$\qB$:] $\qB$ generates $(\pk,\sk)\gets \Gen(1^\secp)$, $\skekey \gets \SKE.\Gen(1^\secp)$, and $\iop \gets \Enc(\pk,0^{2\secp})$, sets and sends $\pp \seteq \pk$ to $\qA$, and obtains $\msg$ from $\qA$.
$\qB$ constructs $D[\skekey,\msg]$, sends $D[\skekey,\msg]$ to its challenger, and receives $\prfg_{\notin \cV}$. $\qB$ sends $\iop$ and $\prfg_{\notin \cV}$ to $\qA$ as the challenge public tag and marked circuit.
\begin{description}
\item[$\Oracle{sim}$:] When $\qA$ sends $\iop'$ and $i'\in[\msglen]$ to $\Oracle{sim}$, if $\iop'\ne\iop$, $\qB$ computes $(\prfg',\skekey')\gets\Dec(\sk,\iop')$, and computes and returns the answer $(\gamma,x,y)$ by using $(\prfg',\skekey')$.
If $\iop'=\iop$, $\qB$ returns $(\gamma,x,y)$ computed as follows.
$\qB$ chooses $\gamma \chosen \zo{}$, and generates $x\gets \SKE.\Enc(\skekey,i'\concat \gamma)$. $\qB$ finally sends $x$ to its PRF evaluation oracle and receives $y\gets\prfg(x)$.
\end{description}

After finishing $\qA$'s oracle access to $\Oracle{sim}$,$\qB$ sends $(\gamma^\ast,x^\ast,y^\ast)$ computed as follows to $\qA$.
$\qB$ first chooses $\gamma^\ast \chosen \zo{}$ and generates $x^\ast \gets \SKE.\Enc(\skekey,i^\ast\concat \gamma^\ast)$. If $\msg[i^\ast] =\gamma^\ast$, $\qB$ sends $x^\ast$ to its challenge oracle and receives $y^\ast$.
If $\msg[i^\ast] \ne \gamma^\ast$, $\qB$ sends $x^\ast$ to its PRF evaluation oracle and receives $y^\ast$.

Finally, when $\qA$ terminates with output $\coin'$, $\qB$ outputs $\coin'$ and terminates.
\end{description}

$\qB$ perfectly simulates $\hybi{2}$ if the challenge oracle returns $y^\ast = \prfg(x^\ast)$, and $\hybi{3}$ if it returns $y^\ast \chosen \zo{\outlen}$.
Note that in these games, $x^\ast\gets \SKE.\Enc(\skekey,i^\ast\concat \gamma^\ast)$, and thus if $\msg[i^\ast] = \gamma^\ast$, we have $x^\ast \in \cV$ ($D[\skekey,\msg](x^\ast)=1$).
 This completes the proof.
\end{proof}

\begin{proposition}\label{prop:pcprf_three_four}
If $\CPRF$ satisfies selective single-key privacy, it holds that
\[\abs{\Pr[\hybi{3}=1]- \Pr[\hybi{4}=1]}\le \negl(\secp).\]
\end{proposition}
\begin{proof}[Proof of~\cref{prop:pcprf_three_four}]
We use selective single-key privacy of $\prfg$.
We construct an algorithm $\qB$ that breaks the selective privacy of $\prfg$ by using $\qA$.

\begin{description}
\item[$\qB$:] $\qB$ generates $(\pk,\sk)\gets \Gen(1^\secp)$, $\skekey \gets \SKE.\Gen(1^\secp)$, and $\iop\gets \Enc(\pk,0^{2\secp})$, sends $\pp \seteq \pk$ to $\qA$, and obtains $\msg$ from $\qA$.
$\qB$ constructs $D[\skekey,\msg]$, sends $D[\skekey,\msg]$ to its challenger, and receives $\prfg^\ast$. $\qB$ sends $\iop$ and $\prfg^{\ast}$ to $\qA$ as the challenge public tag and marked circuit.
\begin{description}
\item[$\Oracle{sim}$:] When $\qA$ sends $\iop'$ and $i'\in[\msglen]$ to $\Oracle{sim}$, if $\iop'\ne \iop$, $\qB$ computes $(\prfg',\skekey')\gets\Dec(\sk,\iop')$ and returns the answer $(\gamma,x,y)$ computed by using $(\prfg',\skekey')$. If $\iop'= \iop$, $\qB$ returns the answer $(\gamma,x,y)$ computed as follows.
$\qB$ chooses $\gamma \chosen \zo{}$, and generates $x\gets \SKE.\Enc(\skekey,i'\concat \gamma)$. $\qB$ sends $x$ to its oracle and receives $y$.
\end{description}

After finishing $\qA$'s oracle access to $\Oracle{sim}$, $\qB$ sends $(\gamma^\ast,x^\ast,y^\ast)$ computed as follows to $\qA$.
$\qB$ chooses $\gamma^\ast \chosen \zo{}$ and generates $x^\ast \gets \SKE.\Enc(\skekey,i^\ast\concat \gamma^\ast)$. If $\msg[i^\ast] = \gamma^\ast$, $\qB$ chooses $y^\ast \chosen \zo{\outlen}$. If $\msg[i^\ast] \ne \gamma^\ast$, $\qB$ sends $x^\ast$ to its oracle and receives $y^\ast$.

Finally, when $\qA$ terminates with output $\coin'$, $\qB$ outputs $\coin'$ and terminates.
\end{description}

$\qB$ perfectly simulates $\hybi{3}$ if $\prfg^\ast = \CPRF.\Constrain(\prfg,D[\skekey,\msg])$ and $\qB$ has access to $\prfg(\cdot)$, and $\hybi{4}$ if $\prfg^\ast = \Sim_1(1^\kappa,1^\secp)$ and $\qB$ has access to $\Sim_2(\stinfo_\Sim,\cdot,D[\skekey,\msg](\cdot))$.
 This completes the proof.
\end{proof}

\begin{proposition}\label{prop:pcprf_five_six}
If $\SKE$ satisfies ciphertext pseudorandomness, it holds that
\[\abs{\Pr[\hybi{4}=1]- \Pr[\hybi{5}=1]}\le \negl(\secp).\]
\end{proposition}
\begin{proof}[Proof of~\cref{prop:pcprf_five_six}]
We construct an algorithm $\qB$ that breaks the ciphertext pseudorandomness of $\SKE$ by using $\qA$.

\begin{description}
\item[$\qB$:] $\qB$ generates $(\pk,\sk)\gets \Gen(1^\secp)$ and sends $\pp \seteq \pk$ to $\qA$, and obtains $\msg$ from $\qA$.
$\qB$ then generates $\iop\gets \Enc(\pk,0^{2\secp})$ and $(\stinfo_\Sim,\simprfg) \gets \Sim_1(1^\kappa,1^\secp)$, and sends $\iop$ and $\simprfg$ to $\qA$ as the challenge public tag and marked circuit.
\item[$\Oracle{sim}$:] When $\qA$ sends $\iop'$ and $i'\in[\msglen]$ to $\Oracle{sim}$, if $\iop'\ne \iop$, $\qB$ computes $(\prfg',\skekey')\gets\Dec(\sk,\iop)$ and returns the answer $(\gamma,x,y)$ computed by using $(\prfg',\skekey')$. If $\iop'= \iop$, $\qB$ returns the answer $(\gamma,x,y)$ computed as follows.
$\qB$ chooses $\gamma \chosen \zo{}$, sends $i'\concat \gamma$ to its encryption oracle, and receives $x \gets \SKE.\Enc(\skekey,i'\concat \gamma)$.
$\qB$ computes $y \gets\Sim_2(\stinfo_\Sim,x,1)$ if $\gamma= \msg[i']$ and $y \gets \Sim_2(\stinfo_\Sim,x,0)$ otherwise.

After finishing $\qA$'s oracle access to $\Oracle{sim}$, $\qB$ sends $(\gamma^\ast,x^\ast,y^\ast)$ computed as follows to $\qA$.
$\qB$ chooses $\gamma^\ast \chosen \zo{}$, sends $i^\ast\concat\gamma^\ast$ to its challenger as the challenge plaintext, and receives $x^\ast$.
$\qB$ generates $y^\ast\chosen \zo{\outlen}$ if $\msg[i^\ast] = \gamma^\ast$ and $y^\ast \gets \Sim_2(\stinfo_\Sim,x,0)$ otherwise.

Finally, when $\qA$ terminates with output $\coin'$, $\qB$ outputs $\coin'$ and terminates.
\end{description}

$\qB$ perfectly simulates $\hybi{4}$ if $x^\ast \gets \SKE.\Enc(\skekey,i^\ast\concat \gamma^\ast)$, and $\hybi{5}$ if $x^\ast \chosen \zo{\inplen}$.
 This completes the proof.
\end{proof}
\begin{proposition}\label{prop:pcprf_six_seven}
If $\CPRF$ satisfies selective single-key privacy, it holds that
\[\abs{\Pr[\hybi{5}=1]- \Pr[\hybi{6}=1]}\le \negl(\secp).\]
\end{proposition}
\begin{proof}[Proof of~\cref{prop:pcprf_six_seven}]
This proof is almost the same as that of~\cref{prop:pcprf_three_four}.
\end{proof}




\begin{proposition}\label{prop:pcprf_seven_eight}
If $\PKE$ is CCA secure, it holds that $\abs{\Pr[\hybi{6}=1]- \Pr[\hybi{7}=1]}\le \negl(\secp)$.
\end{proposition}
\begin{proof}[Proof of~\cref{prop:pcprf_seven_eight}]
This proof is almost the same as that of~\cref{prop:pcprf_one_two}.
\end{proof}

\begin{proposition}\label{prop:pcprf_eight_nine}
If $\CPRF$, $\SKE$, and $\PKE$ are correct, it holds that $\abs{\Pr[\hybi{7}=1]- \Pr[\hybi{8}=1]}\le \negl(\secp)$.
\end{proposition}
\begin{proof}[Proof of~\cref{prop:pcprf_eight_nine}]
This proof is almost the same as that of~\cref{prop:pcprf_zero_one}.
\end{proof}

By~\cref{prop:pcprf_zero_one,prop:pcprf_one_two,prop:pcprf_two_three,prop:pcprf_three_four,prop:pcprf_five_six,prop:pcprf_six_seven,prop:pcprf_seven_eight,prop:pcprf_eight_nine}, we complete the proof of~\cref{thm:extraction-less_watermarking_pcprf}.
\end{proof}

\paragraph{Extended weak pseudorandomness.} Next, we prove the extended pseudorandomness of $\PRF_\cprf$.

\begin{proof}[Proof of~\cref{thm:pcprf_extended_weak_pseudorandomness}]
Let $\qA$ be an adversary attacking the extended weak pseudorandomness of $\PRF_\cprf$.
We construct $\qB$ that attacks the selective single-key pseudorandomness of $\CPRF$.

\begin{description}
\item[$\qB$:] Given $\pp$ from $\qA$, $\qB$ first generates $\skekey\gets\SKE.\Gen(1^\secp)$, sends $\Diop[\skekey]$ to its challenger, and obtains $\prfg^\ast$.
$\qB$ sets $\pp:=\pk$, generates $\iop\gets\Enc(\pk,(\prfg^\ast,\skekey))$, and sends it to $\qA$.
$\qB$ answers $\qA$'s queries as follows.
\begin{description}
\item[$\Oracle{wprf}$:] When this is invoked (no input), $\qB$ generates $a\gets\zo{\inplen}$, sends it to its evaluation oracle, and obtains $b$. Then, $\qB$ returns $(a,b)$ to $\qA$.
\item[$\Oracle{chall}$:] When this is invoked (no input), $\qB$ generates $a^\ast\gets\zo{\inplen}$, outputs $a^\ast$ as its challenge input, and obtains $b^\ast$. $\qB$ returns $(a^\ast,b^\ast)$ to $\qA$. Note that this oracle is invoked only once.
\end{description}
When $\qA$ terminates with output $b^\prime$, $\qB$ outputs $b^\prime$ and terminates.
\end{description}
Due to the sparseness of $\SKE$, without negligible probability, we have $\SKE.\Dec(\skekey,a^\ast)=\bot$ thus $\Diop[\skekey](a^\ast)=1$, and $a$ generated when answering to a query to $\Oracle{wprf}$ is different from $a^\ast$.
Therefore, without negligible probability, $\qB$ is a valid adversary against the selective single-key pseudorandomness of $\CPRF$.
When $\qB$ is valid, we see that the advantage of $\qB$ is the same as that of $\qA$.
This completes the proof.
\end{proof}

%% file: section/elwmprf_from_IO_no_mk.tex

\newcommand{\SIH}{\algo{SIH}}
\newcommand{\siGen}{\algo{siGen}}
\newcommand{\preGen}{\algo{PreGen}}
\newcommand{\prehk}{\hk_{\mathsf{pre}}}
\renewcommand{\inplen}{\ell_{\mathsf{in}}}
\renewcommand{\outlen}{\ell_{\mathsf{out}}}
\renewcommand{\sfreal}{\mathsf{Real}}
\newcommand{\seed}{s}
\newcommand{\PRGc}{\PRG(\seed)}
\newcommand{\hashh}{\mathsf{H}}
\newcommand{\hybG}{\mathsf{Game}}

\section{Extraction-Less Watermarking PRF with Public Simulation from IO}\label{sec:pub_ext_watermarking_IO}
We construct an extraction-less watermarking PRF satisfying SIM-MDD security with public simulation.
\ifnum\submission=0
We first introduce a tool.
\input{section/pe_revisited}
\subsection{Construction of Extraction-less Watermarking PRF with Public Simulation}\label{sec:extraction-less_pub_sim_from_IO}
\else
In the construction, we use puncturable encryption (PE)~\cite{SIAMCOMP:CHNVW18}.
We provide the definition of PE in \cref{sec:PE_revisited}.
\fi

We describe our extraction-less watermarking PRF $\PRF_\io$ for message space $\zo{\msglen}$ with domain $\zo{\inplen}$ and range $\zo{\outlen}$ below. We use the following tools:
\begin{itemize}
\item PPRF $\PRF=\PRF.(\Gen,\Eval,\Puncture)$. We denote a PRF evaluation circuit $\PRF.\Eval_\prfk(\cdot)$ by $\prf: \zo{\inplen}\ra \zo{\outlen}$, a PRF evaluation circuit with punctured key $\PRF.\Eval_{\prfk_{\ne x}}(\cdot)$ by $\prf_{\ne x}$ (that is, we omit $\prfk$ and simply write $\prf(\cdot)$ instead of $\prf_\prfk(\cdot)$) for ease of notations.
\item PE scheme $\PE= \PE.(\Gen,\Puncture,\Enc,\Dec)$. The plaintext and ciphertext space of PE are $\zo{\ptlen}$ and $\zo{\ctlen}$, respectively, where $\ptlen = \ell + \log{\msglen}+1$ and $\inplen \seteq \ctlen$ ($\ctlen = \poly(\ell,\log{\msglen})$).
\item Indistinguishability obfuscator $\iO$.
\item PRG $\PRG: \zo{\ell} \ra \zo{\outlen}$.
\end{itemize}
\begin{description}

 \item[$\Setup(1^\secp)$:] $ $
 \begin{itemize}
 \item Output $(\pp,\xk)\seteq (\bot,\bot)$.
 \end{itemize}
 \item[$\Gen(\pp)$:] $ $
 \begin{itemize}
 \item Parse $\pp=\bot$.
 \item Compute $\prf \lrun \PRF.\Gen(1^\secp)$.
 \item Generate $(\peek,\pedk)\gets \PE.\Gen(1^\secp)$.
 \item Output $\prfk \seteq (\prf,\pedk)$ and $\iop \seteq \peek$.
 \end{itemize}
 \item[$\Eval(\prfk,x \in \zo{\inplen})$:] $ $
 \begin{itemize}
 \item Parse $\prfk =(\prf,\pedk)$.
 \item Compute and output $y \lrun \prf(x)$.
 \end{itemize}
 \item[$\Mark (\pp,\prfk, \msg \in \zo{\msglen})$:] $ $ 
\begin{itemize}
    \item Parse $\pp =\bot$ and $\prfk=(\prf,\pedk)$.
\item Construct a circuit $D[\prf,\pedk,\msg]$ described in~\cref{fig:description_D}.
\item Compute and output $\tlC \seteq \iO(D[\prf,\pedk,\msg])$.
\end{itemize}

\item[$\Sim(\xk,\iop,i)$:] $ $
\begin{itemize}
    \item Parse $\xk=\bot$ and $\iop = \peek$.
\item Choose $\gamma \chosen \zo{}$ and $\seed \chosen \zo{\ell}$.
\item Compute $y \seteq \PRGc$.
\item Compute $x \lrun \PE.\Enc(\peek,\seed\concat i\concat \gamma)$.
\item Output $(\gamma,x,y)$
\end{itemize}
\end{description}
The size of the circuit $D$ is appropriately padded to be the maximum size of all modified circuits, which will appear in the security proof.

\protocol
{Circuit $D[\prf,\pedk,\msg]$}
{The description of $D$}
{fig:description_D}
{
\begin{description}
\setlength{\parskip}{0.3mm} 
\setlength{\itemsep}{0.3mm} 
\item[Constants:] A PRF $\prf$, a PE decryption key $\pedk$, and a message $\msg$.
\item[Input:] A string $x \in \zo{\inplen}$.
\end{description}
\begin{enumerate}
\item Compute $d \lrun \PE.\Dec(\pedk,x)$.
\item If $d \ne \bot$, do the following
\begin{enumerate}
 \item Parse $d = \seed \concat i\concat \gamma$, where $\seed\in\zo{\ell}$, $i\in [\msglen]$, and $\gamma \in \zo{}$.
 \item If $\msg[i] \ne \gamma$, output $\PRGc$. Otherwise, output $\prf(x)$.
 \end{enumerate} 
\item Otherwise, output $\prf(x)$.
\end{enumerate}
}

The evaluation correctness of $\PRF_\io$ immediately follows from the sparseness of $\PE$ and the functionality of $\iO$.\footnote{In fact, $\PRF_\io$ satisfies a stronger evaluation correctness than one written in~\cref{def:extless_watermarking_prf}. The evaluation correctness holds even for any PRF key $\prfk$ and input $x \in \Domprf$ like the statistical correctness by Cohen et al.~\cite{SIAMCOMP:CHNVW18}.}
$\PRF_\io$ trivially satisfies pseudorandomness (against an authority) since $\Setup$ outputs nothing, $\iop$ is a public key $\peek$, and $\Eval$ is independent of $(\peek,\pedk)$ ($\pedk$ is not used in $\Eval$).
Moreover, we have the following theorem.

\begin{theorem}\label{thm:io_const_extraction-less_watermarkable_PRF}
If $\PRF$ is a secure PPRF, $\PRG$ is a secure PRG, $\PE$ is a secure PE with strong ciphertext pseudorandomness, and $\iO$ is a secure IO, then $\PRF_\io$ is an extraction-less watermarking PRF satisfying SIM-MDD security with public simulation.
\end{theorem}

\ifnum\submission=0
\input{section/proof_elwmprf_io}

\else
We prove \cref{thm:io_const_extraction-less_watermarkable_PRF} in \cref{sec:proof_elwmprf_io}
\fi

%% file: section/pe_revisited.tex
\ifnum\submission=0
\subsection{Puncturable Encryption, Revisited}\label{sec:PE_revisited}
\else
\section{Puncturable Encryption, Revisited}\label{sec:PE_revisited}
\fi

Cohen et al.~\cite{SIAMCOMP:CHNVW18} introduced the notion of puncturable encryption (PE).
They used a PE scheme as a crucial building block to construct a publicly extractable watermarking PRF against classical adversaries.
We also use a PE scheme to construct an extraction-less watermarking PRF with public simulation (against quantum adversaries).
However, we find that the original PE definition is not sufficient for proving unremovability (and our purpose) since there is a subtle issue in the security proof by Cohen et al.~\cite{SIAMCOMP:CHNVW18}. However, we can fix the issue since their PE scheme satisfies a stronger security notion than what they proved.
Thus, we introduce a stronger security notion for PE in this section.

The syntax of PE is almost the same as that of the original PE.

\begin{definition}[Puncturable Encryption (Syntax)]\label{def:pe_syntax}
A puncturable encryption (PE) scheme $\PE$ for a plaintext space $\cP = \zo{\ptxtlen}$ is a triple of PPT algorithms $(\Gen, \Puncture, \Enc)$ and a deterministic algorithm $\Dec$. The ciphertext space will be $\zo{\ctlen}$ where $\ctlen = \poly(\secp,\ptxtlen)$.
\begin{description}
\item[$\Gen(1^\secp)\ra (\ek,\dk)$:] The key generation algorithm takes as input the security parameter $1^\secp$ and outputs an encryption key $\ek$ and a decryption key $\dk$.
\item[$\Puncture(\dk, \setbk{c^{\ast}}) \ra \dk_{\ne c^\ast}$:] The puncturing algorithm takes as input $\dk$ and a string $c^\ast \in \zo{\ctlen}$, and outputs a ``punctured'' decryption key $\dk_{\ne c^\ast}$. 
\item[$\Enc(\ek, m)\ra c$:]  The encryption algorithm takes as input $\ek$ and a plaintext $m \in \zo{\ptxtlen}$, and outputs a ciphertext $c$ in $\zo{\ctlen}$.
\item[$\Dec(\dk^\prime, c^\prime) \ra m^\prime \text{ or } \bot$:] The decryption algorithm takes a possibly punctured decryption key $\dk^\prime$ and a string $c^\prime \in \zo{\ctlen}$. It outputs a plaintext $m^\prime$ or the special symbol $\bot$.
\end{description}
\end{definition}

There are four security requirements on PE. Three of those are the same as those in the original PE security.
The difference is ciphertext pseudorandomness.

\begin{definition}[Puncturable Encryption Security]\label{def:pe_security}
A PE scheme $\PE = (\Gen,\Puncture, \Enc,\Dec)$ with plaintext space $\cP=\zo{\ptxtlen}$ and ciphertext space $\cC=\zo{\ctlen}$ is required to satisfy the following properties.
\begin{description}
\item[Correctness:] We require that for all plaintext $m\in \cP$ and $(\ek, \dk) \lrun \Gen(1^\secp)$, it holds that $\Dec(\dk, \Enc(\ek, m)) = m$.
\item[Punctured Correctness:] We require the same to hold for punctured keys. For all possible keys $(\ek, \dk) \lrun \Gen(1^\secp)$, all string $c^\ast \in \cC$, all punctured keys $\dk_{\ne c^\ast} \lrun \Puncture(\dk, \setbk{c^\ast})$, and all potential ciphertexts $c\in \cC \setminus \{c^\ast\}$:
\[
\Dec(\dk, c) = \Dec(\dk_{\ne c^\ast}, c).
\]

\item[Sparseness:]
We also require that most strings are not valid ciphertexts:
\begin{align}
\Pr \left [
\Dec(\dk, c) \neq \bot ~\left|~
(\ek, \dk) \lrun \Gen(1^\secp), c \chosen \zo{\ctlen} \right.
\right ] \le \negl(\secp).
\end{align}

\item[Ciphertext Pseudorandomness:] We require that PE has strong ciphertext pseudorandomness defined in~\cref{def:pe_strong_pseudorandomness}.
\end{description}
\end{definition}


\begin{definition}[Strong Ciphertext Pseudorandomness]\label{def:pe_strong_pseudorandomness}
We define the following experiment $\expb{\qA}{s}{cpr}(\secp)$.
\begin{enumerate}
\item $\qA$ sends a message $m^* \in \cP= \zo{\ptxtlen}$ to the challenger.
\item The challenger does the following:

\begin{itemize}
\item Generate $(\ek,\dk) \lrun \Gen(1^\secp)$

\item Compute a ciphertext $c^* \lrun \Enc(\ek, m^*)$.

\item Choose $r^* \chosen \cC=\zo{\ctlen}$.

\item Choose $\coin \chosen \zo{}$ and set $x_0 \seteq c^\ast$ and $x_1 \seteq r^\ast$.

\item Generate a punctured key $\dk_{\ne x_\coin} \lrun \Puncture(\dk, \{x_\coin\})$

\item Send $(x_\coin,\ek,\dk_{\ne x_\coin})$ to $\qA$:
\end{itemize}
\item $\qA$ outputs $\coin^\ast$ and the experiment outputs $1$ if $\coin = \coin^\ast$; otherwise $0$.
\end{enumerate}
We say that $\PE$ has strong ciphertext pseudorandomness if for every QPT adversary $\qA$, it holds that
\[
\advb{\qA}{s}{cpr}(\secp)\seteq 2\abs{\Pr[\expb{\cA}{s}{cpr}(\secp) \out 1] -\frac{1}{2}}  \leq \negl(\secp).
\]
\end{definition}
\begin{remark}[Difference from the original PE]\label{remark:diffrence_ct_pseudorandomness}
In the original PE definition, $\Puncture$ takes two strings $\setbk{c_0,c_1} \subset \zo{\ctlen}$ and outputs a punctured decryption key $\dk_{\notin \setbk{c_0,c_1}}$ and punctured correctness is accordingly defined.

In the original ciphertext pseudorandomness (described in~\cref{sec:original_pe_ct_pseudorandomness}), a punctured decryption key is punctured at both $c^\ast$ and $r^\ast$. That is, the information about $m^\ast$ remains in $\dk_{\notin \setbk{c^\ast,r^\ast}}$ for $\coin \in \zo{}$. This is an issue for our purpose (and the proof by Cohen et al.~\cite{SIAMCOMP:CHNVW18}). Thus, we introduce the strong ciphertext pseudorandomness, where the information about $m^\ast$ disappears in the case $\coin=1$ since the punctured decryption key is $\dk_{\ne r^\ast}$ when $\coin =1$.
\end{remark}

In fact, the PE scheme $\PE$ by Cohen~\etal~\cite{SIAMCOMP:CHNVW18} satisfies strong ciphertext pseudorandomness (and thus, we can also fix the issue in the proof by Cohen et al.\footnote{See~\cref{sec:original_pe_ct_pseudorandomness} for the detail of the issue.}).
\begin{theorem}\label{thm:CHNVW_PE_strong_ct_pseudorandom}
If there exists secure IO for circuits and the QLWE assumption holds, there exists secure PE that satisfies strong ciphertext pseudorandomness.
\end{theorem}
We prove this theorem in~\cref{sec:pe_construction}.

%% file: section/proof_elwmprf_io.tex

\ifnum\submission=0
\begin{proof}[Proof of~\cref{thm:io_const_extraction-less_watermarkable_PRF}]
\else
\section{Security Proof for $\PRF_\io$}\label{sec:proof_elwmprf_io}
\fi
We define a sequence of hybrid games.
We sometimes omit hard-coded values when we write some circuits. For example, we simply write $D$ instead of $D[\prf,\pedk,\msg]$ when hard-coded values $(\prf,\pedk,\msg)$ are not important in arguments or clear from the context.

\begin{description}
\item[$\hybi{0}$:] This is the same as the case $b=1$ in $\expc{i^\ast,\qA,\PRF_\io}{pub}{sim}{mdd}(\secp)$.  In this game, $\qA$ is given $\iop = \peek$ and $\tlC = \iO(D[\prf,\pedk,\msg])$ as a public tag and a marked circuit, where $(\peek,\pedk)\gets \PE.\Gen(1^\secp)$, $(\prf,\pedk)$ is the target PRF key, and $\msg$ is the target message from $\qA$.
Also, $\qA$ is given $(\gamma^\ast,x^\ast,y^\ast)=(\gamma_1,x_1,y_1)$ as the challenge tuple, where $\gamma_1 \chosen \zo{}$, $x_1 \lrun \PE.\Enc(\peek,\seed^\ast \concat i^\ast \concat \gamma_1)$, $y_1 \seteq \PRG(\seed^\ast)$, and $\seed^\ast \chosen \zo{\ell}$.

\item[Case {$\gamma^\ast =\msg[i^\ast]$}:] We consider two cases separately hereafter.
First, we consider the case where $\gamma^\ast =\msg[i^\ast]$. We denote these hybrid games by $\hybij{k}{=}$. Note that we can choose $\gamma^\ast$ at any time and hard-code it into $\tlD$ in the proof since it is a uniformly random bit in all hybrid games.

\item[$\hybij{1}{=}$:] This is the same as $\hybi{0}$ except that if $\gamma_1 = \msg[i^\ast]$, we use $\tlD \lrun \iO(D_{\ne x_1}^{\$}[\prf,\pedk_{\ne x_1},\msg,\gamma_1,x_1,\oly])$, where $D_{\ne x^\ast}^{\$}$ is described in~\cref{fig:description_D_punc_random} and $\oly \seteq \prf(x_1)$. We use a punctured decryption key $\pedk_{\ne x_1}$ instead of $\pedk$. However, we do \emph{not} use a punctured key for $\prf$.

\protocol
{Circuit $D_{\ne x^\ast}^{\$}[\prf,\pedk^{\prime},\msg,\gamma^\ast,x^\ast,\oly]$}
{The description of $D_{\ne x^\ast}^{\$}$ (for $\gamma^\ast = \msg[i^\ast]$)}
{fig:description_D_punc_random}
{
\begin{description}
\setlength{\parskip}{0.3mm} 
\setlength{\itemsep}{0.3mm} 
\item[Constants:] A PRF key $\prf$, a (possibly punctured) PE decryption key $\pedk^\prime$, a message $\msg$, a bit $\gamma^\ast$, and strings $x^\ast \in \zo{\inplen}, \oly \in\zo{\outlen}$.
\item[Input:] A string $x \in \zo{\inplen}$.
\end{description}
\begin{enumerate}
\item \redline{If $x = x^\ast$, output $\oly$.}
\item Compute $d \lrun \PE.\Dec(\redline{\pedk^{\prime}},x)$.
\item If $d \ne \bot$, do the following
\begin{enumerate}
 \item Parse $d =  \seed \concat i\concat \gamma$, where $\seed \in\zo{\ell}$, $i\in [\msglen]$, and $\gamma \in \zo{}$.
 \item If $\msg[i] \ne \gamma$, output $\PRG(\seed)$. Otherwise, output $\prf(x)$.
 \end{enumerate} 
\item Otherwise, output $\prf(x)$.
\end{enumerate}
}


\item[$\hybij{2}{=}$:] This is the same as $\hybij{1}{=}$ except that we generate
\begin{itemize}
\item $x_0 \chosen \zo{\inplen}$,
\item $\tlD \lrun \iO(D_{\ne x_0}^{\$}[\prf,\pedk_{\ne x_0},\msg,\gamma_1,x_0,\oly])$.
\end{itemize}
That is, we replace $x^\ast = x_1$ and $\pedk^\prime = \pedk_{\ne x_1}$ with $x^\ast = x_0$ and $\pedk^\prime = \pedk_{\ne x_0}$, respectively.

 We also rename $\gamma_1 \chosen \zo{}$ into $\gamma_0 \chosen \zo{}$ (these distributions are the same).
\item[$\hybij{3}{=}$:] This is the same as $\hybij{2}{=}$ except that we use $y_0 \chosen\zo{\outlen}$ instead of $y_1 \seteq \PRG(\seed^\ast)$.


We describe the high-level overview of hybrid games for $\gamma^\ast = \msg[i^\ast]$ in~\cref{fig:extractionless_public_hybrid_first_step,fig:extractionless_public_hybrid_equal_case}.
\item[Case {$\gamma^\ast \ne \msg[i^\ast]$}:] Next, we consider the case where $\gamma^\ast \ne \msg[i^\ast]$. We denote these hybrid games by $\hybij{k}{\ne}$.
\item[$\hybij{1}{\ne}$:] This is the same as $\hybi{0}$ except that if $\gamma_1 \ne \msg[i^\ast]$, we generate $\tlD \lrun \iO(D_{\ne x_1}^{\mathsf{real}}[\prf,\pedk_{\ne x_1},\msg,\gamma_1,x_1,y_1])$, where $D_{\ne x^\ast}^{\mathsf{real}}$ is described in~\cref{fig:description_D_punc_real} and $y_1 \seteq \PRG(\seed^\ast)$. We use a punctured decryption key $\pedk_{\ne x_1}$ instead of $\pedk$. However, we do \emph{not} use a puncture key for $\prf$ \emph{at this point}.

\protocol
{Circuit $D_{\ne x^\ast}^{\mathsf{real}}[\prf^\prime,\pedk^{\prime},\msg,\gamma^\ast,x^\ast,y^\ast]$}
{The description of $D_{\ne x^\ast}^{\mathsf{real}}$ (for $\gamma^\ast = 1-\msg[i^\ast]$)}
{fig:description_D_punc_real}
{
\begin{description}
\setlength{\parskip}{0.3mm} 
\setlength{\itemsep}{0.3mm} 
\item[Constants:] A (possibly punctured) PRF key $\prf^{\prime}$, a (possibly punctured) PE decryption key $\pedk^\prime$, a message $\msg$, a bit $\gamma^\ast$, and strings $x^\ast \in \zo{\inplen}, y^\ast\in\zo{\outlen}$.
\item[Input:] A string $x \in \zo{\inplen}$.
\end{description}
\begin{enumerate}
\item \redline{If $x = x^\ast$, output $y^\ast$.}
\item Compute $d \lrun \PE.\Dec(\redline{\pedk^{\prime}},x)$.
\item If $d \ne \bot$, do the following
\begin{enumerate}
 \item Parse $d = \seed\concat i\concat \gamma$, where $\seed\in\zo{\ell}$, $i\in [\msglen]$, and $\gamma \in \zo{}$.
 \item If $\msg[i] \ne \gamma$, output $\PRG(\seed)$. Otherwise, output \redline{$\prf^{\prime}(x)$}.
 \end{enumerate} 
\item Otherwise, output \redline{$\prf^{\prime}(x)$}.
\end{enumerate}
}

\item[$\hybij{2}{\ne}$:] This is the same as $\hybij{1}{\ne}$ except that 
\begin{itemize}
 \item $x_0 \chosen \zo{\inplen}$,
 \item $\tlD \lrun \iO(D_{\ne x_0}^{\mathsf{real}}[\prf,\pedk_{\ne x_0},\msg,\gamma_1,x_0,y_1])$.
 \end{itemize}
  That is, we replace $x^\ast = x_1$ and $\pedk^\prime = \pedk_{\ne x_1}$ with $x^\ast = x_0$ and $\pedk^\prime = \pedk_{\ne x_0}$, respectively.
 We also rename $\gamma_1 \chosen \zo{}$ into $\gamma_0 \chosen \zo{}$ (these distributions are the same).
\item[$\hybij{3}{\ne}$:] This is the same as $\hybij{2}{\ne}$ except that we use $y_0 \chosen \zo{\outlen}$ instead of $y_1 \seteq \PRG(\seed^\ast)$.
\item[$\hybij{4}{\ne}$:] This is the same as $\hybij{3}{\ne}$ except that we use $\prf_{\ne x_0}$ instead of $\prf$.
\item[$\hybij{5}{\ne}$:] This is the same as $\hybij{4}{\ne}$ except that we use $y_0 \seteq \prf(x_0)$ instead of $y_0 \chosen \zo{\outlen}$.
\item[$\hybij{6}{\ne}$:] This is the same as $\hybij{5}{\ne}$ except that we use $\prf$ instead of $\prf_{\ne x_0}$.

We describe the high-level overview of hybrid games for $\gamma^\ast \ne \msg[i^\ast]$ in~\cref{fig:extractionless_public_hybrid_first_step,fig:extractionless_public_hybrid_ne_case}.
\item[End of case analysis:] The two case analyses end. Remaining transitions are the reverse of transitions from $\hybi{0}$ to $\hybij{1}{=}$ or $\hybij{1}{\ne}$.
\item[$\hybij{4}{=}$ and $\hybij{7}{\ne}$:] These are the same as $\hybij{3}{=}$ and $\hybij{6}{\ne}$, respectively except that
\begin{itemize}
\item if $\gamma_0 = \msg[i^\ast]$, we use $\tlD \lrun \iO(D[\prf,\pedk,\msg])$ instead of $\tlD \lrun \iO(D_{\ne x_0}^{\$}[\prf,\pedk_{\ne x_0},\msg,\gamma_0,x_0,\oly])$, where $D_{\ne x^\ast}^{\$}$ is described in~\cref{fig:description_D_punc_random} and $\oly \seteq \prf(x_0)$.

\item if $\gamma_0 \ne \msg[i^\ast]$, we use $\tlD \lrun \iO(D[\prf,\pedk,\msg])$ instead of $\tlD \lrun \iO(D_{\ne x_0}^{\mathsf{real}}[\prf,\pedk_{\ne x_0},\msg,\gamma_0,x_0,y_0])$, where $D_{\ne x^\ast}^{\mathsf{real}}$ is described in~\cref{fig:description_D_punc_real} and $y_0 \seteq \prf(x_0)$.
\end{itemize}
\end{description}
The each last hybrid is the same as the case $b=0$ in $\expc{i^\ast,\qA,\PRF_\io}{pub}{sim}{mdd}(\secp)$. That is, $\qA$ is given $\tlC = \iO(D[\prf,\pedk,\msg])$ and $(\gamma_0,x_0,y_0) \lrun \Dreal{i^\ast}$. Recall that $x_0 \gets \zo{\inplen}$ and
\begin{itemize}
\item $y_0 \chosen \zo{\outlen}$ if $\gamma_0 = \msg[i^\ast]$ (see $\hybij{3}{=}$),
\item $y_0 \seteq \prf(x_0)$ if $\gamma_0 \ne \msg[i^\ast]$ (see $\hybij{5}{\ne}$).
\end{itemize}

\begin{figure}[ht]
    \centering
    \begin{tabular}{l lll l}\toprule
 & $\prfk$ &    \multicolumn{2}{c}{$\Oracle{chall}$} &  security\\\cmidrule(r){3-4}
 & &  $\gamma_1 = \msg[i^\ast]$ &$\gamma_1 \ne \msg[i^\ast]$ &  \\ \midrule
$\hybi{0}$ &    $(\prf,\pedk)$ & $\iO(D)$ & $\iO(D)$ & \\
$\hybij{1}{=}$ &  $(\prf,\redline{\pedk_{\ne x_1}})$ &  \redline{$\iO(D_{\ne x_1}^{\$})$} & N/A &  IO \& PE p-Cor.\\
$\hybij{1}{\ne}$ &  $(\prf,\redline{\pedk_{\ne x_1}})$ & N/A & \redline{$\iO(D_{\ne x_1}^{\mathsf{real}})$} & IO \& PE p-Cor.\\\bottomrule
    \end{tabular}
\caption{High-level overview of hybrid games from $\hybi{0}$ to $\hybij{1}{=}$ and $\hybij{1}{\ne}$. Note that in these hybrid games, $(\gamma_1,x_1,y_1)\gets \Sim(\xk,\iop,i^\ast)$ and $x^\ast =x_1$. We use $\pedk_{\ne x_1}$ in $D_{\ne x_1}^{\$}$ and $D_{\ne x_1}^{\mathsf{real}}$, but $\prf$ is not punctured yet.
In ``security'' column, PE p-Cor. means PE punctured correctness.}
\label{fig:extractionless_public_hybrid_first_step}
\end{figure}

\begin{figure}[ht]
    \centering
    \begin{tabular}{l lllll l}\toprule
    &\multicolumn{5}{c}{$\gamma^\ast = \msg[i]$}\\\cmidrule(r){2-6}
 & $x^\ast$ ($x_1$/$x_0$) & $y^\ast$ ($y_1$/$y_0$) & $\oly$ & $\Oracle{chall}$ & $\prfk$ & security\\\midrule
$\hybij{1}{=}$ & $\PE.\Enc(p_{i^\ast})$ & $\PRG(\seed^\ast)$ &  $\prf(x_1)$ & $\iO(D_{\ne x_1}^{\$}[x_1 \mapsto \bar{y}])$ & $(\prf,\pedk_{\ne x_1})$ & \\
$\hybij{2}{=}$ & \redline{$x_0 \chosen \$$} & $\PRG(\seed^\ast)$ &  $\prf(x_0)$ & $\iO(D_{\ne x_0}^{\$}[x_0 \mapsto \bar{y}])$ & $(\prf,\pedk_{\ne x_0})$& S-CPR\\
$\hybij{3}{=}$ & $\$$ & \redline{$y_0\chosen \$$} &  $\prf(x_0)$ & $\iO(D_{\ne x_0}^{\$}[x_0 \mapsto \bar{y}])$ & $(\prf,\pedk_{\ne x_0})$ & PRG \\
$\hybij{4}{=}$ & $\$$ & $\$$ &  $\prf(x_0)$ & \redline{$\iO(D)$} & $(\prf,\redline{\pedk})$ & IO \& PE p-Cor.\\\bottomrule
    \end{tabular}
\caption{High-level overview of hybrid games from $\hybij{1}{=}$ to $\hybij{4}{=}$. Here, $p_{i^\ast} \seteq ( \seed^\ast \concat i^\ast\concat \gamma_1)$ and $x_1\gets \PE.\Enc(\peek,p_{i^\ast})$. Note that $\oly$ is an output of $D_{\ne x^\ast}^{\$}$ for input $x^\ast$ for $\gamma_1 = \msg[i^\ast]$ case. $D_{\ne x^\ast}^{\$}[x^\ast \mapsto \oly]$ means $D_{\ne x^\ast}^{\$}(x^\ast)$ outputs the hard-coded value $\oly$.
In ``security'' column, S-CPR means the Strong Ciphertext PseudoRandomness of PE.}
\label{fig:extractionless_public_hybrid_equal_case}
\end{figure}

\begin{figure}[ht]
    \centering
    \begin{tabular}{l llll l}\toprule
    &\multicolumn{4}{c}{$\gamma^\ast \ne \msg[i]$}\\\cmidrule(r){2-5}
 & $x^\ast$ ($x_1$/$x_0$) & $y^\ast$ ($y_1$/$y_0$) &  $\Oracle{chall}$ & $\prfk$ & security\\\midrule
$\hybij{1}{\ne}$ & $\PE.\Enc(p_{i^\ast})$ & $\PRG(\seed^\ast)$ &   $\iO(D_{\ne x_1}^{\mathsf{real}}[x_1 \mapsto y_1])$ & $(\prf,\pedk_{\ne x_1})$ & \\
$\hybij{2}{\ne}$ & \redline{$x_0\chosen \$$} & $\PRG(\seed^\ast)$ &   $\iO(D_{\ne x_0}^{\mathsf{real}}[x_0 \mapsto y_1])$ & $(\prf,\pedk_{\ne x_0})$ & S-CPR\\
$\hybij{3}{\ne}$ & $\$$ & \redline{$y_0 \chosen \$$} & $\iO(D_{\ne x_0}^{\mathsf{real}}[x_0 \mapsto y_0])$ & $(\prf,\pedk_{\ne x_0})$ & PRG \\
$\hybij{4}{\ne}$ & $\$$ & $\$$ & $\iO(D_{\ne x_0}^{\mathsf{real}}[x_0 \mapsto y_0])$ & $(\redline{\prf_{\ne x_0}},\pedk_{\ne x_0})$ & IO \& PPRF p-Cor.\\
$\hybij{5}{\ne}$ & $\$$ & \redline{$\prf(x_0)$} & $\iO(D_{\ne x_0}^{\mathsf{real}}[x_0 \mapsto y_0])$ & $(\prf_{\ne x_0},\pedk_{\ne x_0})$ & PPRF \\
$\hybij{6}{\ne}$ & $\$$ & $\prf(x_0)$ & $\iO(D_{\ne x_0}^{\mathsf{real}}[x_0 \mapsto y_0])$ & $(\redline{\prf},\pedk_{\ne x_0})$ & IO \& PPRF p-Cor. \\
$\hybij{7}{\ne}$ & $\$$ & $\prf(x_0)$ & \redline{$\iO(D)$} & $(\prf,\redline{\pedk})$ & IO \& PE p-Cor. \\\bottomrule
    \end{tabular}
\caption{High-level overview of hybrid games from $\hybij{1}{\ne}$ to $\hybij{7}{\ne}$. Here, $p_{i^\ast} \seteq (\seed^\ast \concat i^\ast\concat \gamma_1)$ and $x_1 \gets \PE.\Enc(\peek,p_{i^\ast})$.
$D_{\ne x^\ast}^{\mathsf{real}}[x^\ast \mapsto y^\ast]$ means $D_{\ne x^\ast}^{\mathsf{real}}(x^\ast)$ outputs the hard-coded value $y^\ast$.
In ``security'' column, S-CPR and PPRF p-Cor. mean the Strong Ciphertext PseudoRandomness of PE and punctured correctness of PPRF, respectively.}
\label{fig:extractionless_public_hybrid_ne_case}
\end{figure}

We prove the lemma by proving the following propositions.

\paragraph{The case $\gamma^\ast = \msg[i]$.} We first prove propositions for the case $\gamma^\ast = \msg[i]$.
\begin{proposition}\label{prop:io_const_equal_first}
If $\iO$ is a secure IO and $\PE$ satisfies punctured correctness, it holds that
\ifnum\submission=1
$\abs{\Pr[\hybi{0}=1]- \Pr[\hybij{1}{=}=1]}\le \negl(\secp)$.
\else
\[\abs{\Pr[\hybi{0}=1]- \Pr[\hybij{1}{=}=1]}\le \negl(\secp).\]
\fi
\end{proposition}
\begin{proof}[Proof of~\cref{prop:io_const_equal_first}]
The difference between the two games is that $D_{\ne x_1}^{\$}[\prf,\pedk_{\ne x_1},\msg,\gamma,x_1,\oly]$ is used for $\Oracle{chall}$ instead of $D[\prf,\pedk,\msg]$ in the case where $\gamma_1=\msg[i^\ast]$. These two circuits are the same except that
\begin{itemize}
 \item for input $x_1$, $D_{\ne x_1}^{\$}$ directly outputs $\oly$,
 \end{itemize}
 due to the punctured correctness of $\PE$.
 Thus, if the following hold, $D_{\ne x_1}^{\$}$ and $D$ are functionally equivalent:
 \begin{itemize}
  \item $D(x_1)$ outputs $\oly=\prf(x_1)$ when $\gamma_1 =\msg[i^\ast]$.
  \end{itemize}
  This holds since $x_1 \gets \PE.\Enc(\peek, \seed^\ast \concat i^\ast \concat \gamma_1)$ and $D(x_1)$ runs the item (b) in~\cref{fig:description_D}, but $\gamma_1 \ne \msg[i^\ast]$ does \emph{not} hold in this case.

  Thus, $D_{\ne x_1}^{\$}$ and $D$ are functionally equivalent and the proposition holds due to IO security.
\end{proof}

\begin{proposition}\label{prop:io_const_equal_second}
If $\PE$ satisfies strong ciphertext pseudorandomness, it holds that
\ifnum\submission=1
$\abs{\Pr[\hybij{1}{=}=1]- \Pr[\hybij{2}{=}=1]}\le \negl(\secp)$.
\else
\[\abs{\Pr[\hybij{1}{=}=1]- \Pr[\hybij{2}{=}=1]}\le \negl(\secp).\]
\fi
\end{proposition}
\begin{proof}[Proof of~\cref{prop:io_const_equal_second}]
We construct an algorithm $\qB$ for strong ciphertext pseudorandomness by using $\qA$.
$\qB$ generates $\prf \gets \PRF.\Gen(1^\secp)$, and chooses $\seed^\ast \chosen \zo{\ell}$ and $\gamma_1 \chosen \zo{}$.
$\qB$ sends $\seed^\ast\concat i^\ast \concat \gamma_1$ to the challenger. The challenger returns $(x^\ast,\peek,\pedk_{\ne x^\ast})$ to $\qB$.

Then, $\qB$ passes $\pp\seteq\bot$ and $\xk \seteq \bot$ to $\qA$.
$\qB$ also computes $\oly \seteq \prf(x^\ast)$.
\begin{description}
\item[Challenge:] When $\qA$ sends a challenge query $\msg$, $\qB$ does the following
\begin{itemize}
 \item Construct $D_{\ne x^\ast}^{\$}[\prf,\pedk_{\ne x^\ast},\msg,\gamma_1,x^\ast,\oly]$ as described in~\cref{fig:description_D_punc_random}.
 \item Return $\tlC \seteq \iO(D_{\ne x^\ast}^{\$}[\prf,\pedk_{\ne x^\ast},\msg,\gamma_1,x^\ast,\oly])$ and $\iop \seteq \peek$ to $\qA$.
 \end{itemize} 
\end{description}

After finishing $\qA$'s challenge query, $\qB$ computes $y^\ast \seteq \PRG(\seed^\ast)$ and sends $(\gamma_1,x^\ast,y^\ast)$ to $\qA$.
Finally, when $\qA$ terminates with output $\coin'$, $\qB$ outputs $\coin'$ and terminates. $\qB$ perfectly simulates
\begin{itemize}
\item $\hybij{1}{=}$ if $x^\ast \gets \PE.\Enc(\peek, \seed^\ast\concat i^\ast \concat \gamma_1)$,
\item $\hybij{2}{=}$ if $x^\ast \chosen \zo{\inplen}$.
\end{itemize}
Thus, we see that the proposition holds.
\end{proof}

\begin{proposition}\label{prop:io_const_equal_third}
If $\PRG$ is a secure PRG, it holds that $\abs{\Pr[\hybij{2}{=}=1]- \Pr[\hybij{3}{=}=1]}\le \negl(\secp)$.
\end{proposition}
\begin{proof}[Proof of~\cref{prop:io_const_equal_third}]
The difference between the two games is that $y^\ast$ in the target triple $(\gamma_0,x_0,y^\ast)$ is $\PRG(\seed^\ast)$ or random in the case where $\gamma_0 = \msg[i]$. Recall that we rename $\gamma_1 \chosen \zo{}$ to $\gamma_0\chosen \zo{}$.
Note that we randomly choose $x_0 \chosen \zo{\inplen}$ and use $\prf$ and $\pedk_{\ne x_0}$ in these games. Thus, we can apply pseudorandomness of $\PRG$ since the value $\seed^\ast$ is never used anywhere else.
\end{proof}

\paragraph{The case where $\gamma^\ast \ne \msg[i]$.} Next, we prove propositions for the case where $\gamma^\ast \ne \msg[i]$.

\begin{proposition}\label{prop:io_const_ne_first}
If $\iO$ is a secure IO and $\PE$ satisfies punctured correctness, it holds that
\ifnum\submission=1
$\abs{\Pr[\hybi{0}=1]- \Pr[\hybij{1}{\ne}=1]}\le \negl(\secp)$.
\else
\[\abs{\Pr[\hybi{0}=1]- \Pr[\hybij{1}{\ne}=1]}\le \negl(\secp).\]
\fi
\end{proposition}
\begin{proof}[Proof of~\cref{prop:io_const_ne_first}]
The difference between the two games is that $D_{\ne x_1}^{\mathsf{real}}[\prf,\pedk_{\ne x_1},\msg,\gamma,x_1,y_1]$ is used for the challenge query instead of $D[\prf,\pedk,\msg]$ in the case where $\gamma_1 \ne \msg[i^\ast]$. These two circuits are the same except that
\begin{itemize}
 \item for input $x_1$, $D_{\ne x_1}^{\mathsf{real}}$ directly outputs the hard-wired value $y_1 =\PRG(\seed^\ast)$,
 \end{itemize}
 due to the punctured correctness of $\PE$.
 Thus, if the following hold, $D_{\ne x_1}^{\mathsf{real}}$ and $D$ are functionally equivalent:
 \begin{itemize}
  \item $D(x_1)$ outputs $y_1 = \PRG(\seed^\ast)$ when $\gamma_1 \ne \msg[i^\ast]$,
  \end{itemize}
  This holds since $x_1 \gets \PE.\Enc(\peek, \seed^\ast \concat i^\ast \concat \gamma_1)$, $D(x_1)$ runs the item (b) in~\cref{fig:description_D}, and $\gamma_1 \ne \msg[i^\ast]$ holds in this case.
Thus, $D_{\ne x_1}^{\mathsf{real}}$ and $D$ are functionally equivalent and the proposition holds due to IO security.
\end{proof}

\begin{proposition}\label{prop:io_const_ne_second}
If $\PE$ satisfies strong pseudorandom ciphertext, it holds that
\ifnum\submission=1
$\abs{\Pr[\hybij{1}{\ne}=1]- \Pr[\hybij{2}{\ne}=1]}\le \negl(\secp)$.
\else
\[\abs{\Pr[\hybij{1}{\ne}=1]- \Pr[\hybij{2}{\ne}=1]}\le \negl(\secp).\]
\fi
\end{proposition}
\begin{proof}[Proof of~\cref{prop:io_const_ne_second}]
We construct an algorithm $\qB$ for strong ciphertext pseudorandomness by using a distinguisher $\qA$.
$\qB$ generates $\prf \gets \PRF.\Gen(1^\secp)$ and chooses $\seed^\ast \chosen \zo{\ell}$ and $\gamma_1 \chosen \zo{}$.
$\qB$ sends $\seed^\ast\concat i^\ast \concat \gamma_1$ to the challenger. The challenger returns $(x^\ast,\peek,\pedk_{\ne x^\ast})$ to $\qB$.

Then, $\qB$ passes $\pp\seteq\bot$ and $\xk \seteq \bot$ to $\qA$.
$\qB$ also computes $y^\ast \seteq \PRG(\seed^\ast)$.
\begin{description}
\item[Challenge:] When $\qA$ sends a challenge query $\msg$, $\qB$ does the following
\begin{itemize}
 \item Construct $D_{\ne x^\ast}^{\mathsf{real}}[\prf,\pedk_{\ne x^\ast},\msg,\gamma_1,x^\ast,y^\ast]$ where $y^\ast \seteq \PRG(\seed^\ast)$ as described in~\cref{fig:description_D_punc_real}.
 \item Return $\tlC\seteq \iO(D_{\ne x^\ast}^{\mathsf{real}})$ and $\iop \seteq \peek$ to $\qA$.
 \end{itemize} 
\end{description}

After finishing $\qA$'s challenge, $\qB$ sends $(\gamma_1,x^\ast,y^\ast)$ to $\qA$.
Finally, when $\qA$ terminates with output $\coin'$, $\qB$ outputs $\coin'$ and terminates.

$\qB$ perfectly simulates
\begin{itemize}
\item $\hybij{1}{\ne}$ if $x^\ast \gets \PE.\Enc(\peek, \seed^\ast\concat i^\ast \concat \gamma_1)$,
\item $\hybij{2}{\ne}$ if $x^\ast \chosen \zo{\inplen}$.
\end{itemize}
Thus, we see that the proposition holds.
\end{proof}

\begin{proposition}\label{prop:io_const_ne_third}
If $\PRG$ is a secure PRG, it holds that $\abs{\Pr[\hybij{2}{\ne}=1]- \Pr[\hybij{3}{\ne}=1]}\le \negl(\secp)$.
\end{proposition}
\begin{proof}[Proof of~\cref{prop:io_const_ne_third}]
The difference between the two games is that $y^\ast$ in the target triple $(\gamma_0,x_0,y^\ast)$ is $\PRG(\seed^\ast)$ or random in the case where $\gamma_0 \ne \msg[i]$. Recall that we rename $\gamma_1 \chosen \zo{}$ to $\gamma_0\chosen \zo{}$.
Note that we randomly choose $x_0 \chosen \zo{\inplen}$ and use $\prf$ and $\pedk_{\ne x_0}$ in these games. Thus, we can apply pseudorandomness of $\PRG$ since the value $\seed^\ast$ is never used anywhere else.
\end{proof}

\begin{proposition}\label{prop:io_const_ne_fourth}
If $\iO$ is a secure IO and $\prf$ satisfies punctured correctness, it holds that
\ifnum\submission=1
$\abs{\Pr[\hybij{3}{\ne}=1]- \Pr[\hybij{4}{\ne}=1]}\le \negl(\secp)$.
\else
\[\abs{\Pr[\hybij{3}{\ne}=1]- \Pr[\hybij{4}{\ne}=1]}\le \negl(\secp).\]
\fi
\end{proposition}
\begin{proof}[Proof of~\cref{prop:io_const_ne_fourth}]
The difference between the two games is that $D_{\ne x_0}^{\mathsf{real}}[\prf_{\ne x_0},\pedk_{\ne x_0},\msg,\gamma_0,x_0,y_0]$ is used for the challenge query instead of $D_{\ne x_0}^{\mathsf{real}}[\prf,\pedk_{\ne x_0},\msg,\gamma_0,x_0,y_0]$ in the case where $\gamma_0 \ne \msg[i^\ast]$. These two circuits are the same except that we use $\prf_{\ne x_0}$ instead of $\prf$.
Those two circuits above are functionally equivalent since $\prf_{\ne x_0}(\cdot)$ is functionally equivalent to $\prf$ except for $x_0$ and both $D_{\ne x_0}^{\mathsf{real}}[\prf_{\ne x_0},\pedk_{\ne x_0},\msg,\gamma_0,x_0,y_0](x_0)$ and $D_{\ne x_0}^{\mathsf{real}}[\prf,\pedk_{\ne x_0},\msg,\gamma_0,x_0,y_0](x_0)$ directly outputs $y_0$ by the description of $D_{\ne x_0}^{\mathsf{real}}$. Note that $D_{\ne x_0}^{\mathsf{real}}$ does not have any ``if branch'' condition that uses $\prf$ or $\prf_{\ne x_0}$.

Thus, the proposition holds due to IO security.
\end{proof}

\begin{proposition}\label{prop:io_const_ne_fifth}
If $\prf$ satisfies punctured pseudorandomness, it holds that
\[\abs{\Pr[\hybij{4}{\ne}=1]- \Pr[\hybij{5}{\ne}=1]}\le \negl(\secp).\]
\end{proposition}
\begin{proof}[Proof of~\cref{prop:io_const_ne_fifth}]
We construct an algorithm $\qB$ that breaks the pseudorandomness at punctured points of $\prf$ by using $\qA$.

$\qB$ generates $(\peek,\pedk)\lrun \PE.\Gen(1^\secp)$, chooses $x_0 \chosen \zo{\inplen}$ and $\gamma_0\chosen \zo{}$, sends $x_0$ as the challenge to its challenger of $\prf$, and receives $\prf_{\ne x_0}$ and $y^\ast$.
Here $x_0$ does not rely on $\msg$, so we can generate $x_0$ before $\msg$ is fixed.
$\qB$ sends $\pp\seteq \bot$ and $\xk \seteq \bot$ to $\qA$. $\qB$ also computes $\pedk_{\ne x_0} \lrun \PE.\Puncture(\pedk,x_0)$.
\ifnum\submission=0
\begin{description}
\item [Challenge:] For query $\msg$, $\qB$ can simulate the target marked circuit $\tlC=\iO(D_{\ne x_0}^{\mathsf{real}}[\prf_{\ne x_0},\pedk_{\ne x_0},\msg,\gamma_0,x_0,y^\ast])$ by using $\pedk_{\ne x_0}$, $\prf_{\ne x_0}$, $y^\ast$, and the public tag $\iop=\peek$.
\end{description}
\else
\begin{description}
\item [Challenge:] For query $\msg$, $\qB$ can simulate the target marked circuit, that is, $\tlC=\iO(D_{\ne x_0}^{\mathsf{real}}[\prf_{\ne x_0},\pedk_{\ne x_0},\msg,\gamma_0,x_0,y^\ast])$ by using $\pedk_{\ne x_0}$, $\prf_{\ne x_0}$, $y^\ast$, and the public tag $\iop=\peek$.
\end{description}
\fi
After finishing $\qA$'s challenge query, $\qB$ sends $(\gamma_0,x_0,y^\ast)$ to $\qA$.
Finally, when $\qA$ terminates with output $\coin'$, $\qB$ outputs $\coin'$ and terminates.

$\cB$ perfectly simulates
\begin{itemize}
\item $\hybij{4}{\ne}$ if $y^\ast \chosen \zo{\outlen}$,
\item $\hybij{5}{\ne}$ if $y^\ast \seteq \prf(x_0)$.
\end{itemize}
The punctured pseudorandomness of $\prf$ immediately implies this proposition.
\end{proof}

\begin{proposition}\label{prop:io_const_ne_sixth}
If $\iO$ is a secure IO and $\prf$ satisfies punctured correctness, it holds that
\ifnum\submission=1
$\abs{\Pr[\hybij{5}{\ne}=1]- \Pr[\hybij{6}{\ne}=1]}\le \negl(\secp)$.
\else
\[\abs{\Pr[\hybij{5}{\ne}=1]- \Pr[\hybij{6}{\ne}=1]}\le \negl(\secp).\]
\fi
\end{proposition}
\begin{proof}[Proof of~\cref{prop:io_const_ne_sixth}]
The difference between the two games is that $D_{\ne x_0}^{\mathsf{real}}[\prf,\pedk_{\ne x_0},\msg,\gamma_0,x_0,y_0]$ is used for the challegne query instead of $D_{\ne x_0}^{\mathsf{real}}[\prf_{\ne x_0},\pedk_{\ne x_0},\msg,\gamma_0,x_0,y_0]$ in the case where $\gamma_0 \ne \msg[i^\ast]$. These two circuits are the same except that we use $\prf$ instead of $\prf_{\ne x_0}$.
This proof is the same as that of~\cref{prop:io_const_ne_fifth} (in a reverse manner). Thus, we omit it.
\end{proof}

\paragraph{End of case analyses.} We complete the two case analyses.

\begin{proposition}\label{prop:io_const_seventh}
If $\iO$ is a secure IO, $\prf$ satisfies punctured correctness, and $\PE$ satisfies punctured correctness, it holds that $\abs{\Pr[\hybij{3}{=}=1]- \Pr[\hybij{4}{=}=1]}\le \negl(\secp)$ and $\abs{\Pr[\hybij{6}{\ne}=1]- \Pr[\hybij{7}{\ne}=1]}\le \negl(\secp)$.
\end{proposition}
\begin{proof}
This proof is the same as that of~\cref{prop:io_const_equal_first} and~\cref{prop:io_const_ne_first}, respectively (in a reverse manner). Thus, we omit them.
\end{proof}

We complete the proof of~\cref{thm:io_const_extraction-less_watermarkable_PRF}.
\ifnum\submission=0
\end{proof}
\else\fi

%% file: section/altogether.tex

\section{Putting Pieces Altogether}\label{sec:altogether}

\paragraph{Privately extractable watermarking PRF.}
We summarize how to obtain our privately extractable watermarking PRF.

By~\cref{thm:extraction-less_watermarking_pcprf,thm:pcprf_extended_weak_pseudorandomness,thm:pcprf_lwe,thm:CCA_QLWE,thm:pseudorandom_ske}, we obtain an extraction-less watermarking with private simulation from the QLWE assumption. By combining this with~\cref{thm:qelwmprf-from-elwmprf,thm:watermarking_from_extraction-less}, we obtain the following theorem.

\begin{theorem}
If the QLWE assumption holds, there exists a privately extractable watermarking PRF.
\end{theorem}

\paragraph{Publicly extractable watermarking PRF.}
We summarize how to obtain our publicly extractable watermarking PRF.

By~\cref{thm:io_const_extraction-less_watermarkable_PRF,thm:CHNVW_PE_strong_ct_pseudorandom,thm:pprf-owf,thm:owf_prg}, we obtain an extraction-less watermarking with public simulation from IO and the QLWE assumption since OWFs can be instantiated with the QLWE assumption.
By combining this with~\cref{thm:watermarking_from_extraction-less_public}, we obtain a publicly extractable watermarking PRF from IO and the QLWE assumption.
Thus, we obtain the following theorem.

\begin{theorem}
If there exists a secure IO and the QLWE assumption holds, there exists a publicly extractable watermarking PRF.
\end{theorem}

%% file: section/qelwmprf-from-elwmprf.tex

\newcommand{\QEL}{\mathsf{QEL}}
\newcommand{\qek}{\mathsf{qek}}
\newcommand{\qxk}{\mathsf{qxk}}

\section{Achieving QSIM-MDD from SIM-MDD}\label{sec:qelwmprf-from-elwmprf}
We prove \cref{thm:qelwmprf-from-elwmprf}, that is, we show that we can transform extraction-less watermarking PRF satisfying SIM-MDD security with private simulation into one satisfying QSIM-MDD security with private simulation, by using a QPRF.
Before the proof, we introduce semi-classical one-way to hiding (O2H) lemma. 

\subsection{Semi-Classical One-Way to Hiding (O2H) Lemma}\label{sec:qro}

We recall a few lemmas.

\begin{definition}[Punctured oracle]
Let $F:X\ra Y$ be any function, and $S\subset X$ be a set.
The oracle $F\setminus S$ (``$F$ punctured by $S$'') takes as input a value $x\in X$.
It first computes whether $x\in S$ into an auxiliary register and measures it.
Then it computes $F(x)$ and returns the result.
Let $\Find$ be the event that any of the measurements returns $1$.
\end{definition}

\begin{lemma}[{Semi-classical O2H~\cite[Theorem 1]{C:AmbHamUnr19}}]\label{lem:sco2h}
Let $G,H:X\ra Y$ be random functions, $z$ be a random value, and $S\subseteq X$ be a random set such that $G(x)=H(x)$ for every $x\notin S$.
The tuple $(G,H,S,z)$ may have arbitrary joint distribution.
Furthermore, let $\qA$ be a quantum oracle algorithm.
Let $\Ev$ be any classical event.
Then we have
\begin{align}
\abs{\Pr[\Ev: \qA^{\ket{H}}(z)]-\Pr[\Ev: \qA^{\ket{G}}(z)]} \leq 2\sqrt{(q+1)\cdot\Pr[\Find: \qA^{\ket{H\setminus S}}(z)]}
\enspace.
\end{align}
\end{lemma}

\begin{lemma}[{Search in semi-classical oracle~\cite[Theorem 2]{C:AmbHamUnr19}}]\label{lem:search-in-sc}
Let $H:X\ra Y$ be a random function, let $z$ be a random value, and let $S\subset X$ be a random set.
$(H,S,z)$ may have arbitrary joint distribution.
Let $\qA$ be a quantum oracle algorithm.
If for each $x\in X$, $\Pr[x\in S]\leq \epsilon$ (conditioned on $H$ and $z$), then we have
\begin{align}
\Pr[\Find:\qA^{\ket{H\setminus S}}(z)]\leq 4q\epsilon
\enspace,
\end{align}
where $q$ is the number of queries to $H$ by $\qA$.
\end{lemma}

Note that the above lemma is originally introduced in \cite{C:AmbHamUnr19}, but we use a variant that is closer to Lemma 4 in \cite{TCC:BHHHP19}.

\subsection{Proof}

\paragraph{Construction.}
We start with the construction.
Let $\ELWMPRF=(\Setup,\Gen,\Eval,\Mark,\Sim)$ be an extraction-less watermarking PRF scheme satisfying SIM-MDD security with private simulation.
We also let the message space of $\ELWMPRF$ is $\bit^{\msglen}$.
Let $\PRF$ be a QPRF with domain $\zo{\secp}$ and range $\cR_{\Sim}$, which is the randomness space of $\Sim$.
We construct an extraction-less watermarking PRF scheme $\QELWMPRF=(\QEL.\Setup,\QEL.\Gen,\QEL.\Eval,\QEL.\Mark,\allowbreak\QEL.\Sim)$ satisfying QSIM-MDD security with private simulation as follows.
We use $\Gen$, $\Eval$, and $\Mark$ as $\QEL.\Gen$, $\QEL.\Eval$, and $\QEL.\Mark$, respectively. The domain and range of $\QELWMPRF$ are the same as those of $\ELWMPRF$.
The mark space of $\QELWMPRF$ is $\bit^{\msglen}$.
Also, we construct $\QEL.\Setup$ and $\QEL.\Sim$ as follows.

\begin{description}
\item[$\QEL.\Setup(1^\secp)$:] $ $
\begin{itemize}
\item Generate $(\pp,\xk)\gets\Setup(1^\secp)$.
\item Genrate $K\gets\bit^\secp$.
\item Outputs $(\pp,\qxk:=(\xk,K))$.
\end{itemize}

\item[$\QEL.\Sim(\qxk,\iop, i;r)$:] $ $
\begin{itemize}
\item Parse $(\xk,K)\gets\qxk$.
\item Output $(\gamma,x,y) \gets \Sim(\xk,\iop,i;\PRF_K(r))$.
\end{itemize}
\end{description}

%
%
%
%
%
\paragraph{Security analysis.}
Let $i^\ast\in[\msglen]$ and $\qA$ be any QPT adversary for QSIM-MDD security with private simulation making total $q$ queries to $\Oracle{sim}$ and $\Oracle{api}$.
We prove that for any polynomial $w$, it holds that $\adva{i^\ast,\qA,\QELWMPRF}{q\textrm{-}sim\textrm{-}mdd}(\secp)\leq 1/w$.
We prove it using hybrid games.
Let $\SUC_X$ be the event that the final output is $1$ in Game $X$.
We define a distribution $\D_{\iop',i'}$ as
\begin{description}
\item[$D_{\iop',i'}$:] Output $(\gamma,x,y)\gets\Sim(\xk,\iop',i')$.
\end{description}

\begin{description}
\item[Game $1$:]This is $\expt{i^\ast,\qA,\QELWMPRF}{q\textrm{-}sim\textrm{-}mdd}(\secp)$. Thus, $\adva{i^\ast,\qA,\QELWMPRF}{q\textrm{-}sim\textrm{-}mdd}(\secp)= 2\abs{\Pr[\SUC_1]-1/2}$.
\begin{enumerate}
\item The challenger generates $(\pp,\xk) \chosen \Setup(1^\secp)$ and $K\gets\bit^\secp$, and gives $\pp$ to $\qA$.
$\qA$ send $\msg\in\bit^{\msglen}$ to the challenger.
The challenger generates $(\iop,\prfk)\gets\Gen(\pp)$, computes $\tlC\gets\Mark(\pp,\prfk,\msg)$, and sends $\tlC$ to $\qA$.

\item $\qA$ can access to the following oracles.
\begin{description}
\item[$\Oracle{sim}$:]On input $\iop'$ and $i'$, it returns $\Sim(\xk,\iop',i';\PRF_K(r))$, where $r\gets\bit^\secp$.
\item[$\Oracle{api}$:]On input $(\epsilon,\delta,\iop',i')$ and a quantum state $\qstateq$, it returns the result of $\API^{\epsilon,\delta}_{\cP,\D^{\PRF}_{\iop',i'}}(\qstateq)$ and the post measurement state, where $\D^{\PRF}_{\iop',i'}=\D_{\iop',i'}(\PRF_K(\cdot))$.

\end{description}

\item The challenger generates $\coin\gets\bit$.
If $\coin=0$, the challenger samples $(\gamma,x,y)\gets\Dreal{i^\ast}$.
If $\coin=1$, the challenger generates $(\gamma,x,y)\gets\Sim(\xk,\iop,i^\ast;\PRF_K(r^*))$, 
where, $r^*\gets\bit^\secp$.
The challenger sends $(\gamma,x,y)$ to $\qA$.

\item When $\qA$ terminates with output $\coin'$, the challenger outputs $1$ if $\coin=\coin'$ and $0$ otherwise.

\end{enumerate}
\end{description}

\begin{description}
\item[Game $2$:] This game is the same as Game $1$ except that $\PRF_K$ is replaced with a quantum-accessible random function $\Rand$.
\end{description}
We have $\abs{\Pr[\SUC_1]-\Pr[\SUC_2]}=\negl(\secp)$ from the the security of $\PRF$.

\begin{description}
\item[Game $3$:] This game is the same as Game $2$ except that $\Rand$ is replaced with
\begin{align}
V(r)=
\begin{cases}
v^* & (\textrm{if~~} r=r^*)\\
\Rand(r) & (\textrm{otherwise}),
\end{cases}
\end{align}
where $v^* \chosen \cR_{\Sim}$.
\end{description}
We have $\abs{\Pr[\SUC_2]-\Pr[\SUC_3]}=0$.

\begin{description}
\item[Game $4$:] This game is the same as Game $3$ except the followings.
When $\qA$ makes a query $\iop'$ and $i'$ to $\Oracle{sim}$, $\Sim(\xk,\iop',i';R(r))$ is returned instead of $\Sim(\xk,\iop',i';V(r))$.
Also, when $\qA$ makes a query $(\epsilon,\delta,\iop',i')$ to $\Oracle{api}$, $\API^{\epsilon,\delta}_{\cP,D^{\Rand}_{\iop',i'}}(\qstateq)$ is performed instead of $\API^{\epsilon,\delta}_{\cP,D^V_{\iop',i'}}(\qstateq)$, where $D^{\Rand}_{\iop',i'}=\D_{\iop',i'}(\Rand(\cdot))$ and $D^V_{\iop',i'}=\D_{\iop',i'}(V(\cdot))$.

By this change, $V$ is now used only for generating the challenge tuple $(\gamma,x,y)\gets\Sim(\xk,\iop,i^\ast;V(r^*))=\Sim(\xk,\iop,i^\ast;v^\ast)$.
\end{description}
We have $\abs{\Pr[\SUC_3]-\Pr[\SUC_4]}=O(\sqrt{\frac{q^2}{2^\secp}})$ from \cref{lem:sco2h} and \cref{lem:search-in-sc}.

\begin{description}
\item[Game $5$:] This game is the same as Game $4$ except that $\Rand$ is replaced with $G\circ F$, where $F:\bit^\secp\ra[s]$ and $G:[s]\ra\cR_\Sim$ are random functions and $s$ is a polynomial of $\secp$ specified later.
\end{description}
\begin{theorem}[Small Range Distribution~\cite{FOCS:Zhandry12}]\label{thm:SRD}
For any QPT adversary $\qB$ making $q$ quantum queries to $\Rand$ or $G\circ F$, we have $\abs{\Pr[\qB^{\ket{\Rand}}(1^\secp)=1]-\Pr[\qB^{\ket{G\circ F}}(1^\secp)=1]}\leq O(q^3/s)$.
\end{theorem}
By the above theorem, we have $\abs{\Pr[\SUC_4]-\Pr[\SUC_5]}=O(q^3/s)$.

We can simulate $F$ using a $2q$-wise independent function $E$ by the following theorem.
\begin{theorem}[\cite{C:Zhandry12}]\label{thm:q_random_q_wise}
For any QPT adversary $\qB$ making $q$ quantum queries to $F$ or $E$, we have $\Pr[\qB^{\ket{F}}(1^\secp)=1]=\Pr[\qB^{\ket{E}}(1^\secp)=1]$.
\end{theorem}
We can efficiently simulate $\API^{\epsilon,\delta}_{\cP,\D_{\iop',i'}^{G\circ E}}$ in Game $5$ using $s$ samples from $\D_{\iop',i'}$ since $\D_{\iop',i'}(G(\cdot))$ can be interpreted as a mapping for $s$ samples from $\D_{\iop',i'}$.
Then, from the SIM-MDD security with private simulation of $\ELWMPRF$, we have $\abs{\Pr[\SUC_5]-1/2}=\negl(\secp)$.
From the above, we also have $\advc{i^\ast,\qA,\QELWMPRF}{q}{sim}{mdd}(\secp)\leq O(q^3/s)+2\gamma$ for some negligible function $\gamma$.
Thus, by setting $s=O(q^3\cdot w^2)$, we obtain $\advc{i^\ast,\qA,\QELWMPRF}{q}{sim}{mdd}(\secp)\leq 1/w$.

Since $w$ is any polynomial, this means that $\advc{i^\ast,\qA,\QELWMPRF}{q}{sim}{mdd}(\secp)=\negl(\secp)$.

\begin{remark}
It is easy to see that the extended weak pseudorandomness of $\ELWMPRF$ is preserved after we apply the transformation above since the evaluation algorithm is the same as that of $\ELWMPRF$ and extended weak pseudorandomness holds against adversaries that generate $\pp$. Thus, we omit a formal proof.
\end{remark}

%% file: section/pe_construction.tex
\section{Puncturable Encryption with Strong Ciphertext Pseudorandomness}\label{sec:pe_construction}

We prove~\cref{thm:CHNVW_PE_strong_ct_pseudorandom} in this section.

\subsection{Tools for PE}\label{sec:pe_tools}

\begin{definition}[Statistically Injective PPRF]\label{def:injective_PPRF}
If a PPRF family $\cF = \{\prf_{K}: \zo{\ell_1(\secp)} \ra \zo{\ell_2(\secp)} \mid K \in\zo{\secp}\}$ satisfies the following, we call it a statistically injective PPRF family with failure probability $\epsilon(\cdot)$. With probability $1-\epsilon(\secp)$ over the random choice of $K \gets \prfgen(1^\secp)$, for all $x,x^\prime \in \zo{\ell_1(\secp)}$, if $x\ne x^\prime$, then $\prf_K(x)\ne \prf_K(x^\prime)$.
If $\epsilon(\cdot)$ is not specified, it is a negligile function.
\end{definition}
Sahai and Waters show that we can convert any PPRF into a statistically injective PPRF~\cite{SIAMCOMP:SahWat21}.

\begin{theorem}[\cite{SIAMCOMP:SahWat21}]\label{thm:statistical_injective_PPRF}
If OWFs exist, then for all efficiently computable functions $n(\secp)$, $m(\secp)$, and $e(\secp)$ such that $m(\secp) \ge 2n(\secp) + e(\secp)$, there exists a statistically injective PPRF family with failure probability $2^{-e(\secp)}$ that maps $n(\secp)$ bits to $m(\secp)$ bits.
\end{theorem}

\begin{definition}\label{def:injective_com}
An injective bit-commitment with setup consists of PPT algorithms $(\Gen,\Commit)$.
\begin{description}
 \item[$\Gen(1^\secp)$:] The key generation algorithm takes as input the security parameter $1^\secp$ and outputs a commitment key $\ck$.
 \item[$\Commit_\ck(b)$:] The commitment algorithm takes as input $\ck$ and a bit $b$ and outputs a commitment $\com$.
 \end{description}
 These satisfy the following properties.
\begin{description}
  \item[Computationally Hiding:] For any QPT $\qA$, it holds that
  \[
  \abs{\Pr[\qA(\Commit_\ck(0))=1 \mid \ck \gets \Gen(1^\secp)] - \Pr[\qA(\Commit_\ck(1))=1 \mid \ck \gets \Gen(1^\secp)]} \le \negl(\secp).
  \]
  \item[Statistically Binding:] It holds that
  \begin{align}
  \Pr \left [\com_0 = \com_1 \  \middle | \begin{array}{l}
  \ck \gets \Gen(1^\secp)\\
  \com_0 \gets \Commit_\ck( 0) \\
  \com_1 \gets \Commit_\ck( 1)
  \end{array}
  \right ] \le \negl(\secp).
  \end{align}
  \item[Injective:] For every security parameter $\secp$, there is a bound $\ell_r$ on the number of random bits used by $\Commit$ such that if $\ck \gets \Gen(1^\secp)$, $\Commit_\ck(\cdot\ ; \cdot)$ is an injective function on $\zo{} \times \zo{\ell_{r}}$ except negligible probability.
\end{description}
\end{definition}

\begin{theorem}\label{thm:injective_com_lwe}
If the QLWE assumption holds, there exists a secure injective bit-commitment with setup.
\end{theorem}
This theorem follows from the following theorems.

\begin{theorem}[\cite{JC:Naor91}]\label{thm:naor_bit_com}
If there exists (injective) OWFs, there exists (injective) bit-commitment.
\end{theorem}
\begin{theorem}[{\cite[Adapted]{SIAMCOMP:PeiWat11,C:AKPW13}}]\label{thm:lossy_lwe}
If the QLWE assumption holds, there exists a secure injective OWF with evaluation key generation algorithms.
\end{theorem}
\begin{remark}
The injective OWFs achieved in~\cref{thm:lossy_lwe} needs evaluation key generation algorithms unlike the standard definition of OWFs. However, OWFs with evaluation key generation algorithms are sufficient for proving~\cref{thm:injective_com_lwe} by using~\cref{thm:naor_bit_com} since we use \emph{commitment key generation algorithm $\Gen$} (i.e., setup) in~\cref{def:injective_com}. Note that there is no post-quantum secure injective OWF \emph{without evaluation key generation algorithm} so far.
\end{remark}

\subsection{PE Scheme Description}\label{sec:pe_scheme_description}

We review the puncturable encryption scheme by Cohen et al.~\cite{SIAMCOMP:CHNVW18}.
We can see \cref{thm:CHNVW_PE_strong_ct_pseudorandom} holds by inspecting their PE scheme.
The scheme utilizes the following ingredients and the length $n$ of ciphertexts is $12$ times the length $\ell$ of plaintexts:
\begin{itemize}
\setlength{\parskip}{0.4mm}
\setlength{\itemsep}{0.4mm}
\item A length-doubling $\PRG:\zo{\ell} \ra \zo{2\ell}$
\item An injective PPRFs (See~\cref{def:injective_PPRF}) $F : \zo{3\ell} \ra \zo{9\ell}$.
\item A PPRF $G : \zo{9\ell} \ra \zo{\ell}$.
\item An injective bit-commitment with setup $(\Commit.\Gen,\Commit)$ using randomness in $\zo{9 \ell}$.  We only use this in our security proof.
\end{itemize}

\paragraph{Scheme.}
The scheme $\PE$ by Cohen et al.~\cite{SIAMCOMP:CHNVW18} is as follows.

\begin{description}
\item[$\Gen(1^\secp$):] Sample functions $F$ and $G$, generates $\peek$ as the obfuscated circuit $\iO(E)$ where $E$ is described in~\cref{fig:pe_enc_E}, and returns $(\peek, \pedk) \seteq (\iO(E),D)$, where $\pedk$ is the (un-obfuscated) program $D$ in \cref{fig:pe_dec_D}.

\item[$\Puncture(\pedk, c^\ast)$:] Output $\pedk_{\ne c^\ast}$, where $\pedk_{\ne c^\ast}$ is the obfuscated circuits $\iO(D_{\ne c^\ast})$ where $D_{\ne c^\ast}$ is  described in \cref{fig:pe_dec_D_punc}, that is, $\pedk_{\ne c^\ast} \seteq \iO(D_{\ne c^\ast})$.

\item[$\Enc(\peek, m)$:] Take $m \in \zo{\ell}$, sample $s\gets \zo{\ell}$, and outputs $c \gets \peek(m,s)$.

\item[$\Dec(\pedk, c)$:] Take $c \in \zo{12 \ell}$ and returns $m \seteq \pedk(c)$.
\end{description}
The size of the circuits is appropriately padded to be the maximum size of all modified circuits, which will appear in the security proof.


\protocol
{Circuit $E[F,G]$}
{Description of encryption circuit $E$}
{fig:pe_enc_E}
{
\begin{description}
\item[Constants]: Injective PPRF $F: \zo{3 \ell} \rightarrow \zo{9 \ell}$, PPRF $G : \zo{9 \ell} \rightarrow \zo{\ell}$
\item[Inputs:] $m \in \zo{\ell}, s \in \zo{\ell}$
\end{description}
\begin{enumerate}
\item Compute $\alpha = \PRG(s)$.
\item Compute $\beta = F(\alpha \| m)$.
\item Compute $\gamma = G(\beta) \xor m$.
\item Output $(\alpha, \beta, \gamma)$.
\end{enumerate}
}

\protocol
{Circuit $D[F,G]$}
{Description of decryption circuit $D$}
{fig:pe_dec_D}
{
\begin{description}
\item[Constants]: Injective PPRF $F: \zo{3 \ell} \rightarrow \zo{9 \ell}$, PPRF $G : \zo{9 \ell} \rightarrow \zo\ell$
\item[Inputs:] $c = (\alpha\|\beta\|\gamma)$, where $\alpha \in \zo{2 \ell}$, $\beta \in \zo{9 \ell}$, and $\gamma \in \zo\ell$.
\end{description}
\begin{enumerate}
\item Compute $m = G(\beta) \xor \gamma$.
\item If $\beta = F(\alpha \| m)$, output $m$.
\item Else output $\bot$.
\end{enumerate}
}

\protocol
{Circuit $D_{\ne c^\ast}[F,G,c^\ast]$}
{Description of punctured decryption circuit $D_{\ne c^\ast}$ at $c^\ast$}
{fig:pe_dec_D_punc}
{
\begin{description}
\item[Constants]: Point $c^\ast \in \zo{12\ell}$, injective PPRF $F: \zo{3 \ell} \rightarrow \zo{9 \ell}$, and PPRF $G : \zo{9 \ell} \rightarrow \zo\ell$
\item[Inputs:] $c = (\alpha\|\beta\|\gamma)$, where $\alpha \in \zo{2 \ell}$, $\beta \in \zo{9 \ell}$, and $\gamma \in \zo\ell$.
\end{description}
\begin{enumerate}
\item If $c = c^\ast$, output $\bot$.
\item Compute $m = G(\beta) \xor \gamma$.
\item If $\beta = F(\alpha \| m)$, output $m$.
\item Else output $\bot$.
\end{enumerate}
}

\subsection{PE Security Proof}\label{sec:pe_security_proof}

Cohen et al.~\cite{SIAMCOMP:CHNVW18} proved correctness, punctured correctness, and sparseness of $\PE$ above by using secure PRG $\PRG$, secure injective PPRF $F$, secure PPRF $G$, and secure IO $\iO$. Thus, we complete the proof of~\cref{thm:CHNVW_PE_strong_ct_pseudorandom} by combining~\cref{thm:statistical_injective_PPRF,thm:injective_com_lwe}, and~\cref{thm:post-quantum_PE_strong} below, which we prove in this section.

\begin{theorem}\label{thm:post-quantum_PE_strong}
If $\PRG$ is a secure PRG, $F$ is a secure injective PPRF, $G$ is a secure PPRF, $\Commit$ is a secure injective bit-commitment with setup, and $\iO$ is a secure IO, then $\PE$ is a secure PE that satisifes strong ciphertext pseudorandomness.
\end{theorem}
\begin{proof}[Proof of~\cref{thm:post-quantum_PE_strong}]
To prove $x_0 \seteq c^\ast \gets \Enc(\peek,m^\ast)$ is indistinguishable from $x_1 \seteq r^\ast \chosen \zo{\ell}$, we define a sequence hybrid games.
\begin{description}
\item[$\sfreal$:] This is the same as the real game with $b=0$. That is, for queried $m^\ast$ the challenger does the following.
\begin{enumerate}
\item Choose an injective PPRF $F: \zo{3\ell} \ra \zo{9\ell}$ and PPRF $G: \zo{9\ell} \ra \zo{\ell}$.
\item Choose $s\chosen \zo{\ell}$ and compute $\alpha_0 \seteq \PRG(s)$, $\beta_0 \seteq F(\alpha_0\concat m^\ast)$, and $\gamma_0 \seteq G(\beta_0) \xor m^\ast$.
\item Set $x_0 \seteq \alpha_0\concat \beta_0\concat \gamma_0$ and computes $\peek \seteq \iO(E)$ and $\pedk_{\ne x_0} \seteq \iO(D_{\ne x_0})$.
\item Send $(x_0,\peek,\pedk_{\ne x_0})$ to the adversary.
\end{enumerate}
\item[$\hybi{1}$:] This is the same as $\hybi{0}(0)$ except that $\alpha_0$ is uniformly random.
\item[$\hybi{2}$:] This is the same as $\hybi{1}$ except that we use punctured $F_{\ne \alpha_0\concat m^\ast}$ and modified circuits $E_{\ne \alpha_0\concat m^\ast}$ and $D_{\ne \alpha_0\concat m^\ast}^2$ described in~\cref{fig:pe_enc_E2_punc,fig:pe_dec_D2_punc}. Intuitively, these modified circuits are punctured at input $\alpha_0\concat m^\ast$ and use exceptional handling for this input.
\item[$\hybi{3}$:] This is the same as $\hybi{2}$ except that $\beta_0 \chosen \zo{9\ell}$.
\item[$\hybi{4}$:] This is the same as $\hybi{3}$ except that we use punctured $G_{\ne \beta_0}$ and modified circuits $E_{\ne \alpha_0\concat m^\ast,\ne \beta_0}$ and $D_{\ne \alpha_0 \concat m^\ast,\ne\beta_0}^4$ described in~\cref{fig:pe_enc_E4_punc,fig:pe_dec_D4_punc}.
Intuitively, these modified circuits are punctured at input $\beta_0$ and use $F_{\ne \alpha_0\concat m^\ast}$ and exceptional handling for $\beta_0$.
\item[$\hybi{5}=\sfrand_2$:] This is the same as $\hybi{4}$ except that $\gamma_0$ is uniformly random. Now, $\alpha_0$, $\beta_0$, $\gamma_0$ are uniformly random and we rewrite them into $\alpha_1$, $\beta_1$, $\gamma_1$, respectively. For ease of notation, we also denote this game by $\sfrand_2$.
\item[$\sfrand_{1}$:] This is the same as $\hybi{5}=\sfrand_2$ except that we use un-punctured $G$, circuit $E_{\ne \alpha_0\concat m^\ast,\ne\beta_0}$ reverts to $E_{\ne \alpha_1\concat m^\ast}$ described in~\cref{fig:pe_enc_E2_punc}, and we change circuit $D_{\ne \alpha_1 \concat m^\ast,\ne \beta_1}^4$ into $D_{\ne \alpha_1 \concat m^\ast}^{\mathsf{r}}$ described in~\cref{fig:pe_dec_Dr_punc}.
\item[$\sfrand$:] This is the same as the real game with $b=1$. That is, for queried $m^\ast$ the challenger does the following.
\begin{enumerate}
\item Choose an injective PPRF $F: \zo{3\ell} \ra \zo{9\ell}$ and PPRF $G: \zo{9\ell} \ra \zo{\ell}$.
\item Choose $\alpha_1 \chosen \zo{2\ell}$, $\beta_1\chosen \zo{9\ell}$, and $\gamma_1\chosen \zo{\ell}$.
\item Set $x_1 \seteq \alpha_1\concat \beta_1\concat \gamma_1$ and computes $\peek \seteq \iO(E)$ and $\pedk_{\ne x_1} \seteq \iO(D_{\ne x_1})$.
\item Send $(x_1,\peek,\pedk_{\ne x_1})$ to the adversary.
\end{enumerate}
\end{description}
We described the overview of these hybrid games in~\cref{fig:pe_strong_ct_pseudorandom}.
\begin{figure}[ht]
    \centering
    \begin{tabular}{l lllll}\toprule
          & $\alpha^\ast$ & $\beta^\ast$ & $\gamma^\ast$ & $\peek\seteq \iO(\cdot)$  & $\pedk\seteq \iO(\cdot)$ \\\midrule
$\sfreal$ & $\PRG(s)$ &  $F(\alpha_0\concat m^\ast)$ & $G(\beta_0)\xor m^\ast$ & $E$ & $D_{\ne x_0}$ \\
$\hybi{1}$ & \redline{$\$$} &  $F(\alpha_0\concat m^\ast)$ & $G(\beta_0)\xor m^\ast$ & $E$ & $D_{\ne x_0}$ \\
$\hybi{2}$ & $\$$ &  $F(\alpha_0\concat m^\ast)$ & $G(\beta_0)\xor m^\ast$ & \redline{$E_{\ne \alpha_0\concat m^\ast}$} & \redline{$D_{\ne \alpha_0\concat m^\ast}^2[F_{\ne \alpha_0\concat m^\ast}]$} \\
$\hybi{3}$ & $\$$ &  \redline{$\$$} & $G(\beta_0)\xor m^\ast$ & $E_{\ne \alpha_0\concat m^\ast}$ & $D_{\ne \alpha_0\concat m^\ast}^2[F_{\ne \alpha_0\concat m^\ast}]$ \\
$\hybi{4}$ & $\$$ &  $\$$ & $G(\beta_0)\xor m^\ast$ & \redline{$E_{\ne \alpha_0\concat m^\ast, \ne \beta_0}$} & \redline{$D_{\ne \alpha_0\concat m^\ast, \ne \beta_0}^4[F_{\ne \alpha_0\concat m^\ast},G_{\ne \beta_0}]$} \\
$\hybi{5}$ & $\$$ &  $\$$ & \redline{$\$$} & $E_{\ne \alpha_1\concat m^\ast, \ne \beta_1}$ & $D_{\ne \alpha_1\concat m^\ast, \ne \beta_1}^4[F_{\ne \alpha_1\concat m^\ast},G_{\ne \beta_1}]$ \\
$\sfrand_1$ & $\$$ &  $\$$ & $\$$ &  \redline{$E_{\ne \alpha_1\concat m^\ast}$} & \redline{$D_{\ne \alpha_1\concat m^\ast}^{\mathsf{r}}[F_{\ne \alpha_1\concat m^\ast}]$} \\
$\sfrand$ & $\$$ &  $\$$ & $\$$ & \redline{$E$} & \redline{$D_{\ne x_1}$} \\\bottomrule
    \end{tabular}
\caption{High-level overview of hybrid games from $\sfreal$ to $\sfrand$. Recall that $\hybi{5}=\sfrand_2$.
Transitions from $\sfrand_2$ to $\sfrand$ are baiscally the reverse transitions from $\hybi{0}$ to $\hybi{4}$, but there are subtle differences.}
\label{fig:pe_strong_ct_pseudorandom}
\end{figure}
If we prove these hybrid games are indistinguishable, we complete the proof of~\cref{thm:post-quantum_PE_strong}.
\end{proof}
We prove that those hybrid games in~\cref{fig:pe_strong_ct_pseudorandom} are indistinguishable by~\cref{lem:pe_zero_one,lem:pe_one_two,lem:pe_two_three,lem:pe_three_four,lem:pe_four_five,lem:pe_rand_zero_one,lem:pe_rand_one_two}.

\paragraph{From $\sfreal$ to $\hybi{5}$.} We first move from $\sfreal$ to $\hybi{5}$.
\begin{lemma}\label{lem:pe_zero_one}
If $\PRG$ is a secure PRG, it holds that $\abs{\Pr[\hybi{0}(0)=1] - \Pr[\hybi{1}=1]} \le \negl(\secp)$.
\end{lemma}
\begin{proof}[Proof of~\cref{lem:pe_zero_one}]
The randomness $s$ for encryption is never used anywhere except $\alpha_0 \seteq \PRG(s)$. 
We can apply the PRG security and immediately obtain the lemma.
\end{proof}

\begin{lemma}\label{lem:pe_one_two}
If $\iO$ is a secure IO and $F$ is a secure injective PPRF, it holds that
\ifnum\submission=1
$\abs{\Pr[\hybi{1}=1] - \Pr[\hybi{2}=1]} \le \negl(\secp)$.
\else
\[\abs{\Pr[\hybi{1}=1] - \Pr[\hybi{2}=1]} \le \negl(\secp).\]
\fi
\end{lemma}
\begin{proof}[Proof of~\cref{lem:pe_one_two}]
We change $E$ and $D_{\ne x_0}$ into $E_{\ne \alpha_0\concat m^\ast}$ and $D_{\ne \alpha_0\concat m^\ast}^2$, respectively.
\protocol
{Circuit $E_{\ne \alpha^\ast\concat m^\ast}[F^\prime,G]$}
{Description of encryption circuit $E_{\ne \alpha^\ast\concat m^\ast}$}
{fig:pe_enc_E2_punc}
{
\begin{description}
\item[Constants]: Injective PPRF \redline{$F^\prime$}, PPRF $G$
\item[Inputs:] $m \in \zo{\ell}, s \in \zo{\ell}$
\end{description}
\begin{enumerate}
\item Compute $\alpha = \PRG(s)$.
\item Compute $\beta = F^\prime(\alpha \| m)$.
\item Compute $\gamma = G(\beta) \xor m$.
\item Output $(\alpha, \beta, \gamma)$.
\end{enumerate}
}

\protocol
{Circuit $D_{\ne \alpha^\ast\concat m^\ast}^2[F^\prime,G,\alpha^\ast,\beta^\ast,\gamma^\ast,m^\ast]$}
{Description of punctured decryption circuit $D_{\ne \alpha^\ast\concat m^\ast}^2$}
{fig:pe_dec_D2_punc}
{
\begin{description}
\item[Constants]: Point $x^\ast = \alpha^\ast\concat \beta^\ast \concat \gamma^\ast \in \zo{12\ell}$, injective PPRF \redline{$F^\prime$}, PPRF $G$, \redline{$m^\ast$}.
\item[Inputs:] $c = (\alpha\|\beta\|\gamma)$, where $\alpha \in \zo{2 \ell}$, $\beta \in \zo{9 \ell}$, and $\gamma \in \zo\ell$.
\end{description}
\begin{enumerate}
\item If $c = x^\ast$, output $\bot$.
\item Compute $m = G(\beta) \xor \gamma$.
\item \redline{If $(\alpha,m) = (\alpha^\ast,m^\ast)$, output $\bot$.}
\item If $\beta = \redline{F^\prime}(\alpha \| m)$, output $m$.
\item Else output $\bot$.
\end{enumerate}
}
We define a sequence of sub-hybrid games.
\begin{description}
\item[$\hybij{1}{1}$:] This is the same as $\hybi{1}$ except that we generate $F_{\ne \alpha_0\concat m^\ast}$ and set $F^\prime \seteq F_{\ne \alpha_0\concat m^\ast}$ and $\peek \seteq \iO(E_{\ne \alpha_0\concat m^\ast})$ described in~\cref{fig:pe_enc_E2_punc}.
\item[$\hybij{1}{2}$:] This is the same as $\hybij{1}{1}$ except that we set $\pedk_{\ne x_0} \seteq \iO(D_{\ne \alpha_0\concat m^\ast}^2[F,G,\alpha_0,\beta_0,\gamma_0,m^\ast])$ described in~\cref{fig:pe_dec_D2_punc}. That is, we still use $F$, but modify the circuit.
\end{description}

\begin{proposition}\label{prop:pe_sub_one_oneone}
If $\iO$ is a secure IO, it holds that $\abs{\Pr[\hybi{1}=1] - \Pr[\hybij{1}{1}=1]} \le \negl(\secp)$.
\end{proposition}
\begin{proof}[Proof of~\cref{prop:pe_sub_one_oneone}]
In these games, value $\alpha_0 \chosen \zo{2\ell}$ is not in the image of $\PRG$ except with negligible probability. The only difference between the two games is that $F_{\ne \alpha_0\concat m^\ast}$ is used in $\hybij{1}{1}$. Thus, $E$ and $E_{\ne \alpha_0\concat m^\ast}$ are functionally equivalent except with negligible probability. We can obtain the proposition by applying the IO security.
\end{proof}

\ifnum\submission=0
\begin{proposition}\label{prop:pe_sub_oneone_onetwo}
If $\iO$ is a secure IO and $F$ is injective, it holds that $\abs{\Pr[\hybij{1}{1}=1] - \Pr[\hybij{1}{2}=1]} \le \negl(\secp)$.
\end{proposition}
\else
\begin{proposition}\label{prop:pe_sub_oneone_onetwo}
$\abs{\Pr[\hybij{1}{1}=1] - \Pr[\hybij{1}{2}=1]} \le \negl(\secp)$ holds if $\iO$ is a secure IO and $F$ is injective.
\end{proposition}
\fi
\begin{proof}[Proof of~\cref{prop:pe_sub_oneone_onetwo}]
We analyze the case where $(\alpha,m) = (\alpha_0,m^\ast)$ since it is the only difference between $D_{\ne x_0}$ and $D_{\ne \alpha_0\concat m^\ast}^2$.
\begin{itemize}
\item If $c = x_0$, $D_{\ne \alpha_0\concat m^\ast}^2$ outputs $\bot$ by the first line of the description. Thus, the output of $D_{\ne \alpha_0\concat m^\ast}^2(x_0)$ is the same as that of $D_{\ne x_0}(x_0)$.
\item If $c \ne x_0$, it holds $(\beta_0,\gamma_0)\ne (\beta,\gamma)$ in this case. However, it should be $\beta_0 = \beta$ due to the injectivity of $F$ and $\beta_0 =F(\alpha_0\concat m^\ast)$. Thus, both $D_{\ne x_0}(c)$ and $D_{\ne \alpha_0\concat m^\ast}^2(c)$ output $\bot$ in this case ($D_{\ne x_0}(c)$ outputs $\bot$ at the first line).
\end{itemize}
 Therefore, $D_{\ne x_0}$ and $D_{\ne \alpha_0\concat m^\ast}^2$ are functionally equivalent. We can obtain the proposition by applying the IO security.
\end{proof}

\begin{proposition}\label{prop:pe_sub_onetwo_two}
If $\iO$ is a secure IO, it holds that $\abs{\Pr[\hybij{1}{2}=1] - \Pr[\hybi{2}=1]} \le \negl(\secp)$.
\end{proposition}
\begin{proof}[Proof of~\cref{prop:pe_sub_onetwo_two}]
Due to the exceptional handling in the third item of $D_{\ne \alpha_0\concat m^\ast}^2$, $F(\alpha\concat m)$ is never computed for input $(\alpha_0,m^\ast)$. Thus, even if we use $F_{\ne \alpha_0\concat m^\ast}$ instead of $F$, $D_{\ne \alpha_0\concat m^\ast}^2[F]$ and $D_{\ne \alpha_0\concat m^\ast}^2[F_{\ne \alpha_0\concat m^\ast}]$ are functionally equivalent. We can obtain the proposition by the IO security.
\end{proof}
We complete the proof of~\cref{lem:pe_one_two}.
\end{proof}

\begin{lemma}\label{lem:pe_two_three}
If $F$ is a secure injective PPRF, it holds that $\abs{\Pr[\hybi{2}=1] - \Pr[\hybi{3}=1]} \le \negl(\secp)$.
\end{lemma}
\begin{proof}[Proof of~\cref{lem:pe_two_three}]
The difference between these two games is that $\beta_0$ is $F(\alpha_0\concat m^\ast)$ or random. We can immediately obtain the lemma by applying punctured pseudorandomness of $F$ since we use $F_{\ne \alpha_0 \concat m^\ast}$ in these games.
\end{proof}

\begin{lemma}\label{lem:pe_three_four}
If $\iO$ is a secure IO and $F$ is a secure injective PPRF, it holds that
\ifnum\submission=1
$\abs{\Pr[\hybi{3}=1] - \Pr[\hybi{4}=1]} \le \negl(\secp)$.
\else
\[\abs{\Pr[\hybi{3}=1] - \Pr[\hybi{4}=1]} \le \negl(\secp).\]
\fi
\end{lemma}
\begin{proof}[Proof of~\cref{lem:pe_three_four}]
We change $E_{\ne \alpha^\ast\concat m^\ast}$ and $D_{\ne \alpha^\ast\concat m^\ast}^2$ into $E_{\ne \alpha^\ast\concat m^\ast,\ne\beta^\ast}$ and $D_{\ne \alpha^\ast\concat m^\ast,\ne\beta^\ast}^4$, respectively.
\protocol
{Circuit $E_{\ne \alpha^\ast\concat m^\ast,\ne \beta^\ast}[F^\prime,G^\prime]$}
{Description of encryption circuit $E_{\ne \alpha^\ast\concat m^\ast,\ne \beta^\ast}$}
{fig:pe_enc_E4_punc}
{
\begin{description}
\item[Constants]: Injective PPRF $F^\prime$, PPRF \redline{$G^\prime$}
\item[Inputs:] $m \in \zo{\ell}, s \in \zo{\ell}$
\end{description}
\begin{enumerate}
\item Compute $\alpha = \PRG(s)$.
\item Compute $\beta = F^\prime(\alpha \| m)$.
\item Compute $\gamma = \redline{G^\prime}(\beta) \xor m$.
\item Output $(\alpha, \beta, \gamma)$.
\end{enumerate}
}

\protocol
{Circuit $D_{\ne \alpha^\ast\concat m^\ast,\ne \beta^\ast}^4[F^\prime,G^\prime,\alpha^\ast,\beta^\ast,\gamma^\ast,m^\ast]$}
{Description of punctured decryption circuit $D_{\ne \alpha^\ast \concat m^\ast,\ne \beta^\ast}^4$}
{fig:pe_dec_D4_punc}
{
\begin{description}
\item[Constants]: Point $x^\ast \seteq \alpha^\ast \concat \beta^\ast \concat \gamma^\ast \in \zo{12\ell}$, injective PPRF $F^\prime$, PPRF \redline{$G^\prime$}, $m^\ast$.
\item[Inputs:] $c = (\alpha\|\beta\|\gamma)$, where $\alpha \in \zo{2 \ell}$, $\beta \in \zo{9 \ell}$, and $\gamma \in \zo\ell$.
\end{description}
\begin{enumerate}
\item If \redline{$\beta = \beta^\ast$}, output $\bot$.
\item Compute $m = \redline{G^\prime}(\beta) \xor \gamma$.
\item If $(\alpha,m) = (\alpha^\ast,m^\ast)$, output $\bot$.
\item If $\beta = F^\prime(\alpha \| m)$, output $m$.
\item Else output $\bot$.
\end{enumerate}
}

We define a sequence of sub-hybrid games.
\begin{description}
\item[$\hybij{3}{1}$:] This is the same as $\hybi{3}$ except that we use punctured $G_{\ne \beta_0}$ and set $\peek \seteq \iO(E_{\ne \alpha_0\concat m^\ast,\ne \beta_0}[F_{\ne \alpha_0\concat m^\ast},G_{\ne \beta_0}])$.
\item[$\hybij{3}{2}$:] This is the same as $\hybij{3}{1}$ except that we still use $G$ but set $\pedk_{\ne c_0} \seteq \iO(D_{\ne \alpha_0\concat m^\ast,\ne \beta_0}^4[F_{\ne \alpha_0\concat m^\ast},G])$.
\end{description}

\begin{proposition}\label{prop:pe_sub_three_threeone}
If $\iO$ is a secure IO, it holds that $\abs{\Pr[\hybi{3}=1] - \Pr[\hybij{3}{1}=1]} \le \negl(\secp)$.
\end{proposition}
\begin{proof}[Proof of~\cref{prop:pe_sub_three_threeone}]
In these games $\beta_0 \chosen \zo{9\ell}$ is uniformly random. By the sparsity of $F$, $\beta_0$ is not in the image of $F$ except with negligible probability. Thus, $E_{\ne \alpha_0\concat m^\ast}$ and $E_{\ne \alpha_0\concat m^\ast, \ne \beta_0}$ are functionally equivalent except with negligible probability.
We obtain the proposition by the IO security.
\end{proof}
\begin{proposition}\label{prop:pe_sub_threeone_threetwo}
If $\iO$ is a secure IO, it holds that $\abs{\Pr[\hybij{3}{1}=1] - \Pr[\hybij{3}{2}=1]} \le \negl(\secp)$.
\end{proposition}
\begin{proof}[Proof of~\cref{prop:pe_sub_threeone_threetwo}]
The difference between $D_{\ne \alpha_0\concat m^\ast}^2$ and $D_{\ne \alpha_0\concat m^\ast,\ne \beta_0}^4$ is that we replace ``If $c=x_0$, outputs $\bot$.'' with ``If $\beta= \beta_0$, outputs $\bot$.''. In these games, $\beta_0 \chosen \zo{9\ell}$ is not in the image of $F$ except with negligible probability. Recall that $c=x_0$ means $c = \alpha_0\concat \beta_0\concat \gamma_0$. Thus, those two circuits may differ when $\beta= \beta_0$ but $(\alpha,\gamma) \ne (\alpha_0,\gamma_0)$. However, it does not happen $\beta = F^\prime(\alpha \concat (G(\beta)\xor \gamma))$ in this case due to the injectivity of $F$. Thus, $D_{\ne \alpha_0\concat m^\ast}^2$ and $D_{\ne \alpha_0\concat m^\ast,\ne \beta_0}^4$ are functionally equivalent and we obtain the proposition by applying the IO security.
\end{proof}
\begin{proposition}\label{prop:pe_sub_threetwo_four}
If $\iO$ is a secure IO, it holds that $\abs{\Pr[\hybij{3}{2}=1] - \Pr[\hybi{4}=1]} \le \negl(\secp)$.
\end{proposition}
\begin{proof}[Proof of~\cref{prop:pe_sub_threetwo_four}]
The difference between these two games that we use $D_{\ne \alpha_0\concat m^\ast,\ne \beta_0}^4[F_{\ne \alpha_0\concat m^\ast},G_{\ne \beta_0}]$ instead of $D_{\ne \alpha_0\concat m^\ast,\ne \beta_0}^4[F_{\ne \alpha_0\concat m^\ast},G]$. However, $G_{\ne \beta_0}(\beta_0)$ is never computed by the first item of $D_{\ne \alpha_0\concat m^\ast,\ne \beta_0}^4$. We obtain the proposition by the IO security.
\end{proof}
We complete the proof of~\cref{lem:pe_three_four}.
\end{proof}

\begin{lemma}\label{lem:pe_four_five}
If $G$ is a secure PPRF, it holds that $\abs{\Pr[\hybi{4}=1] - \Pr[\hybi{5}=1]} \le \negl(\secp)$.
\end{lemma}

\begin{proof}[Proof of~\cref{lem:pe_four_five}]
The difference between these two games is that $\gamma_0$ is $G(\beta_0)$ or random. We can immediately obtain the lemma by applying punctured pseudorandomness of $G$ since we use $G_{\ne \beta_0}$ in these games.
\end{proof}

In $\hybi{5}$, $\alpha_0$, $\beta_0$, and $\gamma_0$ are uniformly random strings as $\alpha_1$, $\beta_1$, and $\gamma_1$.

\paragraph{From $\sfrand$ to $\hybi{5}$.}
We leap to $\sfrand$ and move from $\sfrand$ to $\sfrand_2 = \hybi{5}$ instead of directly moving from $\hybi{5}=\sfrand_2$ to $\sfrand$ since $\sfreal \approx \hybi{5}$ and $\sfrand_{2} \approx \sfrand$ is almost symmetric (but not perfectly symmetric).

\begin{lemma}\label{lem:pe_rand_zero_one}
If $\iO$ is a secure IO and $F$ is a secure injective PPRF, it holds that
\ifnum\submission=1
$\abs{\Pr[\sfrand=1] - \Pr[\sfrand_1 =1]} \le \negl(\secp)$.
\else
\[\abs{\Pr[\sfrand=1] - \Pr[\sfrand_1 =1]} \le \negl(\secp).\]
\fi
\end{lemma}
\begin{proof}[Proof of~\cref{lem:pe_rand_zero_one}]
We change $E$ and $D_{\ne x_1}$ into $E_{\ne \alpha_1\concat m^\ast}$ and $D_{\ne \alpha_1\concat m^\ast}^{\mathsf{r}}$, respectively.

We define a sequence of sub-hybrid games.
\begin{description}
\item[$\rhybij{}{1}$:] This is the same as $\sfrand$ except that we generate $F_{\ne \alpha_1\concat m^\ast}$ and set $F^\prime \seteq F_{\ne \alpha_1\concat m^\ast}$ and $\peek \seteq \iO(E_{\ne \alpha_1\concat m^\ast})$ described in~\cref{fig:pe_enc_E2_punc}.
\item[$\rhybij{}{2}$:] This is the same as $\rhybij{}{1}$ except that we set $\pedk_{\ne x_1} \seteq \iO(D_{\ne \alpha_1\concat m^\ast}^{\mathsf{r}\textrm{-}2}[F,G])$ described in~\cref{fig:pe_dec_Dr-2_punc}. That is, we still use $F$, but the modified circuit that outputs $m^\ast$ for input $\alpha_1\concat \hat{\beta}\concat \hat{\gamma}$, where $\hat{\beta} \seteq F(\alpha_1\concat m^\ast)$ and $\hat{\gamma} \seteq G(\hat{\beta})\xor m^\ast$.
\item[$\rhybij{}{3}$:] This is the same as $\rhybij{}{2}$ except that we set $\pedk_{\ne x_1} \seteq \iO(D_{\ne \alpha_1\concat m^\ast}^{\mathsf{r}}[F,G])$ described in~\cref{fig:pe_dec_Dr_punc}. That is, we still use $F$, but the modified circuit outputs $\bot$ for an input such that $(\alpha,m)=(\alpha_1, m^\ast)$.
\end{description}

\protocol
{Circuit $D_{\ne \alpha^\ast\concat m^\ast}^{\mathsf{r}\textrm{-}2}[F^\prime,G,\alpha^\ast,\beta^\ast,\gamma^\ast,\hat{\beta},\hat{\gamma},m^\ast]$}
{Description of punctured decryption circuit $D_{\ne \alpha^\ast\concat m^\ast}^{\mathsf{r}\textrm{-}2}$}
{fig:pe_dec_Dr-2_punc}
{
\begin{description}
\item[Constants]: Point $x^\ast = \alpha^\ast\concat \beta^\ast \concat \gamma^\ast \in \zo{12\ell}$, injective PPRF \redline{$F^\prime$}, PPRF $G$, \redline{$\hat{\beta}$, $\hat{\gamma}$, $m^\ast$}.
\item[Inputs:] $c = (\alpha\|\beta\|\gamma)$, where $\alpha \in \zo{2 \ell}$, $\beta \in \zo{9 \ell}$, and $\gamma \in \zo\ell$.
\end{description}
\begin{enumerate}
  \item \redline{If $\alpha = \alpha^\ast$ and $\beta = \hat{\beta}$ and $\gamma = \hat{\gamma}$, output $m^\ast$.}
\item If $c = x^\ast$, output $\bot$.
\item Compute $m = G(\beta) \xor \gamma$.
\item If $\beta = \redline{F^\prime}(\alpha \| m)$, output $m$.
\item Else output $\bot$.
\end{enumerate}
}

\protocol
{Circuit $D_{\ne \alpha^\ast\concat m^\ast}^{\mathsf{r}}[F^\prime,G,\alpha^\ast,\beta^\ast,\gamma^\ast,\hat{\beta},\hat{\gamma},m^\ast]$}
{Description of punctured decryption circuit $D_{\ne \alpha^\ast\concat m^\ast}^{\mathsf{r}}$}
{fig:pe_dec_Dr_punc}
{
\begin{description}
\item[Constants]: Point $x^\ast = \alpha^\ast\concat \beta^\ast \concat \gamma^\ast \in \zo{12\ell}$, injective PPRF \redline{$F^\prime$}, PPRF $G$, \redline{$\hat{\beta}$, $\hat{\gamma}$, $m^\ast$}.
\item[Inputs:] $c = (\alpha\|\beta\|\gamma)$, where $\alpha \in \zo{2 \ell}$, $\beta \in \zo{9 \ell}$, and $\gamma \in \zo\ell$.
\end{description}
\begin{enumerate}
  \item \redline{If $\alpha = \alpha^\ast$ and $\beta = \hat{\beta}$ and $\gamma = \hat{\gamma}$, output $m^\ast$.}
\item If $c = x^\ast$, output $\bot$.
\item Compute $m = G(\beta) \xor \gamma$.
\item \redline{If $(\alpha,m) = (\alpha^\ast,m^\ast)$, output $\bot$.}
\item If $\beta = \redline{F^\prime}(\alpha \| m)$, output $m$.
\item Else output $\bot$.
\end{enumerate}
}

\begin{proposition}\label{prop:pe_sub_rand_zero_one}
If $\iO$ is a secure IO, it holds that $\abs{\Pr[\sfrand=1] - \Pr[\rhybij{}{1}=1]} \le \negl(\secp)$.
\end{proposition}
\begin{proof}[Proof of~\cref{prop:pe_sub_rand_zero_one}]
In these games, value $\alpha_1 \chosen \zo{2\ell}$ is not in the image of $\PRG$ except with negligible probability. Thus, $E$ and $E_{\ne \alpha_1\concat m^\ast}$ are functionally equivalent except with negligible probability. We can obtain the proposition by applying the IO security.
\end{proof}

\ifnum\submission=0
\begin{proposition}\label{prop:pe_sub_rand_one_two}
If $\iO$ is a secure IO and $F$ is injective, it holds that $\abs{\Pr[\rhybij{}{1}=1] - \Pr[\rhybij{}{2}=1]} \le \negl(\secp)$.
\end{proposition}
\else
\begin{proposition}\label{prop:pe_sub_rand_one_two}
$\abs{\Pr[\rhybij{}{1}=1] - \Pr[\rhybij{}{2}=1]} \le \negl(\secp)$ holds if $\iO$ is a secure IO and $F$ is injective.
\end{proposition}
\fi
\begin{proof}[Proof of~\cref{prop:pe_sub_rand_one_two}]
The difference between $D_{\ne x_1}$ and $D_{\ne \alpha_1\concat m^\ast}^{\mathsf{r}\textrm{-}2}$ is ``If $\alpha = \alpha^\ast$ and $\beta = \hat{\beta}$ and $\gamma = \hat{\gamma}$, output $m^\ast$.''. Although $\alpha_1\concat \hat{\beta}\concat \hat{\gamma}$ is a valid encryption, $\hat{\beta} = F(\alpha_1\concat m^\ast)$ is not equal to $\beta_1$ except with negligible probability since $\beta_1$ is uniformly random. Similarly, $\hat{\gamma}$ is not equal to $\gamma_1$ except with negligible probability. Thus, $D_{\ne x_1}(\alpha_1\concat \hat{\beta}\concat \hat{\gamma})$ outputs $m^\ast$. That is, $D_{\ne x_1}$ and $D_{\ne \alpha_1\concat m^\ast}^{\mathsf{r}\textrm{-}2}$ are functionally equivalent. We can obtain the proposition by applying the IO security.
\end{proof}
\ifnum\submission=0
\begin{proposition}\label{prop:pe_sub_rand_two_three}
If $\iO$ is a secure IO and $F$ is injective, it holds that $\abs{\Pr[\rhybij{}{2}=1] - \Pr[\rhybij{}{3}=1]} \le \negl(\secp)$.
\end{proposition}
\else
\begin{proposition}\label{prop:pe_sub_rand_two_three}
$\abs{\Pr[\rhybij{}{2}=1] - \Pr[\rhybij{}{3}=1]} \le \negl(\secp)$ holds if $\iO$ is a secure IO and $F$ is injective.
\end{proposition}
\fi
\begin{proof}[Proof of~\cref{prop:pe_sub_rand_two_three}]
We analyze the case where $(\alpha,m) = (\alpha_1,m^\ast)$. We can reach the forth line of $D_{\ne \alpha_1\concat m^\ast}^{\mathsf{r}}$ if $c \ne x_1$. If $c\ne x_1$ and $(\alpha,m) = (\alpha_1,m^\ast)$, it holds that $(\beta,\gamma)\ne (\beta_1,\gamma_1)$. However, it should be $\beta_1 = \beta$ in this case due to the injectivity of $F$. That is, if $D_{\ne \alpha_1\concat m^\ast}^{\mathsf{r}}(c)$ outputs $\bot$ at the fourth line, $D_{\ne \alpha_1\concat m^\ast}^{\mathsf{r}\textrm{-}2}(c)$ also outputs $\bot$ at the second line. Therefore, $D_{\ne \alpha_1\concat m^\ast}^{\mathsf{r}\textrm{-}2}$ and $D_{\ne \alpha_1\concat m^\ast}^{\mathsf{r}}$ are functionally equivalent. We can obtain the proposition by applying the IO security.
\end{proof}

\begin{proposition}\label{prop:pe_sub_rand_three_one}
If $\iO$ is a secure IO, it holds that $\abs{\Pr[\rhybij{}{3}=1] - \Pr[\sfrand_1 =1]} \le \negl(\secp)$.
\end{proposition}
\begin{proof}[Proof of~\cref{prop:pe_sub_rand_three_one}]
Due to the exceptional handling in the fourth line of $D_{\ne \alpha_1\concat m^\ast}^{\mathsf{r}}$, $F(\alpha\concat m)$ is never computed for input $(\alpha_1,m^\ast)$. Thus, even if we use $F_{\ne \alpha_1\concat m^\ast}$ instead of $F$, $D_{\ne \alpha_1\concat m^\ast}^{\mathsf{r}}[F]$ and $D_{\ne \alpha_1\concat m^\ast}^{\mathsf{r}}[F_{\ne \alpha_1\concat m^\ast}]$ are functionally equivalent. We can obtain the proposition by the IO security.
\end{proof}
We complete the proof of~\cref{lem:pe_rand_zero_one}.
\end{proof}

\begin{lemma}\label{lem:pe_rand_one_two}
If $\iO$ is a secure IO, $F$ is a secure injective PPRF, and $(\Commit.\Gen,\Commit)$ is a secure injective bit-commitment with setup, it holds that $\abs{\Pr[\sfrand_{1} =1] - \Pr[\sfrand_2 =1]} \le \negl(\secp)$.
\end{lemma}
\begin{proof}[Proof of~\cref{lem:pe_rand_one_two}]
We change $E_{\ne \alpha^\ast\concat m^\ast}$ and $D_{\ne \alpha^\ast\concat m^\ast}^{\mathsf{r}}$ into $E_{\ne \alpha^\ast\concat m^\ast,\ne\beta^\ast}$ and $D_{\ne \alpha^\ast\concat m^\ast,\ne\beta^\ast}^4$, respectively.

\protocol
{Circuit $D_{\ne \alpha^\ast\concat m^\ast}^{\com}[F^\prime,G,\alpha^\ast,\beta^\ast,\gamma^\ast,\hat{z},\ck,\hat{\gamma},m^\ast]$}
{Description of punctured decryption circuit $D_{\ne \alpha^\ast \concat m^\ast}^{\com}$}
{fig:pe_dec_Dcom_punc}
{
\begin{description}
\item[Constants]: Point $x^\ast =\alpha^\ast\concat \beta^\ast\concat \gamma^\ast \in \zo{12\ell}$, injective PPRF $F^\prime$, PPRF $G$, $m^\ast$, \redline{$\hat{z}$, $\ck$}, $\hat{\gamma}$.
\item[Inputs:] $c = (\alpha\|\beta\|\gamma)$, where $\alpha \in \zo{2 \ell}$, $\beta \in \zo{9 \ell}$, and $\gamma \in \zo{\ell}$.
\end{description}
\begin{enumerate}
\item If $\alpha=\alpha^\ast$ and \redline{$\Commit_{\ck}(0;\beta)=\hat{z}$} and $\gamma=\hat{\gamma}$, output $m^\ast$.
\item If $c = x^\ast$, output $\bot$.
\item Compute $m = G(\beta) \xor \gamma$.
\item If $(\alpha,m) = (\alpha^\ast,m^\ast)$, output $\bot$.
\item If $\beta = F^\prime(\alpha \| m)$, output $m$.
\item Else output $\bot$.
\end{enumerate}
}

\protocol
{Circuit $D_{\ne \alpha^\ast\concat m^\ast}^{\sfF}[F^\prime,G,\alpha^\ast,\beta^\ast,\gamma^\ast,m^\ast,\hat{z},\hat{\gamma}]$}
{Description of punctured decryption circuit $D_{\ne \alpha^\ast \concat m^\ast}^{\sfF}$}
{fig:pe_dec_Df_punc}
{
\begin{description}
\item[Constants]: Point $x^\ast =\alpha^\ast\concat \beta^\ast\concat \gamma^\ast \in \zo{12\ell}$, injective PPRF $F^\prime$, PPRF $G$, $m^\ast$, $\hat{z}$, $\hat{\gamma}$.
\item[Inputs:] $c = (\alpha\|\beta\|\gamma)$, where $\alpha \in \zo{2 \ell}$, $\beta \in \zo{9 \ell}$, and $\gamma \in \zo{\ell}$.
\end{description}
\begin{enumerate}
\item If $\alpha=\alpha^\ast$ and \redline{$\textsc{False}$} and $\gamma=\hat{\gamma}$, output $m^\ast$.~~~~~~// Never triggered
\item If $c = x^\ast$, output $\bot$.
\item Compute $m = G(\beta) \xor \gamma$.
\item If $(\alpha,m) = (\alpha^\ast,m^\ast)$, output $\bot$.
\item If $\beta = F^\prime(\alpha \| m)$, output $m$.
\item Else output $\bot$.
\end{enumerate}
}

We define a sequence of sub-hybrid games.
\begin{description}
 \item[$\rhybij{1}{1}$:] This is the same as $\sfrand_1$ except that we use $\hat{\beta}\chosen \zo{9\ell}$ instead of $F(\alpha_1\concat m^\ast)$.
 \item[$\rhybij{1}{2}$:] This is the same as $\rhybij{1}{1}$ except that we use $D_{\ne \alpha_1\concat m^\ast}^{\com}$ described in~\cref{fig:pe_dec_Dcom_punc}, where $\ck \gets \Commit.\Gen(1^\secp)$ and $\hat{z} = \Commit_{\ck}(0;\hat{\beta})$ are hardwired, instead of $D_{\ne \alpha_1\concat m^\ast}^{\mathsf{r}}$.
 \item[$\rhybij{1}{3}$:] This is the same as $\rhybij{1}{2}$ except that we hard-code $\hat{z}= \Commit_{\ck}(1;\hat{\beta})$ into $D_{\ne \alpha_1\concat m^\ast}^{\com}$ instead of $\Commit_{\ck}(0;\hat{\beta})$.
  \item[$\rhybij{1}{4}$:] This is the same as $\rhybij{1}{3}$ except that we use $D_{\ne \alpha_1\concat m^\ast}^{\sfF}$ described in~\cref{fig:pe_dec_Df_punc}
\item[$\rhybij{1}{5}$:] This is the same as $\rhybij{1}{4}$ except that we use punctured $G_{\ne \beta_1}$ and set $\peek \seteq \iO(E_{\ne \alpha_1\concat m^\ast,\ne \beta_1}[F_{\ne \alpha_1\concat m^\ast},G_{\ne \beta_1}])$.
\item[$\rhybij{1}{6}$:] This is the same as $\rhybij{1}{5}$ except that we still use $G$ but set $\pedk_{\ne c_1} \seteq \iO(D_{\ne \alpha_1\concat m^\ast,\ne \beta_1}^4[F_{\ne \alpha_1\concat m^\ast},G])$.
\end{description}

\begin{proposition}\label{prop:pe_sub_rand_one_oneone}
If $F$ is a secure PPRF, it holds that $\abs{\Pr[\sfrand_{1}=1] - \Pr[\rhybij{1}{1}=1]} \le \negl(\secp)$.
\end{proposition}
\begin{proof}[Proof of~\cref{prop:pe_sub_rand_one_oneone}]
In these games, we use $F_{\ne \alpha_1\concat m^\ast}$ in $E_{\ne \alpha^1 \concat m^\ast}$ and $D_{\ne \alpha_1\concat m^\ast}^{\mathsf{r}}$. Thus, we can apply the punctured pseudorandomness and immediately obtain the proposition.
\end{proof}

\begin{proposition}\label{prop:pe_sub_rand_oneone_onetwo}
If $\iO$ is a secure IO and $\Commit_\ck$ is injective, it holds that
\ifnum\submission=1
$\abs{\Pr[\rhybij{1}{1}=1] - \Pr[\rhybij{1}{2}=1]} \le \negl(\secp)$.
\else
\[\abs{\Pr[\rhybij{1}{1}=1] - \Pr[\rhybij{1}{2}=1]} \le \negl(\secp).\]
\fi
\end{proposition}
\begin{proof}[Proof of~\cref{prop:pe_sub_rand_oneone_onetwo}]
The difference between $D_{\ne \alpha_1\concat m^\ast}^{\com}$ and $D_{\ne \alpha_1\concat m^\ast}^{\mathsf{r}}$ is whether we use ``$\Commit_{\ck}(0;\beta)= \hat{z}$'' or ``$\beta=\hat{\beta}$'', where $\hat{z} = \Commit_{\ck} (0;\hat{\beta})$ and $\ck \gets \Commit.\Gen(1^\secp)$. Since $\Commit$ is injective, these two conditions are equivalent. Therefore, those two circuits are functionally equivalent. We obtain the proposition by applying the IO security.
\end{proof}

\begin{proposition}\label{prop:pe_sub_rand_onetwo_onethree}
If $(\Commit.\Gen,\Commit)$ is computationally hiding, it holds that
\ifnum\submission=1
$\abs{\Pr[\rhybij{1}{2}=1] - \Pr[\rhybij{1}{3}=1]} \le \negl(\secp)$.
\else
\[\abs{\Pr[\rhybij{1}{2}=1] - \Pr[\rhybij{1}{3}=1]} \le \negl(\secp).\]
\fi
\end{proposition}
\begin{proof}[Proof of~\cref{prop:pe_sub_rand_onetwo_onethree}]
The only difference between these two games is that $\hat{z}=\Commit_{\ck}(0;\hat{\beta})$ or $\hat{z}=\Commit_{\ck}(1;\hat{\beta})$. Note that $\hat{\beta}$ is never used anywhere else. We can obtain the proposition by the hiding property of $\Commit$.
\end{proof}

\begin{proposition}\label{prop:pe_sub_rand_onethree_onefour}
If $\iO$ is a secure IO and $(\Commit.\Gen,\Commit)$ is statistically binding, it holds that $\abs{\Pr[\rhybij{1}{3}=1] - \Pr[\rhybij{1}{4}=1]} \le \negl(\secp)$.
\end{proposition}
\begin{proof}[Proof of~\cref{prop:pe_sub_rand_onethree_onefour}]
The difference between $D_{\ne \alpha_1\concat m^\ast}^{\sfF}$ and $D_{\ne \alpha_1\concat m^\ast}^{\com}$ is that the first line of $D_{\ne \alpha_1\concat m^\ast}^{\sfF}$ is never executed. However, $\hat{z} = \Commit_{\ck}(1;\hat{\beta})$ is hardwired in $D_{\ne \alpha_1\concat m^\ast}^{\com}$. Thus, the first line of $D_{\ne \alpha_1\concat m^\ast}^{\com}$, in particular, condition ``$\Commit_{\ck}(0;\beta) = \hat{z} = \Commit_{\ck}(1;\hat{\beta})$'' is also never true except negligible probability due to the statistical binding property of $\Commit$. That is, these two circuits are functionally equivalent except negligible probability. We obtain the proposition by applying the IO security.
\end{proof}

\begin{proposition}\label{prop:pe_sub_rand_onefour_onefive}
If $\iO$ is a secure IO, it holds that $\abs{\Pr[\rhybij{1}{4}=1] - \Pr[\rhybij{1}{5}=1]} \le \negl(\secp)$.
\end{proposition}
\begin{proof}[Proof of~\cref{prop:pe_sub_rand_onefour_onefive}]
In these games $\beta_1 \chosen \zo{9\ell}$ is uniformly random. By the sparsity of $F$, $\beta_1$ is not in the image of $F$ except with negligible probability. Thus, $E_{\ne \alpha_1\concat m^\ast}$ and $E_{\ne \alpha_1\concat m^\ast, \ne \beta_1}$ are functionally equivalent except with negligible probability.
We obtain the proposition by the IO security.
\end{proof}

\begin{proposition}\label{prop:pe_sub_rand_onefive_onesix}
If $\iO$ is a secure IO, it holds that $\abs{\Pr[\rhybij{1}{5}=1] - \Pr[\rhybij{1}{6}=1]} \le \negl(\secp)$.
\end{proposition}
\begin{proof}[Proof of~\cref{prop:pe_sub_rand_onefive_onesix}]
The difference between $D_{\ne \alpha_1\concat m^\ast}^{\sfF}$ in~\cref{fig:pe_dec_Df_punc} and $D_{\ne \alpha_1\concat m^\ast,\ne \beta_1}^4$ in~\cref{fig:pe_dec_D4_punc} is that we replace ``If $c=x_1$, outputs $\bot$.'' with ``If $\beta= \beta_1$, outputs $\bot$.'' since the first line of $D_{\ne \alpha_1\concat m^\ast}^{\sfF}$ is never triggered. In these games, $\beta_1 \chosen \zo{9\ell}$ is not in the image of $F$ except with negligible probability. Recall that $c=x_1$ means $c = \alpha_1\concat \beta_1\concat \gamma_1$. Thus, those two circuits may differ when $\beta= \beta_1$ but $(\alpha,\gamma) \ne (\alpha_ 1,\gamma_1)$. However, it does not happen $\beta = F^\prime(\alpha \concat (G(\beta)\xor \gamma))$ in this case due to the injectivity of $F$. Thus, $D_{\ne \alpha_1\concat m^\ast}^{\sfF}$ and $D_{\ne \alpha_1\concat m^\ast,\ne \beta_1}^4$ are functionally equivalent and we obtain the proposition by applying the IO security.
\end{proof}
\begin{proposition}\label{prop:pe_sub_rand_onesix_two}
If $\iO$ is a secure IO, it holds that $\abs{\Pr[\rhybij{1}{6}=1] - \Pr[\sfrand_2 =1]} \le \negl(\secp)$.
\end{proposition}
\begin{proof}[Proof of~\cref{prop:pe_sub_rand_onesix_two}]
The difference between these two games that we use $D_{\ne \alpha_1\concat m^\ast,\ne \beta_1}^4[F_{\ne \alpha_1\concat m^\ast},G_{\ne \beta_1}]$ instead of $D_{\ne \alpha_1\concat m^\ast,\ne \beta_1}^4[F_{\ne \alpha_1\concat m^\ast},G]$. However, $G_{\ne \beta_1}(\beta_1)$ is never computed by the first line of $D_{\ne \alpha_1\concat m^\ast,\ne \beta_1}^4$. We obtain the proposition by the IO security.
\end{proof}
We complete the proof of~\cref{lem:pe_rand_one_two}.
\end{proof}

\subsection{Original Ciphertext Pseudorandomness of PE}\label{sec:original_pe_ct_pseudorandomness}
We describe the original ciphertext pseudorandomness of PE defined by Cohen et al.~\cite{SIAMCOMP:CHNVW18} in this section for reference.

\begin{definition}[Ciphertext Pseudorandomness]\label{def:original_ct_pseudorandomness}
We define the following experiment $\expt{\cA}{cpr}(\secp)$ for PE.
\begin{enumerate}
\item $\cA$ sends a message $m^\ast \in \zo{\ptxtlen}$ to the challenger.
\item The challenger does the following:

\begin{itemize}
\item Generate $(\ek,\dk) \lrun \Gen(1^\secp)$

\item Compute encryption $c^\ast \lrun \Enc(\ek, m^\ast)$. 

\item Choose $r^\ast \chosen \zo{\ctlen}$. 

\item Generate the punctured key $\dk_{\notin \setbk{c^\ast,r^\ast}} \lrun \Puncture(\dk, \{c^\ast,r^\ast\})$

\item Choose $\coin \chosen \zo{}$ and sends the following to $\cA$:
\begin{align}
(c^\ast, r^\ast ,\ek, \dk_{\notin \setbk{c^\ast,r^\ast}}) & \text{ if } \coin=0 \\
(r^\ast,c^\ast, \ek, \dk_{\notin \setbk{c^\ast,r^\ast}}) & \text{ if } \coin=1 
\end{align}
\end{itemize}
\item $\cA$ outputs $\coin^\ast$ and the experiment outputs $1$ if $\coin = \coin^\ast$; otherwise $0$.
\end{enumerate}
We say that $\PE$ has ciphertext pseudorandomness if for every QPT adversary $\cA$, it holds that
\[
\adva{\cA}{cpr}(\secp)\seteq 2\cdot \Pr[\expt{\cA}{cpr}(\secp) \out 1] -1 \leq \negl(\secp).
\]
\end{definition}

\paragraph{Issue in the proof by Cohen et al.}
In the watermarking PRF by Cohen et al.~\cite{SIAMCOMP:CHNVW18}, we use $x_0 \lrun \PE.\Enc(\peek,a\concat b\concat c \concat i)$ to extract an embedded message. They replace $x_0 \lrun \PE.\Enc(\peek,a\concat b\concat c \concat i)$ with $x_1 \chosen \zo{\ctlen}$ in their proof of unremovability~\cite[Lemma 6.7]{SIAMCOMP:CHNVW18}. Then, they use PRG security~\cite[Lemma 6.8]{SIAMCOMP:CHNVW18} to replace $\PRG(c)$ with a uniformly random string since the information about $c$ disappears from the PE ciphertext. However, there is a subtle issue here.
The information about $c$ remains in the punctured decryption key $\dk_{\notin \setbk{x_0,x_1}} \lrun \Puncture(\pedk,\setbk{x_0,x_1})$, which is punctured both at $x_0$ and $x_1$, since they use ciphertext pseudorandomness in~\cref{def:original_ct_pseudorandomness} and need to use the punctured decryption key. Thus, we cannot apply PRG security even after we apply the ciphertext pseudorandomness in~\cref{def:original_ct_pseudorandomness}. This is the reason why we introduce the strong ciphertext pseudorandomness in~\cref{def:pe_strong_pseudorandomness}.